\documentclass[12pt, letterpaper]{article}

\usepackage[utf8]{inputenc}
\usepackage[driver=dvipdfm,top = 1.32in,left=1.25in,right=1.25in,bottom=1.32in]{geometry}
\usepackage{mathtools}
\usepackage{amsmath}
\usepackage{amssymb}
\usepackage{setspace}
\usepackage{booktabs}
\usepackage{tabularx}

\mathtoolsset{showonlyrefs}

\makeatletter

\usepackage{amsthm}
\usepackage{amsfonts}
\usepackage{bbm}
\usepackage{graphics}
\usepackage{enumerate}
\usepackage{float}
\usepackage{caption}
\usepackage{subcaption}
\usepackage[compress]{natbib}
\usepackage{enumerate}
\usepackage{color}
\usepackage{tabularx}

\theoremstyle{remark}
\newtheorem{remark}{Remark}

\theoremstyle{definition}
\newtheorem{thm}{Theorem}
\newtheorem{lem}[thm]{Lemma}
\newtheorem{definition}{Definition}

\newtheorem{assp}{Assumption}

\usepackage{amsmath} 
\usepackage{amssymb}
\usepackage{bm}
\usepackage{ascmac}
\usepackage{setspace}
\usepackage[at]{easylist}

\usepackage[hyphens]{url}
\usepackage[colorlinks,urlcolor=blue, citecolor=blue, menucolor=blue]{hyperref}

\providecommand{\keywords}[1]
{
  \small	
  \textbf{Keywords:} #1
}

\usepackage{algorithm,algpseudocode}

\usepackage{amsmath}
\DeclareMathOperator*{\argmax}{arg\,max}

\newcommand{\Real}{\mathbb{R}}
\newcommand{\Natural}{\mathbb{N}}
\newcommand{\Regret}{\mathrm{Reg}}

\newcommand{\URegret}{\text{U2S-Reg}}
\newcommand{\CRegret}{\text{C2S-Reg}}

\newcommand{\regret}{\mathrm{reg}} %
\newcommand{\Subsid}{\mathrm{Sub}}
\newcommand{\Var}{\mathrm{Var}}
\newcommand{\Ex}{\mathbb{E}}
\newcommand{\Prob}{\mathbb{P}}
\newcommand{\eps}{\epsilon}
\newcommand{\del}{\delta}

\newcommand{\Normal}{\mathcal{N}}
\newcommand{\lambdamin}{\lambda_{\mathrm{min}}}

\newcommand{\Ind}{\textbf{1}}

\newcommand{\tilO}{\tilde{O}}
\newcommand{\tilOmega}{\tilde{\Omega}}
\newcommand{\tilTheta}{\tilde{\Theta}}
\newcommand{\sui}{s^{\text{U-I}}}

\newcommand{\tiliota}{\tilde{\iota}}
\newcommand{\sgn}{\mathrm{sgn}}
\newcommand{\nn}{\nonumber\\}
\newcommand{\Nzero}{N^{(0)}}
\newcommand{\ahyb}{a} %
\newcommand{\ahybconst}{\alpha} %

\newcommand{\Uimp}{\tilde{q}^{\mathrm{H}}}

\newcommand{\shi}{s^{\text{H-I}}}

\newcommand{\ih}{\iota^{\text{H}}}

\newcommand{\norm}[1]{\left\lVert#1\right\rVert} %

\newcommand{\EA}{\mathcal{A}}
\newcommand{\EB}{\mathcal{B}}
\newcommand{\EC}{\mathcal{C}}

\newcommand{\EJ}{\mathcal{J}}

\newcommand{\EM}{\mathcal{M}}

\newcommand{\EP}{\mathcal{P}}
\newcommand{\EQ}{\mathcal{Q}}
\newcommand{\ER}{\mathcal{R}}
\newcommand{\ES}{\mathcal{S}}

\newcommand{\EU}{\mathcal{U}}
\newcommand{\EV}{\mathcal{V}}
\newcommand{\EW}{\mathcal{W}}
\newcommand{\EX}{\mathcal{X}}

\newcommand{\bx}{\bm{x}}

\newcommand{\bbe}{\bm{e}} %
\newcommand{\bv}{\bm{v}}

\newcommand{\bmu}{\bm{\mu}}
\newcommand{\btheta}{\bm{\theta}}
\newcommand{\hbtheta}{\hat{\bm{\theta}}}
\newcommand{\tbtheta}{\tilde{\bm{\theta}}}
\newcommand{\bbtheta}{\bar{\bm{\theta}}}
\newcommand{\htheta}{\hat{\theta}}
\newcommand{\bI}{\bm{I}}
\newcommand{\bA}{\bm{A}}
\newcommand{\bB}{\bm{B}}

\newcommand{\bV}{\bm{V}}

\newcommand{\bX}{\bm{X}}

\newcommand{\bZero}{\bm{0}}

\newcommand{\Simd}{5}
\newcommand{\Simtheta}{(1,1,1,1,1)}
\newcommand{\Simmux}{1.5}
\newcommand{\Simsigmax}{1}
\newcommand{\Simsigmaeps}{0.5}
\newcommand{\Simlambda}{1}
\newcommand{\Simdelta}{0.1}
\newcommand{\Sima}{0.5}
\newcommand{\Simrooneystop}{50} %
\newcommand{\Numrun}{4,000}

\newcommand{\Simsigmaeta}{6} %
\newcommand{\LCont}[1]{L_{#1}}
\newcommand{\LContRaw}{L} %

\allowdisplaybreaks

\makeatother

\begin{document}

\title{On Statistical Discrimination as a Failure of Social Learning: A Multi-Armed Bandit Approach
}

\author{Junpei Komiyama \and Shunya Noda\thanks{Komiyama: Leonard N. Stern School of Business, New York University. E-mail: \href{mailto:junpei.komiyama@gmail.com}{junpei.komiyama@gmail.com}. Noda: Graduate School of Economics, University of Tokyo and Vancouver School of Economics, University of British Columbia. E-mail: \href{mailto:shunya.noda@gmail.com}{shunya.noda@gmail.com}. Noda gratefully acknowledges the financial support from the Social Sciences and Humanities Research Council of Canada. We are grateful to Itai Ashlagi, Tomohiro Hara, Yoko Okuyama, Masayuki Yagasaki, and the seminar participants at Happy Hour Seminar, the University of British Columbia, Tokyo Keizai University, the University of Tokyo, the AFCI Workshop in NeurIPS2020, the AI4SG Workshop in IJCAI 2020, the University of Texas at Austin, CyberAgent, Inc., WEAI International Conference 2021, JEA Spring Meeting 2021, CEA Annual Meeting 2021, and NASMES 2021 for their helpful comments. All remaining errors are our own.}}

\date{First Draft: October 2, 2020\hspace{2em} Current Version: \today}

\maketitle

\begin{abstract}
    We analyze statistical discrimination in hiring markets using a multi-armed bandit model. Myopic firms face workers arriving with heterogeneous observable characteristics. The association between the worker's skill and characteristics is unknown ex ante; thus, firms need to learn it. Laissez-faire causes perpetual underestimation: minority workers are rarely hired, and therefore, the underestimation tends to persist. Even a marginal imbalance in the population ratio frequently results in perpetual underestimation. We propose two policy solutions: a novel subsidy rule (the hybrid mechanism) and the Rooney Rule. Our results indicate that temporary affirmative actions effectively alleviate discrimination stemming from insufficient data.
\end{abstract}


\vspace{3em}

\pagebreak

\section{Introduction}\label{sec: introduction}

\emph{Statistical discrimination} refers to discrimination against minority people, taken by fully rational and non-prejudiced agents. 
Previous studies have shown that, even in the absence of prejudice, discrimination can occur persistently because of various reasons, including the discouragement of human capital investment \citep{arrow1973theory,Foster1992AnEA,coate1993will,moro_general_2004}, information friction \citep{phelps1972statistical,cornell1996culture,Bardhi2019Spiraling}, and search friction \citep{mailath2000endogenous,Che2019RatingsGuided}. The literature has proposed various affirmative-action policies to solve statistical discrimination, with many having been implemented in practice.

This paper demonstrates that statistical discrimination may appear as a failure of social learning. We endogenize the evolution of biased beliefs and analyze their consequences. Our model assumes that (i) all firms (decision-makers) are fully rational and non-prejudiced (i.e., attempt to hire the most productive worker), and (ii) all workers are ex ante symmetric. In such an environment, an unbiased decision policy---hiring workers with superior skills---satisfies numerous fairness notions articulated in scholarly literature, including equalized odds and demographic parity. It also achieves efficiency by maximizing each firm's payoff. However, the long-term persistence of biased beliefs could still occur. This paper underscores that \emph{temporary} affirmative actions can effectively enhance both welfare and equality.

Although our model applies more broadly, we use the terminology of hiring markets to describe our model.
We develop a \emph{multi-armed bandit} model of social learning, in which many myopic and short-lived firms sequentially make hiring decisions. In each round, a firm hires one worker from a set of candidates. Each firm's utility is determined by the hired worker's \emph{skill}, which cannot be observed directly until employment. However, as in the standard statistical discrimination model, each worker also has \emph{observable characteristics} associated with their unobservable skills. Firms learn the statistical association between characteristics and skills using data pertaining to past hiring cases (shared through, e.g., private communication, social media, and recommendation letters) and use the estimators to predict the skills of candidates.

Each worker belongs to a \emph{group} that represents, for example, their gender, race, and ethnicity. We assume that the characteristics of workers who belong to different groups should be interpreted differently. This assumption is realistic. First, previous studies have revealed that underrepresented groups receive unfairly low evaluations.\footnote{For instance, \cite{trix2003exploring} analyze letters of recommendation for medical faculty, finding systematic differences between those written for female and male applicants. \cite{hanna2012discrimination} postulate that students belonging to lower castes in India tend to receive unjustifiably lower exam scores. In the context of teaching evaluations, \cite{macnell2015s} and \cite{mitchell2018gender} illustrate that students rate male identities significantly higher than female ones. In a study of online freelance marketplaces, \cite{hannak2017bias} establish that gender and race significantly correlate with worker evaluations.}
When these evaluations are used as the observable characteristics, firms should be aware of the potential bias. 
Second, evaluations may reflect differences in cultures, living environments and social systems \citep{precht1998cross,al2004get}. For instance, firms need to be conversant with the norms of drafting recommendation letters to interpret them accurately. Therefore, observable characteristics, such as curriculum vitae, exam scores, grading reports, recommendation letters, and so forth, might convey starkly different implications despite their similar presentations. If firms are unbiased and cognizant of these potential biases, they should adapt their interpretation methods for these characteristics, applying varied statistical models to different groups.

When firms learn the statistical association from data, with some probability, the minority group is underestimated because of a large estimation error raised by insufficient data.
Once the minority group is underestimated, it is difficult for a minority worker to appear to be the best candidate---even if he has the greatest skill among the candidates, the firm often dismisses this fact and tends to hire a majority worker.
As long as firms only hire majority workers, society cannot learn about the minority group; thus, the imbalance persists even in the long run. We call this phenomenon \emph{perpetual underestimation}.

We use a linear contextual bandit model to analyze the consequence of social learning. To gauge policy performance, we utilize \emph{regret}, a widely adopted measure in machine-learning literature that assesses welfare loss relative to the optimal decision rule. Regret arises if proficient minority workers are overlooked due to biased estimates by firms; hence, regret not only signifies efficiency but also encapsulates fairness.\footnote{In Appendix~\ref{sec: fairness criteria}, we formally prove that a decision rule has sublinear regret only if it aligns with equalized odds. This notion stipulates that society's hiring policy performs equitably across groups. Moreover, with symmetric groups, sublinear regret also harmonizes with demographic parity, which ensures hiring decisions are irrespective of membership in a minority group.}

We focus on how regret grows as the total number of firms (denoted by $N$) increases. When regret is sublinear in $N$, firms make fair and efficient decisions in the long run.
We first analyze the equilibrium consequence of \emph{laissez-faire} (no policy intervention). When the groups are ex ante symmetric and the population ratio is equal, laissez-faire results in $\tilO(\sqrt{N})$ regret.\footnote{$\tilO, \tilOmega$, and $\tilTheta$ are a Landau notations that ignore polylogarithmic factors. We often treat polylogarithmic factors as if they were constant because these factors grow very slowly ($o(N^\epsilon)$ for any exponent $\epsilon > 0$).} However, when the population ratio is unbalanced, this no longer holds, and expected regret is linear: $\tilOmega(N)$.

We study two policy interventions toward fair and efficient social learning. The first policy is a subsidy rule, based on the idea of \emph{upper confidence bound} (UCB). UCB is an effective solution for balancing exploration and exploitation \citep{lai1985,auer2002}. By incentivizing firms to take actions that are consistent with the recommendations of the UCB, social learning can promote sublinear regret in the long run. 
The subsidy is adjusted to the degree of information externality. We demonstrate that the UCB mechanism has the expected regret of $\tilO(\sqrt{N})$. The subsidy required to implement the UCB mechanism is also $\tilO(\sqrt{N})$.

Improving the UCB mechanism, this paper proposes a \emph{hybrid} mechanism, which lifts affirmative actions upon the collection of a sufficiently rich data set. The hybrid mechanism takes advantage of spontaneous exploration: Once firms obtain a certain amount of data, the diversity of workers' characteristics naturally promotes learning about the minority group. The hybrid mechanism achieves $\tilO(\sqrt{N})$ regret with $\tilO(1)$ subsidy.

The second policy is the Rooney Rule, which requires each firm to interview at least one minority candidate as a finalist for each job opening. We analyze the effect of the Rooney Rule using a two-stage model in which firms observe additional signals of each finalist. 
The Rooney Rule enables minority workers to reveal the additional signal to the firm, which leaves a chance of breaking down the underestimation. However, our assessment of the Rooney Rule is mixed. The imposed interviewing quota could unjustly deprive skilled majority workers of employment opportunities, suggesting reverse discrimination. This drawback is lessened if the Rooney Rule is implemented temporarily.

This paper is framed as a positive analysis elucidating how discrimination arises from social learning conducted by small, rational, and unbiased firms. Alternatively, our study could be viewed as a normative analysis showcasing an efficient hiring policy targeting the long-term average skill of workers hired by a large firm \citep[as explored by][]{Bergman2020}. For this latter scenario, our results for the hybrid mechanism indicate that a firm can cease affirmative action once it has accumulated reasonably comprehensive information about minority workers.

\section{Related Literature}\label{sec: related literature}

\paragraph{Statistical Discrimination} 

Various studies have analyzed statistical discrimination both theoretically \citep{phelps1972statistical,arrow1973theory,Foster1992AnEA,coate1993will,cornell1996culture,mailath2000endogenous} and experimentally \citep[][is an excellent survey]{neumark2018}. 
We contribute to this literature by articulating a new channel of discrimination: endogenous data imbalance and insufficiency. Similar to previous studies, we assume otherwise ex ante identical individuals from different groups to demonstrate how discrimination evolves and persists. Meanwhile, our results provide further indication that demographic minorities suffer from discrimination as an inevitable consequence of laissez-faire.

\cite{Hu2018Short} examines a dynamic reputation model in a labor market, where workers can endogenously select their skill level. As highlighted by \cite{Foster1992AnEA} and \cite{coate1993will}, statistical discrimination can potentially discourage minorities from enhancing their skills. Implementing a fairness constraint through affirmative action at the entry-level may rectify inequalities within the entire labor market. Our study adds to this body of literature by demonstrating that short-term affirmative action successfully tackles inefficiency and inequality, even when the skill level is fixed.

\cite{Kannan2019Downstream} study how a college can design an admission and grading policy to achieve fair employment, assuming employers form a Bayesian belief about students' skills based on the information provided by the college. We also consider a government that introduces an affirmative-action policy taking into account stakeholders' (firms') endogenous response. \cite{Che2019RatingsGuided} examine a rating-guided market, demonstrating that feedback loops can cause discriminatory inferences concerning social groups. We identify endogenously created informational disparities due to feedback loops within a distinct model inspired by a hiring market, deliberate on the underlying causes (demographic imbalance) that instigate discrimination, and propose policy solutions.

\cite{Bohren2019inaccurate,bohren2019dynamics} and \cite{monachou2019discrimination} have demonstrated how misspecified beliefs about groups generate discrimination. Thus far, this literature has attributed belief misspecification to psychological biases and bounded rationality. In contrast, we demonstrate that misspecified beliefs may evolve and persist endogenously, even in the long run. Through a laboratory experiment, \cite{dianat2020lab} reveal that affirmative action's impact becomes fleeting if the measure is discontinued before beliefs undergo transformation. Our hybrid mechanism offers a resolution by optimally choosing the timing to terminate the program, thereby preventing the persistence of underestimation.

\paragraph{Social Learning}
The economics literature has extensively studied herding, information cascade, and social learning \citep[e.g.,][]{bikhchandani1992theory,banerjee1992simple,smith2000pathological}. Additionally, various papers have studied improvements to social welfare through subsidy for exploration  \citep[e.g.,][]{frazier2014incentivizing,kannan2017fairness} and selective information disclosure \citep[e.g.,][]{kremer2014implementing,papanastasiou2018crowdsourcing,Immorlica2020incentivizing,mansour2020bayesian}. We propose novel policy interventions to improve social learning in fairness and efficiency.

\paragraph{Multi-armed Bandit} 
A multi-armed bandit problem stems from the literature of statistics \citep{thompson1933,robbins1952}. This problem is driven by the question of how a single long-lived decision-maker can maximize his payoff by balancing exploration and exploitation. More recently, the machine-learning community has proposed the contextual bandit framework, in which payoffs associated with ``arms'' (actions) depend not only on the hidden state but also on additional information, referred to as ``contexts'' \citep{abe1999,langford2007}. 
We adopt the contextual bandit framework because context enables us to capture the diversity of worker characteristics.\footnote{The trade-off between exploration and exploitation presents itself in a wider context. For example, \cite{Owen2020Tie} propose ``tie-breaker designs'' which are hybrids of randomized controlled trials and regression discontinuity designs, and solve the optimal tradeoff between information gain (exploration) and efficiency in the treatment allocation (exploitation).}

Several previous studies have considered a linear contextual bandit problem and studied the performance of a ``greedy'' algorithm, which makes decisions myopically in accordance with the current information. Because firms take greedy actions under laissez-faire, their results are also relevant to our model. \citet{Bastani17} and \citet{kannan2018} have shown that a greedy algorithm leads to sublinear regret in the long run, if the contexts are diverse enough.\footnote{Our simulation, included in Appendix~\ref{subsec: hybrid vs uniform sampling}, shows that our hybrid mechanism can be interpreted as an efficient approach to collecting initial samples.} We characterize the relationship between the diversity of contexts and the rate of learning. Moreover, we show that the population ratio is crucial to the regret rate (Section~\ref{sec: LF unbalanced group}). As an efficient intervention, \citet{kannan2017fairness} consider a contextually fair UCB-based subsidy rule. Although our subsidy policy also originates from the idea of UCB (Section~\ref{sec: subsidy scheme}), we establish a novel mechanism (the hybrid mechanism, Section~\ref{sec: Hybrid Decision Rule}) that reduces budget expenditure by utilizing spontaneous exploration.

The multi-armed bandit approach has recently found applications in labor market analyses. \cite{Bardhi2019Spiraling} demonstrate that a minor difference in initial beliefs about each worker's type can ultimately yield a substantial disparity in workers' payoffs. \cite{Johari2018Exploration} examine how a labor platform can discern workers' skills to attain an optimal worker assignment when the platform can only observe the outcomes generated by teams, not individual workers.

\cite{Bergman2020} portray a large firm's hiring process as a multi-armed bandit problem and empirically compare the performance of the status quo (screening via manual work), a greedy policy, and a UCB method. They reveal that a UCB method not only screens job applicants efficiently but also preserves diversity. Their findings suggest that a UCB method is both fairer and more efficient when implemented by a large firm. Interpreting this paper as a study of an efficient hiring policy by a large firm, our results enhance \cite{Bergman2020} by providing theoretical foundations that outline the performances of a greedy policy (corresponding to laissez-faire) and a UCB method. Moreover, we illustrate that affirmative action can be discontinued shortly by characterizing the performance of a hybrid mechanism.

\paragraph{Algorithmic Fairness}
The literature on algorithmic fairness is growing. 
This literature has implicitly assumed exogenous asymmetry in worker skills and pursued the approaches to correct between-group inequality. To this end, ``discrimination-aware'' constraints such as equalized odds \citep{HardtPNS16} and demographic parity \citep{PedreschiRT08,CaldersV10} have been proposed, with several papers applying these constraints in the context of multi-armed bandit problems \citep{joseph_fairness_2016} or more general sequential learning \citep{RaghavanSVW18,BechavodL0WW19,ChenAAMAS2020}. While these fairness goals are conflicting in general, we analyze an environment in which many fairness goals are aligned and demonstrate how affirmative action improves them.

\paragraph{Rooney Rule}
The Rooney Rule was originally introduced in the context of the hiring of National Football League senior staff \citep{Eddo-Lodge2017Rooney}. While it is widely used in practice, theoretical analyses of the Rooney Rule are scarce. \citet{DBLP:conf/innovations/KleinbergR18} show that, when a recruiter is unconsciously biased against a group, the Rooney Rule not only improves the representation of that group but also leads to a higher payoff for the recruiter. To the best of our knowledge, this study (Section \ref{sec: Rooney Rule}) constitutes the first attempt to demonstrate the advantage of the Rooney Rule by modeling unbiased agents.

\section{Model}\label{sec: model}

\paragraph{Basic Setting}
We develop a linear contextual bandit problem with myopic agents (firms). We consider a situation where $N$ firms (indexed by $n = 1, \ldots, N$) sequentially hire one worker for each.\footnote{While real-world firms are long-lived and hire multiple workers, the number of workers hired by one firm is typically much smaller than the total number of workers hired in a hiring market. Accordingly, even if we allowed firms to hire multiple (but a small number of) workers, the conclusion would not change qualitatively. Note also that various seminal papers within the social learning literature \citep[such as][]{banerjee1992simple,bikhchandani1992theory,smith2000pathological} have made the same assumption.} In each round $n$, a set of workers $I(n)$ (i.e., arms) arrives. Each worker $i \in I(n)$ takes no action, and firm $n$ hires only one worker $\iota(n) \in I(n)$. Both firms and workers are short-lived. Upon round $n$ ending, firm $n$'s payoff is finalized, and all rejected workers leave the market.\footnote{This assumption is for the sake of simplicity. Since firms have no private information, the fact that a worker was previously rejected by another firm does not influence the worker's evaluation (given that the current firm can also observe the worker's characteristics); thus, entrant workers and incumbent workers have no informational difference. Accordingly, even if workers stay in the hiring market for multiple periods, our conclusion will not be changed qualitatively.}

Each worker $i \in I$ belongs to a group $g \in G$. We assume that the population ratio is fixed: for every round $n$, the number of workers belonging to group $g$ is $K_g \in \Natural$ and $K = \sum_{g \in G} K_g$. Slightly abusing the notation, we denote the group worker $i$ belongs to by $g(i)$. Each worker $i$ also has observable characteristics $\bx_i \in \Real^d$, with $d \in \Natural$ as their dimension. Finally, each worker $i$ also has a skill $y_i \in \Real$ that is not observable until worker $i$ is hired. The characteristics and skills are random variables. 

Because each firm's payoff is equal to the hired worker's skill $y_i$ (plus the subsidy assigned to worker $i$ as an affirmative action, if any), firms want to predict the skill $y_i$ based on the characteristics $\bx_i$. We assume that characteristics and skills are associated as $y_i = \bx_i' \btheta_{g(i)} + \eps_i$,
where $\btheta_g \in \Real^d$ is a \emph{coefficient parameter}, and $\eps_i \sim \Normal(0, \sigma^2_\eps)$ i.i.d.\ is an unpredictable error term. 
We assume $||\btheta_g|| \le S$ for some $S \in \Real_{+}$, where $||\cdot||$ is the standard L2-norm.
Since $\eps_i$ is unpredictable, $q_i \coloneqq \bx_i' \btheta_{g(i)}$ is the best predictor of worker $i$'s skill $y_i$.

The coefficient parameters $(\btheta_{g})_{g\in G}$ are initially unknown. Hence, unless firms share information about past hires, firms are unable to predict each worker's skill $y_i$. We assume that firms share information about past hiring cases.\footnote{Alternatively, we can assume that firms only share information about a certain fraction of workers. We expect that, under this assumption, (i) the results would not change qualitatively, and (ii) the statistical discrimination would become severer because it becomes more difficult to accumulate information about the minority group.} 
Accordingly, when firm $n$ makes a decision, in addition to the characteristics and groups of current workers $(\bx_i, g(i))_{i \in I(n)}$, firm $n$ observes the characteristics, groups, and skills of previously hired workers $(x_{\iota(n')}, g(\iota(n')), y_{\iota(n')})_{n' = 1}^{n - 1}$.
We refer to all realizations of these variables as the \emph{history} in round $n$, and denote it by $h(n)$.
Formally, $h(n)$ is given by
 \begin{equation}
     h(n) = \left((\bx_i, g(i))_{i \in I(n)}, (x_{\iota(n')}, g(\iota(n')), y_{\iota(n')})_{n' = 1}^{n - 1} \right).
\end{equation}
Note that, $h(n)$ does not include information about (i) the worker hired by firm $n$, or (ii) that worker's actual skill. This is because the notation $h(n)$ represents the information set firm $n$ faces when it makes a hiring decision. 
We denote the set of all the possible histories in round $n$ by $H(n)$. 
The firm's decision rule for hiring and the government's subsidy rule are defined as a function that maps a history to a hiring decision and the subsidy amount (described later). For notational convenience, we often omit $h(n)$.

\paragraph{Prediction}
We assume that firms are not Bayesian but \emph{frequentists}. Hence, firms do not have a prior belief about the parameter $\btheta$ but estimate it only using the available data set. We expect that essentially the same results will be obtained with Bayesian firms (see Appendix~\ref{sec: Bayesian Approach}).

We assume that each firm predicts skill using \emph{ridge regression} (L2-regularized least square).\footnote{For the properties of the ridge estimator, see \citet{Kennedy2008}, for example.} Let $N_g(n)$ be the number of rounds at which group-$g$ workers are hired before round $n$. Let $\bX_g(n) \in \Real^{N_g(n)\times d}$ be a matrix that lists the characteristics of group-$g$ workers hired by round $n$: each row of $\bX_g(n)$ corresponds to $\{\bx_{\iota(n')}: \iota(n') = g\}_{n' = 1}^{n - 1}$. Likewise, let $Y_g(n) \in \Real^{N_g(n)}$ be a vector that lists the skills of group-$g$ workers hired by round $n$: each element of $Y_g(n)$ corresponds to $\{y_{\iota(n')}: \iota(n') = g\}_{n'=1}^{n-1}$. We define $\bV_g(n) \coloneqq (\bX_g(n))'\bX_g(n)$. For a parameter $\lambda > 0$, we define $\bar{\bV}_g(n) = \bV_g(n) + \lambda \bI_d$, where $\bI_d$ denotes the $d\times d$ identity matrix. Firm $n$ estimates the parameter as follows:
\begin{equation}\label{eq: def theta hat}
    \hat{\btheta}_g(n) \coloneqq (\bar{\bV}_g(n))^{-1}(\bX_g(n))'Y_g(n).
\end{equation}
Firm $n$ predicts worker $i$'s skill $q_i$, while substituting the true predicted skill $\btheta_g$ with \emph{estimated skill} $\hat{\btheta}_g(n)$: $\hat{q}_i(n) \coloneqq \bx_i' \hat{\btheta}_{g(i)}(n)$.
Hence, $\hat{q}_i(n)$ and $\hbtheta_{g}(n)$ depend on the history $h(n)$. 
The \emph{ordinary least squares} (OLS) estimator corresponds to the ridge estimator with $\lambda = 0$.
We use the ridge estimator instead of the OLS estimator to stabilize the small-sample inference.
For example, for some history, $\bV_g(n)$ may not have full rank, and the OLS estimator may not be well-defined. Even for such histories, the ridge estimator is always well-defined.

For analytical tractability, we assume that for the first $\Nzero$ rounds, each firm $n$ must hire from a pre-specified group, $g_n$. We refer to the first $\Nzero$ rounds as the \emph{initial sampling phase}. We assume $\Nzero$ to be small and deal $\Nzero$ as a constant.\footnote{The required size of $\Nzero$ is specified by Eq.~\eqref{ineq_suffnzero} in Appendix.} Let $\Nzero_g \coloneqq \sum_{n=1}^{\Nzero}\Ind[g_n = g]$ as the data size of initial sampling for group $g$, where $\Ind[\EA] = 1$ if event $\EA$ holds or $0$ otherwise. The initial sampling phase is exogenous. That is, we ignore the incentives and payoffs of firms and assume that the characteristics $\bx$ of the hired candidate constitute an i.i.d.\ sample of the corresponding group. We analyze mechanism, social welfare, and budget after round $n > \Nzero$. The initial sampling phase can be interpreted as data that has already been produced in history. The welfare cost is already sunk, and the government can no longer make policy interventions for the event that has already occurred in the past.

\paragraph{Mechanism}

In addition to worker skills, firms are also concerned about subsidies. We assume that firm preferences are risk-neutral and quasi-linear. Hence, if firm $n$ hires worker $i$, its payoff (von-Neumann--Morgenstern utility) is given by $y_i + s_i$, where $s_i \in \Real_+$ denotes the amount of the subsidy assigned to worker $i$.

In the beginning of the game, the government commits to a \emph{subsidy rule} $s_i(n,\cdot): H(n) \to \Real_+$, which maps a history to a subsidy amount. Hence, once a history $h(n)$ is specified, firm $n$ can identify the subsidy assigned to each worker $i\in I(n)$. Firm $n$ attempts to maximize
\begin{equation}
    \mathbb{E}\left[\left. y_i + s_i(n;h(n)) \right| h(n)\right] = \hat{q}_i(n;h(n)) + s_i(n;h(n)).
\end{equation}

Firm $n$'s \emph{decision rule} $\iota(n, \cdot): H(n) \to I(n)$ specifies the worker that firm $n$ hires after history $h(n)$. We say that, a decision rule $\iota$ is \emph{implemented} by a subsidy rule $s_i$ if for all $n$ and $h(n)$, we have
\begin{equation}\label{eq: the candidate firm n hires}
    \iota(n;h(n)) = \argmax_{i\in I(n)} \left\{\hat{q}_i(n;h(n)) + s_i(n;h(n))\right\}.
\end{equation}
Throughout this paper, any ties are broken arbitrarily. We call a pair of a decision rule and subsidy rule a \emph{mechanism}. We often drop $h(n)$ from the input of decision rule $\iota$ when it does not cause confusion.

\paragraph{Regret}

\emph{Regret} is a standard measure for evaluating the performance of algorithms in multi-armed bandit models:
\begin{equation}\label{eq: regret definition}
    \Regret(N) \coloneqq \sum_{n = \Nzero+1}^N \left\{\max_{i \in I(n)}q_i - q_{\iota(n)}\right\}.
\end{equation}
Since $\eps_i$ is unpredictable, it is natural to evaluate the performance of the algorithm (or the equilibrium consequence of the policy intervention) by comparing it with $q_i$. If the parameter $(\btheta_g)_{g\in G}$ were known, each firm could easily calculate $q_i$ for each worker $i$ and hire the best worker, $i^*(n) \coloneqq \argmax_{i\in I(n)}q_i$. In this case, regret would be zero. The goal of the policy design is to establish a mechanism that minimizes the expected regret $\Ex[\Regret(N)]$, where the expectation is taken on a random draw of workers. This aim is equivalent to maximizing the sum of the skill of workers hired. 

Following the literature, we often evaluate the performance by the limiting behavior (order) of expected regrets. A decision rule $\iota$ is said to have \emph{sublinear regret} if $\Ex[\Regret(N)]=O(N^a)$ for some $a < 1$. Small regret implies not only efficiency but also fairness. Regret measures the disparate impact that is not justified by skill disparity, and sublinear regret is achieved if firms hire the most skillful workers without regard to the group of workers. In Appendix~\ref{sec: fairness criteria}, we demonstrate that a sublinear-regret decision rule asymptotically aligns with a fairness notion called \emph{equalized odds}, which requires that candidate workers in the majority and minority groups have an equal true positive rate (hired when they have the highest skill predictor $q_i$) and equal false negative rate (not hired when they have the highest $q_i$).

\paragraph{Budget}

Some of the policies we study incentivize exploration through subsidies. The total budget required by a subsidy rule is also an important policy concern. The total amount of the subsidy is given by
$\Subsid(N) \coloneqq \sum_{n = \Nzero+1}^N s_{\iota(n)}(n)$.

\section{Laissez-Faire}\label{sec: laissez-faire}

This section analyzes the equilibrium under \emph{laissez-faire}; that is, the consequence of social learning in the absence of policy intervention.

\begin{definition}[Laissez-Faire]
The \emph{laissez-faire decision rule} always selects the worker who has the greatest estimated skill, i.e., $\iota(n)  = \argmax_{i\in I(n)}\hat{q}_i(n)$.
This decision rule is implemented by the \emph{laissez-faire subsidy rule}, which provides no subsidy $s_i(n) = 0$ after any history.
\end{definition}

Laissez-faire makes no intervention. Each firm hires the worker with the greatest estimated skill, as predicted by the current data set. The multi-armed bandit literature refers to the laissez-faire decision rule as the \emph{greedy algorithm}.

\subsection{Symmetry and Diverse Characteristics}\label{subsec: symmetry diverse}

To illustrate a failure of social learning, we make three assumptions. First, as a minimal environment to analyze discrimination, we focus on the two-group case.
\begin{assp}[Two Groups]\label{assp:twogroups}
The population comprises two groups $G = \{1,2\}$.
\end{assp}
When we consider asymmetric equilibria, we refer to group $1$ as the majority (dominant) group and group $2$ as the minority (discriminated-against) group. The two-group assumption enables the elucidation of how the minority group is discriminated against.

Second, we assume that groups are symmetric.
\begin{assp}[Symmetric Groups]\label{assp:idcontext}
The characteristics of all groups are identical, and the coefficient parameters are the same across the groups. That is, a probability distribution $F$ such that for all $i \in I$, $\bx_i \sim F$, and there exists $\btheta \in \Real^d$ such that, for all $g \in G$, $\btheta_g = \btheta$.
\end{assp}

Note that although we assume that groups are symmetric, firms do not know the true parameters, and therefore, apply different statistical models to different groups. That is, even though the \emph{true} coefficients are identical ($\btheta_g = \btheta_g'$ for all $g,g' \in G$), firms estimate them separately; thus, the values of the \emph{estimated} coefficients are typically different ($\hbtheta_g(n) \neq \hbtheta_{g'}(n)$ for $g\neq g'$).

Although Assumption~\ref{assp:idcontext} is unrealistic (because the characteristics should evidently be interpreted differently), it is useful for elucidating how laissez-faire nourishes statistical discrimination. Under Assumption~\ref{assp:idcontext}, agents are ex ante identical \citep[as assumed in][]{arrow1973theory,Foster1992AnEA,coate1993will,moro_general_2004}, and therefore the differences we observe in the equilibrium are entirely attributed to social learning.

Furthermore, when groups are symmetric, disparate impact is unambiguously unfair. It is well-known that popular fairness notions aim at different goals and are compatible with each other only in highly constrained special cases \citep[see, e.g.,][]{KleinbergMR17}. The symmetric environment specified by Assumption~\ref{assp:idcontext} one of such exceptions: In this environment, sublinear regret implies not only equalized odds but also \emph{demographic parity}, i.e., the probability of a worker to be hired is independent of his group (see Appendix~\ref{sec: fairness criteria}). Since this paper's focus is not to debate which of the various types of fairness notions should be respected, we will concentrate only on the symmetric environment.\footnote{We confirmed through simulations that the proposed mechanisms are effective in a broad class of asymmetric environments. See Appendices~\ref{sec: fairness criteria} and \ref{subsec: simulation on asymmetric models}.}

Third, we assume that characteristics are normally distributed, and therefore, the distribution is non-degenerate. This assumption captures the diversity of workers.

\begin{assp}[Normally Distributed Characteristics]\label{assp:normalcontext}
For every candidate $i$, 
\begin{equation}
    \bx_i \sim \Normal(\bmu_{xg(i)}, \sigma_{xg(i)}^2 \bI_d),
\end{equation}
where $\bmu_{xg} \in \Real^d$ and $\sigma_{xg} \in \Real_{++}$ for every $g \in G$. We also denote $\bx_i = \bmu_{xg(i)} + \bbe_{xi}$ to highlight the noise term $\bbe_{xi}$.
\end{assp}

We consider essentially the same results to hold more generally as long as the characteristics are sufficiently diverse.
Note that when we have both Assumptions~\ref{assp:idcontext} and \ref{assp:normalcontext}, then there exist $\bmu_x, \sigma_x$ such that $\bmu_{xg} = \bmu_x$ and $\sigma_{xg} = \sigma_x$ for all $g \in G$. Hence, $\bx_i \sim \Normal(\bmu_x,\sigma_x^2 \bI_d)$ for all $i$.

\subsection{Perpetual Underestimation}\label{subsec: perpetual underestimation}

To determine whether social learning incurs linear expected regret, it is useful to check whether it results in \emph{perpetual underestimation} with a significant probability.

\begin{definition}[Perpetual Underestimation]\label{def:perpunderrepr}
A group $g_0$ is \textit{perpetually underestimated} if, for all $n > \Nzero$, we have $g(\iota(n)) \neq g_0$.
\end{definition}

When group $g_0$ is perpetually underestimated, no worker from group $g_0$ is hired after the initial sampling phase. If social learning generates perpetual underestimation with a significant probability, then linear expected regret often results. In particular, under Assumption~\ref{assp:idcontext}, perpetual underestimation against any group $g \in G$ implies that firms fail to hire at least $(K_g/K)\left(N - \Nzero \right)$ best candidate, which is linear in $N$. Hence, the constant probability of perpetual underestimation (independent of $N$) precipitates linear expected regret.

Perpetual underestimation is not only inefficient but also unfair in the sense of various fairness notions (formally defined in Appendix~\ref{sec: fairness criteria}); it results in a candidate belonging to an underestimated group not being hired, implying a violation of demographic parity. Furthermore, under a symmetric environment, such a hiring policy cannot be justified by workers' underlying skills, implying a violation of equalized odds. Hence, perpetual underestimation is an extreme form of discrimination that persists for a long time. 

\subsection{Sublinear Regret with Balanced Population}\label{sec: LF balanced group}

This section analyzes the case of only one candidate arriving from each group during each period. The contextual variation implicitly urges firms to explore all the groups with some frequency. Consequently, laissez-faire has sublinear regret, implying that statistical discrimination is eventually resolved.

\begin{thm}[Sublinear Regret with a Balanced Population]\label{thm_smlcand}
Suppose Assumptions \ref{assp:twogroups}, \ref{assp:idcontext}, and \ref{assp:normalcontext}. Suppose also that $K_g = 1$ for $g = 1,2$. 
Then, expected regret $\Regret^{\text{LF}}(N)$ under the laissez-faire policy is bounded as
\begin{equation}
\Ex[\Regret^{\text{LF}}(N)] = \tilO(\sqrt{N}).
\end{equation}
\end{thm}
Let $\mu_x = ||\bmu_x||$ and $\Phi$ be the cumulative distribution function of the standard normal distribution. The constant on the top of $\Ex[\Regret^{\text{LF}}(N)]$ is inverse proportional to $1 - \Phi(\mu_x/\sigma_x)$, which approximately scales as $\exp(-(\mu_x/\sigma_x)^2/2)$.

\paragraph*{Proof.}
See Appendix~\ref{subsec_smlcand}.

To prove Theorem~\ref{thm_smlcand}, we characterize the condition with which underestimation is spontaneously resolved. Let indices $i_1$ and $i_2$ denote the majority candidate and the minority candidate.
With a constant (i.e., independent of $N$) probability, the minority group is underestimated (i.e, $\hbtheta_2(n)$ is misestimated in such that $\bx_{i_2}\hbtheta_2(n) \ll \bx_{i_2}\btheta_2$ often occurs) in early rounds due to a bad realization of the error term. Even in such a case, there is some probability of the minority candidate being hired. Since characteristics are diverse (i.e., $\sigma_x > 0$), with some probability, the majority candidate $i_1$ is not very good (i.e., $x_{i_1}\hbtheta_1(n) \approx x_{i_1}\btheta_1$ is small). In such a round, $x_{i_1} \hbtheta_1(n) < x_{i_2} \hbtheta_2(n)$ holds despite group $2$ being underestimated, and the minority candidate $i_2$ is hired. In such a case, firms update their belief about the minority, leading to a resolution of underestimation. Such events occur more frequently when workers have more diverse characteristics, i.e., $\mu_x/\sigma_x$ is small.

As anticipated by the theory of least squares, the standard deviation of $\hbtheta_g(n)$ is proportional to $(\bar{\bV}_g(n))^{-1/2}$, and we demonstrate that its diameter $(\lambdamin(\bar{\bV}_g(n)))^{-1/2}$ shrinks as $\tilO(1/\sqrt{n})$, where $\lambdamin$ is the minimum eigenvalue of a matrix. The regret per error is defined by this quantity, with the total regret being $\tilO(\sum_{n\le N} (1/\sqrt{n}))= \tilO(\sqrt{N})$.

Theorem~\ref{thm_smlcand} indicates that statistical discrimination is resolved spontaneously when candidate variation is large.
At a glance, this appears to contradict widely known results that state laissez-faire (greedy) may lead to suboptimal results in bandit problems due to underexploration. However, the variation in characteristics naturally incentivizes selfish agents to explore the underestimated group, and therefore, with some additional conditions, the probability of perpetual underestimation is bounded.

\begin{remark}
In Theorem~\ref{thm_smlcand}, we assumed that there is one candidate for each group, $K_1 = K_2 = 1$, for tractability. If we assume a larger but balanced population, $K_1 = K_2 = K/2$, then the analysis would become significantly more challenging because the maximum of normally distributed variables is not normally distributed.
However, we conjecture that a similar result would hold more generally because the variance of the expected skill of the best candidate in each group decreases only slowly as $K$ increases.\footnote{Lemma~\ref{lem_varmax} in the Appendix implies that the variance is in the order of $O(1/\log K)$.}
\end{remark}

\begin{remark}
Theorem~\ref{thm_smlcand} shares certain intuitions with the previous research \citep{kannan2018,Bastani17} demonstrating that the variation in contexts (characteristics) improves the performance of the greedy algorithm (laissez-faire) in contextual multi-armed bandit problems. However, in contrast to \citet{kannan2018}, our theorem makes no assumptions regarding the length of the initial sampling phase. Theorem 1 in \citet{Bastani17} corresponds to our paper's Theorem \ref{thm_smlcand}, and we further characterize the factor of the regret as a function of $\mu_x/\sigma_x$ rather than the diameter of the characteristics.
\end{remark}

\subsection{Large Regret with Unbalanced Population}\label{sec: LF unbalanced group}

While Theorem~\ref{thm_smlcand} implies that statistical discrimination is spontaneously resolved in the long run, it crucially relies on one unrealistic assumption---the balanced population ratio. In many real-world problems, the population ratio is unbalanced, and the discriminated group is often a demographic minority in the relevant market. We indeed find that the population ratio crucially impacts the equilibrium consequence under laissez-faire.

\begin{thm}[Substantial Regret with Unbalanced Populations]\label{thm_reglower}
Suppose Assumptions \ref{assp:twogroups}, \ref{assp:idcontext}, and \ref{assp:normalcontext}.
Suppose also that $K_2 = 1$ and $d=1$. Let $K_1 > \log_2 N$.
Then, under the laissez-faire decision rule, group $2$ is perpetually underestimated with a probability of at least $C_{\text{imb}} = \tilTheta(1)$. Accordingly, the expected regret associated with the laissez-faire decision rule is 
\begin{equation}
    \Ex\left[\Regret^{\text{LF}}(N)\right] \ge \frac{C_{\text{imb}}(N-\Nzero)}{K}  = 
    \tilOmega(N).
\end{equation}
\end{thm}%

\paragraph*{Proof.}
See Appendix \ref{subsec_reglower}. The explicit form of $C_{\text{imb}}$ is shown in Eq.~\eqref{ineq:reglower_finalform}.

In the proof of Theorem~\ref{thm_reglower}, we evaluate the probability that the following two events occur: (i)  $\htheta_2$ is underestimated, and (ii) the characteristics and skills of the hired majority workers are not very bad throughout rounds (i.e., $\max_{i: g(i)=1} x_i \htheta_1 \ge c \mu_x \theta$ for some constant $c>0$). The probability of (i) is polylogarithmic
to $N$ (i.e., $\tilTheta(1)$) and the probability that (ii) consistently holds for all the rounds $n=\Nzero+1,\dots,N$ is polylogarithmic if $K_1 > \log_2 N$. When both (i) and (ii) occur, we always have $\max_{i \in I(n)\setminus\{i_2(n)\}} x_i \htheta_1 > x_{i_2(n)} \htheta_2$ (where $i_2(n)$ is the unique minority candidate of round $n$); thus, the minority worker is never hired. Note that the majority group does not suffer from perpetual underestimation (with a significant probability) because the event that all the majority workers are bad occasionally occurs.

Theorem~\ref{thm_reglower} indicates that we should not be too optimistic about the consequence of laissez-faire. A small imbalance in the population ratio (the ratio of majority to minority is just $\log_2 N$ to $1$) could lead to a substantially unfair job allocation. Once the minority group is underestimated and the majority candidate pool is reasonably large, then the minority group is afforded no hiring opportunity, perpetuating underestimation. This insight applies to many real-world problems because unbalanced populations are commonplace.

We conjecture a substantial probability under a broader environment than the premise of Theorem~\ref{thm_reglower}. Specifically, the assumptions of $d = 1$ and $K_1 > \log_2 N$ are made only for analytical tractability, and (approximately) linear regret should be obtained under a weaker set of assumptions.
Theorem~\ref{thm_reglower} (i) focuses on perpetual underestimation, which is an extreme form of statistical discrimination, and (ii) evaluates the probability of perpetual estimation occurring loosely. 
In Section~\ref{sec: simulation}, we demonstrate that perpetual underestimation occurs with a significant probability even under the assumptions of $d = \Simd$ and $(K_1, K_2) = (10, 2)$, where the premise of Theorem~\ref{thm_reglower} does not hold.

\section{The Upper Confidence Bound Mechanism}\label{sec: subsidy scheme}

Section \ref{sec: laissez-faire} has discussed the equilibrium consequences of laissez-faire. We observed that an unbalanced population ratio leads to a substantial probability of underestimation being perpetuated. Policy intervention is demanded to improve social welfare and the fairness of the hiring market.

This section proposes a subsidy rule to resolve underestimation. We employ the idea of the \emph{upper confidence bound} (UCB) algorithm \citep{lai1985,auer2002}, which has widely been used in the literature on the bandit problem.
The UCB algorithm balances exploration and exploitation by developing a confidence interval for the true reward and evaluating each arm's performance according to its upper confidence bound to achieve this balance. Firms are generally unwilling to follow the UCB decision rule voluntarily; therefore, the government needs to provide a subsidy to incentivize firms to hire a candidate with the greatest UCB index.
This section establishes a UCB-based subsidy rule and evaluates its performance.

The adaptive selection of candidates based on history can induce some bias, meaning the standard confidence bound no longer applies. To overcome this issue, we use martingale inequalities \citep{selfnormalized,rusmevichientong2010,abbasi2011}.
We here introduce the confidence interval for the true coefficient parameter, $(\btheta_g)_{g\in G}$. 

\begin{definition}[Confidence Interval]\label{def:confinterval}
Given the group $g$'s collected data matrix $\bar{\bV}_g(n)$, the \emph{confidence interval} of group $g$'s coefficient parameter $\btheta_g$ is given by
\begin{equation}
    \EC_g(n; \delta) 
    \coloneqq \left\{ \bbtheta_g \in \Real^d: \norm{\bbtheta_g - \hbtheta_g(n)}_{\bar{\bV}_g(n)} \le \sigma_\eps
    \sqrt{d \log\left(\frac{\det(\bar{\bV}_g(n))^{1/2}\det(\lambda \bI_d)^{-1/2}}{\delta}\right)}
    +
    \lambda^{1/2} S
    \right\},
\end{equation}
where $||\bv||_{\bA} = \sqrt{\bv' \bA \bv}$ for a $d$-dimensional vector $\bv$ and $d \times d$ matrix $\bA$.
\end{definition}

\citet{abbasi2011} study the property of this confidence interval, and they prove that the true parameter $\btheta_g$ lies in $\EC_g(n;\delta)$ with probability $1-\delta$ (Lemma \ref{lem:abbasi}). By choosing a  sufficiently small $\delta$,\footnote{We typically choose $\delta = 1/N$ to make the confidence interval asymptotically correct in the limit of $N \to \infty$.} it is ``safe'' to assess that worker $i$'s skill is \emph{at most}
\begin{equation}
    \tilde{q}_i(n) \coloneqq \max_{\bbtheta_{g(i)} \in \EC_{g(i)}(n;\delta)} \bx_i'\bbtheta_{g(i)}.
\end{equation}
We call $\tilde{q}_i(n)$ the \emph{UCB index} of worker $i$'s skill. Intuitively, $\tilde{q}_i(n)$ is worker $i$'s skill in the most optimistic scenario. The confidence interval $\EC_g(n;\delta)$ shrinks as we obtain more data about group $g$. Hence, the UCB index $\tilde{q}_i(n)$ converges to true predicted skill $q_i(n)$ as the size of the data set increases.

\begin{definition}[UCB Decision Rule]
The UCB decision rule selects the worker with the greatest UCB index; i.e.,
\begin{equation}\label{eq: UCB decision rule}
    \iota(n) = \argmax_{i \in I(n)}\tilde{q}_i(n).
\end{equation}
\end{definition}

The UCB index $\tilde{q}_i(n)$ is close to the pointwise estimate $\hat{q}_i(n)$ when society has rich data about group $g(i)$, because $\EC_{g(i)}(n;\delta)$ is small in such cases. However, when information about group $g(i)$ is insufficient, $\tilde{q}_i(n)$ is much larger than $\hat{q}_i(n)$, because the firm is unsure about the true skill of worker $i$ and $\EC_{g(i)}(n;\delta)$ is large. In this sense, the UCB decision rule offers affirmative actions toward underexplored groups.

The subsidy amount is proportional to the uncertainty surrounding the candidate's characteristics, which is represented by the confidence interval $\EC_g(n)$ for $g=g(i)$. The magnitude of the confidence interval $\EC_g(n)$ is inverse proportional to ${\bar{\bV}_g(n)} = \bV_g(n) + \lambda \bI_d$.\footnote{The standard OLS has a confidence bound of the form $\btheta_g - \hbtheta_g(n) \sim \Normal(0, \sigma_\eps^2 \bV^{-1}_g(n))$ and thus $|\btheta_g - \hbtheta_g(n)| \sim \sigma_\eps \bV^{-1/2}_g(n)$. The price of adaptivity causes the martingale confidence bound $\EC_g(n)$ to be larger than the OLS confidence bound for two factors: (i) $\sqrt{d}$ factor, and (ii) $\sqrt{\log(\det(\bar{\bV}_g(n)))}$ factor. As discussed in \citet{Xu2018AFA}, the $\sqrt{d}$ factor unnecessarily overestimates the confidence bound in most cases.} 
Hence, if the data $\bV_g(n)$ do not vary substantially for a particular dimension of $\bx_i$, then that dimension's prediction can be inaccurate. In such cases, the UCB decision rule recommends hiring a candidate that contributes to increasing that dimension's data. For example, when a candidate possesses skills previous hires do not, then the candidate's UCB index tends to become large.

The UCB decision rule efficiently balances exploration and experimentation. Accordingly, it has sublinear regret in general environments.

\begin{thm}[Sublinear Regret of UCB]\label{thm:ucb}
Suppose Assumption \ref{assp:normalcontext}.
Let $\Regret^{\text{UCB}}$ be the regret from the UCB decision rule. Let $\lambda \ge \max(1,(\LCont{1/N})^2)$, where $\LCont{1/N}$ is an $O(\sqrt{d\log{KN}})$ value defined in Lemma~\ref{lem:largestcontext} in Appendix.
Then, by choosing $\delta=1/N$, regret under the UCB decision rule is bounded as
\begin{equation} 
\Ex[\Regret^{\mathrm{UCB}}(N)] = 
\tilO(\sqrt{N}).
\end{equation}
\end{thm}

\paragraph*{Proof.}
See Appendix \ref{subsec_ucb}.

There are three remarks. First, $\tilO(\sqrt{N})$ regret is the optimal rate for these sequential optimization problems under partial feedback \citep{pmlr-v15-chu11a}. Hence, Theorem \ref{thm:ucb} states that the UCB decision rule effectively prevents perpetual underestimation and is asymptotically efficient.
Second, Theorem~\ref{thm:ucb} relies only on Assumption~\ref{assp:normalcontext}, and therefore, the regret under UCB is sublinear even when groups have a fundamental disparity besides their group sizes. Accordingly, even when the groups are asymmetric, the UCB decision rule satisfies several fairness notions (see Appendix~\ref{sec: fairness criteria} for details).
Third, differing from the case of laissez-faire, where the factor depends on the variation of the context (Theorem \ref{thm_smlcand}), Theorem~\ref{thm:ucb} provides a reasonably small regret bound even when $\sigma_x$ is very small.

To implement the UCB decision rule, we need to satisfy the firms' obedience condition \eqref{eq: the candidate firm n hires} in conjunction with the UCB decision rule \eqref{eq: UCB decision rule}. In the following, we propose one of the most straightforward subsidy rules.

\begin{definition}[UCB Index Subsidy Rule]\label{defn: UCB index subsidy rule}
The \emph{UCB index subsidy rule} $s$ subsidizes firm $n$ to hire worker $i$ who arrives by
\begin{equation}
    s_i(n;h(n)) = \tilde{q}_i(n;h(n)) - \hat{q}_i(n;h(n)).
\end{equation}
\end{definition}

The UCB index subsidy rule aligns each firm's incentive with the maximization of the UCB index, thereby incentivizing firms to follow the UCB decision rule.

\begin{thm}[Sublinear Subsidy of the UCB Index Subsidy Rule]\label{thm:ucb index subsidy scheme}
Under the same assumptions as Theorem \ref{thm:ucb}, the amount of the subsidy required by the UCB index subsidy rule is bounded as
    \begin{equation} 
        \Ex[\Subsid^{\text{UCB-I}}(N)] = 
        \tilO(\sqrt{N}).
    \end{equation}
\end{thm}

\paragraph*{Proof.}
See Appendix~\ref{subsec: ucb index subsidy rule}.

\begin{remark}
The UCB index subsidy rule is an \emph{index policy} in the sense that the subsidy amount is independent of the information about rejected workers. The UCB index subsidy rule demands the smallest budget among all index policies implementing the UCB decision rule. In Appendix~\ref{sec: appendix subsidy rule design}, we consider a non-index subsidy rule that implements the UCB decision rule with a smaller budget.
\end{remark}

\section{The Hybrid Mechanism}\label{sec: Hybrid Decision Rule}

Although the UCB mechanism effectively prevents perpetual underestimation and achieves sublinear regret in general environments, it has one drawback: it continues subsidies in perpetuity. Even for a large $n$, there remains a gap between estimated skill $\hat{q}_i(n)$ and the UCB index $\tilde{q}_i(n)$. This is undesirable for several reasons. First, introducing a permanent policy is often more politically difficult than introducing a temporary policy. 
Second, a long-term distribution of subsidies tends to increase the required budget.
Third, in addition to the subsidy itself, the permanent allocation of the subsidy features (unmodeled) administrative costs. 

To overcome these limitations, we propose the \emph{hybrid mechanism}, which initially uses the UCB mechanism but switches to laissez-faire by terminating the subsidy at some point. We abandon the UCB phase upon receiving sufficient minority-group data to induce spontaneous exploration. Similar to the UCB mechanism, our hybrid mechanism has $\tilO(\sqrt{N})$ regret. Furthermore, its expected total subsidy amount is $\tilO(1)$, while the UCB mechanism needs $\tilO(\sqrt{N})$ subsidy.

The construction of the hybrid mechanism is as follows. Let $\sui_i(n) = \tilde{q}_i(n) - \hat{q}_i(n)$ be the size of the confidence bound. Note that, $\sui_i(n)$ corresponds to the amount of the subsidy allocated by the UCB index subsidy rule (Definition~\ref{defn: UCB index subsidy rule}).
The \emph{hybrid index} $\Uimp_i$ is defined as
\begin{equation}\label{eq: new UCB decision rule}
\Uimp_i(n;h(n)) \coloneqq
\begin{cases} 
\tilde{q}_i(n;h(n)) &\mbox{if } \sui_i(n;h(n)) > \ahyb ||\hbtheta_{g(i)}(n;h(n))||, \\
\hat{q}_i(n;h(n)) & \mbox{otherwise},  
\end{cases}
\end{equation}
where $a \ge 0$ is the mechanism's parameter.

The hybrid index is literally a ``hybrid'' of estimated skill $\hat{q}_i(n)$ and the UCB index $\tilde{q}_i(n)$. If the difference between the UCB index and estimated skill surpasses the threshold (i.e., $\sui_i(n) > \ahyb ||\hbtheta_{g(i)}(n)||$), then the hybrid index is equal to the UCB index $\tilde{q}_i(n)$. The confidence bound $|\tilde{q}_i(n) - \hat{q}_i(n)|$ is large when society has insufficient knowledge about group $g(i)$, which is typically the case during early stages of the game. Once this gap falls below the threshold (i.e., $\sui_i(n) \le \ahyb ||\hbtheta_{g(i)}(n)||$), then the hybrid index switches to the estimated skill $\hat{q}_i(n)$.

The hybrid decision rule is defined as the rule that hires the greatest hybrid index.

\begin{definition}[Hybrid Decision Rule]
The \emph{hybrid decision rule} selects the worker who has the greatest hybrid index; i.e.,
\begin{equation}
    \ih(n;h(n)) = \argmax_{i\in I(n)} \Uimp_i(n;h(n)).
\end{equation}
\end{definition}

Since the hybrid decision rule is a hybrid of the UCB decision rule and the laissez-faire decision rule, it can be implemented by mixing the laissez-faire subsidy rule and the UCB index subsidy rule.

\begin{definition}[Hybrid Index Subsidy Rule]
Let $\sui_i$ be the UCB index subsidy rule. The \emph{hybrid index subsidy rule} $\shi$ is defined by
\begin{equation}
    \shi_i(n;h(n)) \coloneqq
    \begin{cases} 
        \sui_i(n;h(n)) &\mbox{if } \sui_i(n;h(n)) > \ahyb ||\hbtheta_{g(i)}(n;h(n))||, \\
        0 & \mbox{otherwise}.
    \end{cases} %
\end{equation}
\end{definition}

The following theorems characterize the regret and the total subsidies associated with the hybrid mechanism.
\begin{thm}[Performance of the Hybrid Mechanism]
\label{thm:ucbimp}
Suppose Assumptions \ref{assp:twogroups}, \ref{assp:idcontext}, and \ref{assp:normalcontext}. Then, by choosing $\delta=1/N$, regret associated with the hybrid decision rule $\ih$ is bounded as
\begin{equation}
\Ex[\Regret^{\text{H}}(N)] = 
\tilO(\sqrt{N}).
\end{equation}
Furthermore, for any $a > 0$, the total amount of the subsidy under the hybrid index subsidy rule ($\Subsid^{\text{H-I}}$) is bounded as
\begin{align}
    \Ex[\Subsid^{\text{H-I}}(N)] 
    = \tilTheta(1).
\end{align}
\end{thm}

\paragraph*{Proof.}
See Appendix \ref{subsec:ucbimp}.

Theorem \ref{thm:ucbimp} states that (i) the order of the regret under the hybrid decision rule is the same as the original UCB, and (ii) the subsidy amount is reduced to $\tilO(1)$ (with respect to $N$). This is a substantial improvement from the UCB mechanism, which requires the $\tilO(\sqrt{N})$ subsidy. 

The threshold for switching from the UCB mechanism to laissez-faire is crucial for guaranteeing the performance of the hybrid mechanism. Our threshold, $\ahyb ||\hbtheta(n)||$, is determined such that the hybrid decision rule $\ih$ satisfies \emph{proportionality}, a new concept that this paper establishes.
We prove that the amount of exploration exerted by the hybrid decision rule is proportional to the UCB decision rule.
This property guarantees that the hybrid rule resolves underestimation and secures the expected regret of $\tilO(\sqrt{N})$. The formal statement of the proportionality appears in Lemma~\ref{lem:parity} in Appendix~\ref{subsec:ucbimp}. 

\section{Interviews and the Rooney Rule}\label{sec: Rooney Rule}

Although subsidy rules effectively resolve statistical discrimination, they are often difficult to implement in practice. This section articulates the advantages and disadvantages of the \emph{Rooney Rule}, a regulation that requires each firm to invite at least one candidate from each group to an on-site interview. The Rooney Rule is easier to implement because it requires neither a subsidy nor meeting a hiring quota.

To incorporate the additional information firms acquire through the interview, we modify the model as follows. In the modified model, each round $n$ comprises two stages. 
At the first stage, firm $n$ observes the characteristics $\bx_i$  of each arriving agent $i \in I(n)$. Based on $\bx_i$, firm $n$ selects a shortlist of \emph{finalists} $I^F(n) \subseteq I(n)$, where $|I^F(n)| = K^F$ for some constant $K^F\in \Natural$. At the second stage, by interviewing finalists, firm $n$ observes an additional signal $\eta_i$ for each finalist $i$ \citep[as assumed in][]{DBLP:conf/innovations/KleinbergR18}. Firm $n$ predicts each finalist $i$'s skill from the characteristics $\bx_i$ and the additional signal $\eta_i$, and hires one worker from the set of finalists, $\iota(n) \in I^F(n)$. Firms are not allowed to hire a worker not selected as a finalist. After the firm's decision, the skill of the hired worker $y_{\iota(n)}$ is publicly disclosed.

We assume the following linear relationship between skill $y_i$ and observable variables $\bx_i$:
$ 
y_i = \bx_i' \btheta_{g(i)} + \eta_i + \eps_i
$ %
The ``noise'' term comprises two variables: $\eta_i$ and $\eps_i$. $\eta_i$ is revealed as an additional signal when the firm chooses $i$ as a finalist. However, $\eps_i$ remains unpredictable even after the interview.
For analytical tractability, we make the following two assumptions.

\begin{assp}[Two Finalists]\label{assp: two finalists}
Each firm can invite only two finalists; i.e., $K^F = 2$.
\end{assp}%
\begin{assp}[Normal Additional Signals]\label{assp:normaleta}
Each additional signal that a finalist reveals follows a normal distribution, $\eta_i \sim \Normal(0, \sigma_\eta^2)$, i.i.d.
\end{assp}

\begin{remark}
If $\sigma_\eta = 0$, then the two-stage model is equivalent to the one-stage model that we have considered in the previous sections.
\end{remark}

\subsection{Failure of Laissez-Faire in the Two-Stage Model}

This subsection analyzes the performance of laissez-faire in this two-stage setting. The result is analogous to the one-stage case (Theorem~\ref{thm_reglower}): laissez-faire often falls in perpetual underestimation, and therefore, has linear regret.

First, we define regret. As in the one-stage model, the benchmark is the first-best decision rule, which is the rule firms would apply if the coefficient parameter $\btheta$ were known. Clearly, the first-best decision rule would greedily invite top-$K^F$ workers in terms of $q_i$ to the final interview. We denote this set of finalists chosen by the first-best decision rule in round $n$ by $\bar{I}^F(n)$. Formally, $\bar{I}^F(n)$ is obtained by solving the following problem:
\begin{equation}\label{ineq:lf_true_finalists}
    \bar{I}^F(n) = \argmax_{I'\subseteq I(n)}\sum_{i\in I'}q_i\hspace{1em}\text{s.t. }|I'| = K^F.
\end{equation}
After that, the first-best decision rule would observe the realization of $\eta_i$ for $i\in \bar{I}^F(n)$, and then hire the worker $i$ who has the greatest skill predictor: $q_i + \eta_i$. \emph{Unconstrained two-stage regret} (U2S-Reg) is defined as the loss compared with this first-best decision rule. (This type of regret is named ``unconstrained'' because we later introduce an alternative definition.)

\begin{definition}[Unconstrained Two-Stage Regret]
In the two-stage hiring model, the \emph{unconstrained two-stage regret} $\URegret$ of decision rule $\iota$ is defined as follows:
\begin{align}\label{eq: regret two-stage}
\URegret(N) &= \sum_{n=1}^N \left\{\max_{i \in \bar{I}^F(n)} \left(q_i  + \eta_i \right) - \left( q_{\iota(n)} + \eta_{\iota(n)} \right)\right\}.
\end{align}
\end{definition}

Under laissez-faire, the optimal strategy of firm $n$ is to choose candidates greedily based on their estimated skills, i.e.,
\begin{equation}
    I^F(n) = \argmax_{I'\subseteq I(n)}\sum_{i\in I'}\hat{q}_i(n)\hspace{1em}\text{s.t. }|I'| = K^F.
\end{equation}
After observing the realization of the additional signals $\eta_i$, firm $n$ selects the candidate who has the greatest estimated skill: $\iota(n) = \argmax_{i \in I^F(n)}\left\{\hat{q}_i(n) + \eta_i\right\}$.

Even in the two-stage model, laissez-faire has linear regret when the population ratio is unbalanced.

\begin{thm}[Failure of Laissez-Faire in the Two-Stage model]\label{thm_2SLF}
Suppose Assumptions \ref{assp:twogroups}, \ref{assp:idcontext}, \ref{assp:normalcontext}, \ref{assp: two finalists}, and \ref{assp:normaleta}.
Suppose also that $K_2 = 1$ and $d=1$. Let $K_1 - \log_2(K_1 + 1) > \log_2 N$. 
Then, under the laissez-faire decision rule, group $2$ is perpetually underestimated with the probability $\tilOmega(1)$.
Accordingly, the expected regret associated with the laissez-faire decision rule is
\begin{equation}
    \mathbb{E}\left[\URegret^{\mathrm{LF}}(N)\right] = \tilOmega(N).
\end{equation}
\end{thm}

\paragraph*{Proof.}
See Appendix~\ref{subsec: proof of thm_2SLF}.

The proof idea of Theorem~\ref{thm_2SLF} is as follows. Under laissez-faire, each firm $n$ interviews the two finalists with the greatest estimated skills, $\hat{q}_i(n)$. If both finalists belong to the majority group, then minority candidates are never hired, regardless of the $\eta_i$ for each finalist. By evaluating the probability that both finalists are majority candidates, we derive the probability of perpetual underestimation.
Thus, even in a two-stage setting, the laissez-faire decision has linear regret under an imbalanced population.

\subsection{The Rooney Rule and Exploration}\label{subsec: rooney exploration}

Given laissez-faire does not mitigate perpetual underestimation, desirable policy intervention is necessary.

\begin{definition}[Rooney Rule]\label{def:rooney}
In the two-stage hiring model, the \emph{Rooney Rule} requires each firm $n$ to select at least one finalist from every group $g \in G$; i.e., for every $n$ and every $g\in G$, $I^F(n)$ must satisfy
\begin{equation}\label{ineq:const_rooney}
    \left|\left\{i\in I^F(n) \mid g(i) = g\right\}\right| \ge 1.
\end{equation}
\end{definition}

Under Assumption \ref{assp:twogroups} and \ref{assp: two finalists}, each firm interviews one majority candidate and one minority candidate. To analyze how the Rooney Rule resolves statistical discrimination, we introduce a weaker notion of regret, \emph{constrained two-stage regret}.

\begin{definition}[Constrained Two-Stage Regret]
In the two-stage hiring model, the \emph{constrained two-stage regret} ($\CRegret$) of decision rule $\iota$ is defined as follows:
\begin{align}\label{eq: regret two-stage constrained}
\CRegret(N) &= \sum_{n=1}^N \left\{\max_{i \in \breve{I}^F(n)} \left(q_i  + \eta_i \right) - \left( q_{\iota(n)} + \eta_{\iota(n)} \right)\right\},
\end{align}
where $\breve{I}^F(n)$ is given by
\begin{align}\label{ineq:rooneyfstbest constrained}
    &\breve{I}^F(n) = \argmax_{I'\subseteq I(n)} \sum_{i\in I}q_i\\
    \text{s.t. }&|I'| = K^F, \label{eq: Rooney const 1}\\
    &\forall g\in G, \ \left|\left\{i\in I' \mid g(i) = g\right\}\right| \ge 1. \label{eq: Rooney const 2}
\end{align}
\end{definition}

In plain words, $\breve{I}^F(n)$ is the best set of finalists who satisfy the constraint \eqref{ineq:const_rooney}. If Eq.~\eqref{ineq:const_rooney} is imposed as an ``exogenous constraint'' (rather than a policy), the first-best decision rule would interview $\breve{I}^F(n)$ to maximize social welfare.
Constrained regret enables us to identify whether the Rooney Rule prevents perpetual underestimation: if perpetual underestimation occurs under the Rooney Rule, then the constrained regret is linear in $N$.

Under the Rooney Rule, myopic firm $n$ greedily chooses candidates based on estimator $\hat{q}_i(n)$ subject to the following constraints:
\begin{align} \label{ineq:rooneyfstbest constrained estimated}
    &I^F(n) = \argmax_{I'\subseteq I(n)} \sum_{i\in I} \hat{q}_i(n) \hspace{1em}\text{s.t. }\eqref{eq: Rooney const 1}\text{ and }\eqref{eq: Rooney const 2}.
\end{align}
and $\iota(n) = \argmax_{i \in I^F(n)} \left\{\hat{q}_i(n) + \eta_i\right\}$.

The following theorem states that the Rooney Rule resolves underestimation.
\begin{thm}[Sublinear Constrained Regret under the Rooney Rule]\label{thm:rooney_main}
Suppose Assumptions \ref{assp:twogroups}, \ref{assp:idcontext}, \ref{assp:normalcontext}, \ref{assp: two finalists}, and \ref{assp:normaleta}. Then, regret under the Rooney Rule is bounded as 
\begin{equation}\label{eq: C2S regret result}
\mathbb{E}\left[\CRegret^{\mathrm{Rooney}}(N)\right] 
= \tilO(\sqrt{N}).
\end{equation}
\end{thm}

\paragraph*{Proof.}
See Appendix \ref{subsec:Rooney_main proof}.

In the proof of Theorem~\ref{thm:rooney_main}, we show that the factor of \eqref{eq: C2S regret result} exhibits an exponential dependency\footnote{See definition of $C_6$ in the proof.} on signal variance $\sigma_\eta$, which implies that a sufficiently large $\sigma_\eta$ is required for a reasonable bound.

\subsection{The Rooney Rule and Exploitation}\label{subsec: rooney exploitation}

Although the Rooney Rule prevents statistical discrimination (Theorem \ref{thm:rooney_main}), it may worsen social welfare in terms of the original unconstrained regret. 
The intuition is as follows.
An unbalanced population ratio produces a significant probability that more than one majority candidate is highly skilled. In that case, the \emph{true} predicted skill of the second-best majority candidate ($q_i$) is likely to be greater than that of the best minority candidate. This feature raises constant regret per round: when $\eta_i$ is normally distributed, any finalist has a positive probability of being hired. Hence, the skill level of all finalists matters, and therefore, firms prefer to interview top-$K^F$ candidates who have the greatest skill. The Rooney Rule prevents this outcome. This effect would present even when firms had perfect information about coefficients $\btheta$. Consequently, the loss from the constraint \eqref{ineq:const_rooney} is constant per round, and the Rooney Rule results in $\Omega(N)$ unconstrained regret for $N$ rounds.

\begin{thm}[Linear Unconstrained Regret under the Rooney Rule]\label{thm:Rooney Unconstrained Failure}
Suppose Assumptions \ref{assp:twogroups}, \ref{assp:idcontext}, \ref{assp:normalcontext}, \ref{assp: two finalists}, and \ref{assp:normaleta}. Then, regret under the Rooney Rule is bounded as 
\begin{equation}
\mathbb{E}\left[\URegret^{\mathrm{Rooney}}(N)\right] = \Omega(N).
\end{equation}
\end{thm}%
The proof is straightforward from the argument above, and therefore, is omitted.

Although the laissez-faire and the Rooney Rule have linear unconstrained regret, these two results have different causes for the outcome in each case: laissez-faire produces linear regret due to underexploration, whereas the Rooney Rule produces linear regret due to underexploitation.
One way to resolve this is by combining the two. By starting with the Rooney Rule and abolishing it after obtaining sufficiently rich data, we could mitigate the approach's disadvantage. 
Section~\ref{sec: simulation} demonstrates the performance of such a mechanism.

\section{Simulation}\label{sec: simulation}
This section presents the outcomes of our simulations. 
Unless specified, model parameters are set as $d = \Simd, \btheta = \Simtheta, \bmu_x = (\Simmux, \dots, \Simmux), \sigma_x = \Simsigmax$, $\sigma_\eps = \Simsigmaeps$, $\lambda = \Simlambda$, and $N= 1,000$. Group sizes are set to be $(K_1,K_2) = (10, 2)$. The initial sample size is $\Nzero = K_1 + K_2$, and the sample size for each group is equal to its population ratio: $\Nzero_1 = K_1, \Nzero_2 = K_2$.
We draw $\Numrun$ paths independently for each simulation scenario.
The value of $\delta$ in the confidence bound is set to $\Simdelta$. 

\subsection{The Effects of Population Ratio}

\begin{figure}[t!]
    \centering
    \begin{minipage}[t]{0.48\textwidth}
        \centering
         \includegraphics[width=\textwidth]{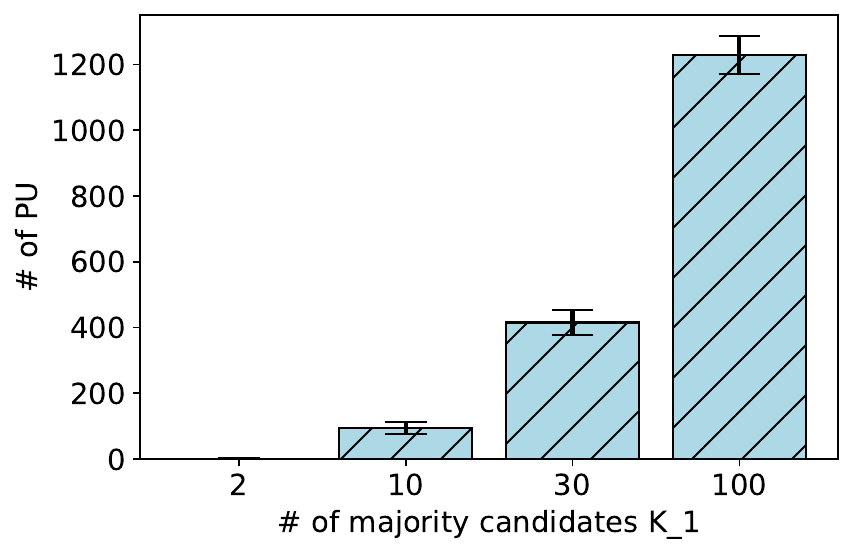}
         \caption{Frequency of perpetual underestimation under laissez-faire.}
         \label{fig:groupsize_pu}
    \end{minipage}
    \hspace{0.02\textwidth}
    \begin{minipage}[t]{0.48\textwidth}
         \centering
         \includegraphics[width=\textwidth]{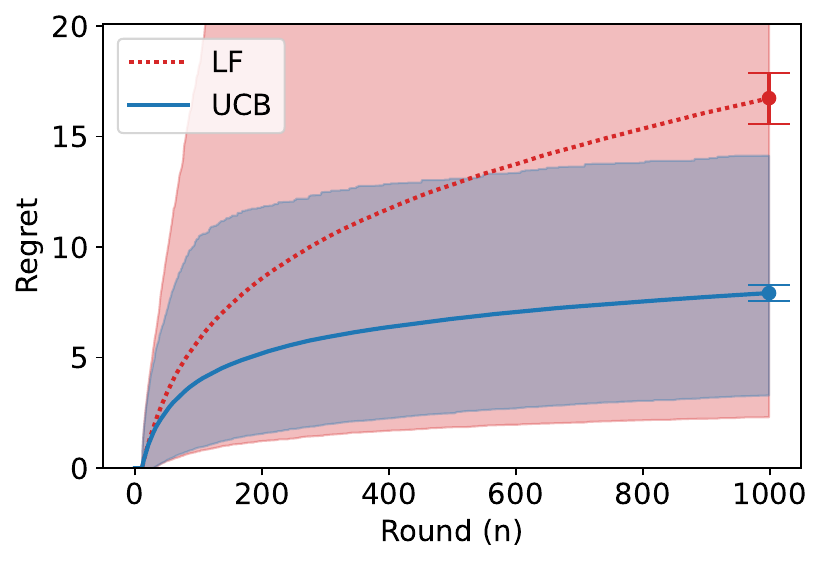}
         \caption{Regret under the LF and UCB decision rules.}
         \label{fig:policycomp_regret}
    \end{minipage}
    \raggedright
    \footnotesize
    \medskip
  
    \textbf{Left:} Across $\Numrun$ runs. The error bars represent the two-sigma binomial confidence intervals.
    
    \textbf{Right:} The lines are averages over sample paths, the areas cover between $5\%$ and $95\%$ percentiles of runs, and the error bars at $N = 1,000$ are the two-sigma confidence intervals.
\end{figure}

We test how the population ratio impacts the frequency of perpetual underestimation. The decision rule is fixed to laissez-faire (LF). We fix the number of minority candidates in each round to two (i.e., $K_2 = 2$) and vary the number of majority candidates ($K_1 = 2, 10, 30, 100$).

Figure~\ref{fig:groupsize_pu} exhibits the simulation result. Consistent with our theoretical analyses, we observe that (i) as indicated by Theorem~\ref{thm_smlcand}, laissez-faire rarely produces perpetual underestimation if the population is balanced (i.e., $K_1$ is close to $K_2 = 2$), and (ii) as indicated by Theorem~\ref{thm_reglower}, the larger the population of majority workers (i.e., $K_1$ increases), the more frequently perpetual underestimation occurs. 
With $K_1 = 10$, perpetual underestimation occurs more than 2\% of runs, which is large enough to ensure that laissez-faire produces (approximately) linear regret.

\subsection{Laissez-Faire vs the UCB Mechanism}

\begin{figure}[t!]
    \centering
    \begin{minipage}[t]{0.48\textwidth}
         \centering
         \includegraphics[width=\textwidth]{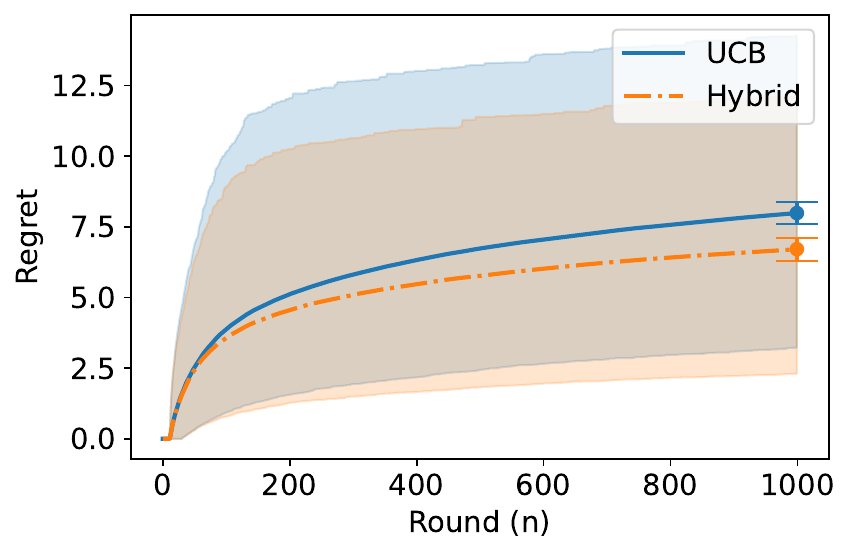}
         \caption{Regret under the UCB and hybrid decision rules.}
         \label{fig:iucb_regret}
    \end{minipage}
    \hspace{0.02\textwidth}
    \begin{minipage}[t]{0.48\textwidth}
         \centering
         \includegraphics[width=\textwidth]{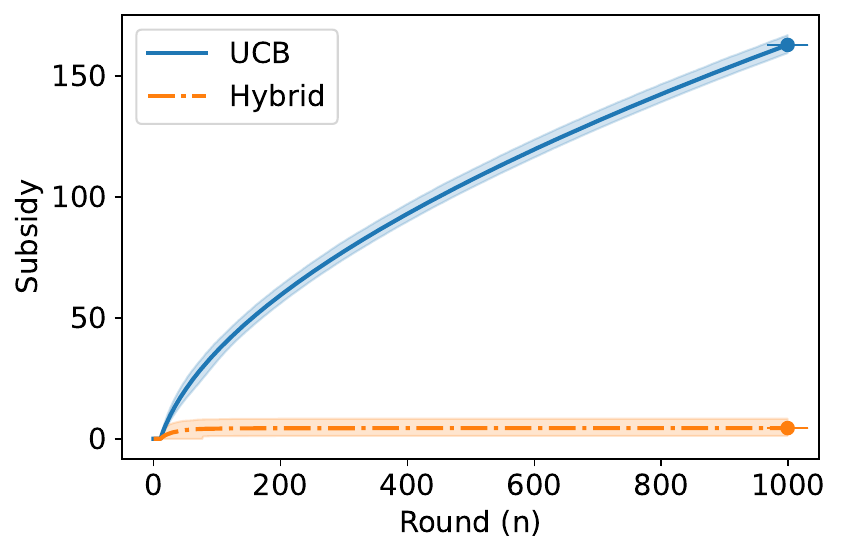}
         \caption{Budget required by the UCB and hybrid index subsidy rules.}
         \label{fig:iucb_subsidy}
    \end{minipage}
    \raggedright
    \footnotesize
    \medskip
    
    \textbf{Note:} The lines are averages over sample paths, the areas cover between $5\%$ and $95\%$ percentiles of runs, and the error bars at $N = 1,000$ are the two-sigma confidence intervals.
\end{figure}

Figure~\ref{fig:policycomp_regret} compares the regret associated with the laissez-faire (LF) decision rule and the UCB decision rule. As indicated by Theorem~\ref{thm_reglower}, our simulation shows that laissez-faire has a significant probability of underestimating the minority group. Consequently, laissez-faire sometimes causes perpetual underestimation, and regret grows (approximately) linearly to $n$. Furthermore, due to the possibility of perpetual underestimation, the confidence intervals of the sample paths (denoted by the red area) are very large, indicating the highly uncertain performance of laissez-faire. In contrast, consistent with Theorem~\ref{thm:ucb}, the UCB decision rule performs much more stably. Since the UCB rule avoids underexploration, it does not cause perpetual underestimation.

\subsection{The UCB Mechanism vs the Hybrid Mechanism}\label{subsec: simulation ucb vs hybrid}

Next, we compare the performance of the UCB and hybrid mechanisms. The parameter of the hybrid mechanism is set to be $a= \Sima$. Figure~\ref{fig:iucb_regret} shows the associated regret. 
As Theorems~\ref{thm:ucb} and \ref{thm:ucbimp} anticipated, the regret associated with the two decision rules are similar (these two decision rules have the same order: $\tilO(\sqrt{N})$). 
Figure~\ref{fig:iucb_subsidy} compares the subsidy rules. As Theorems~\ref{thm:ucb index subsidy scheme} and \ref{thm:ucbimp} predicted, the subsidy required for the UCB index rule grows at the rate of $\tilO(\sqrt{N})$, whereas the hybrid index subsidy rule only requires only a constant subsidy, implying that the policy intervention can be terminated at some point. Furthermore, the hybrid index subsidy rule requires a much smaller budget than the UCB index subsidy rule. To summarize, the hybrid mechanism produces similar regret as the UCB mechanism with a much smaller budget. 

In Appendix~\ref{subsec: simulation cost saving}, we demonstrate that while the budget required by UCB is improved substantially if the subsidy rule does not have to be an index policy, whereas its total subsidy cannot be bounded by a constant and requires a large subsidy in the long run.

\subsection{The Rooney Rule}

\begin{figure}[t!]
    \centering
    \begin{minipage}[t]{0.48\textwidth}
         \centering
         \includegraphics[width=\textwidth]{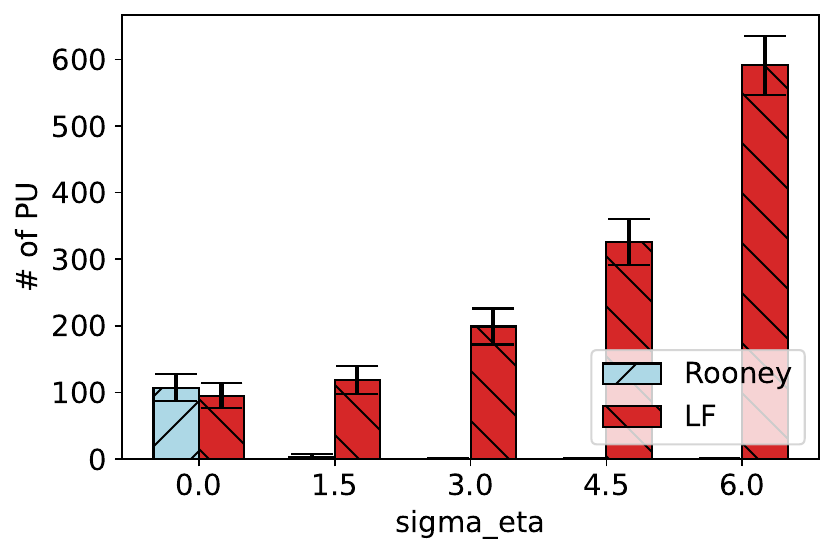}
    \caption{Frequency of perpetual underestimation in the two-stage model.}
    \label{fig:rooney_pu}
    \end{minipage}
    \hspace{0.02\textwidth}
    \begin{minipage}[t]{0.48\textwidth}
         \centering
         \includegraphics[width=\textwidth]{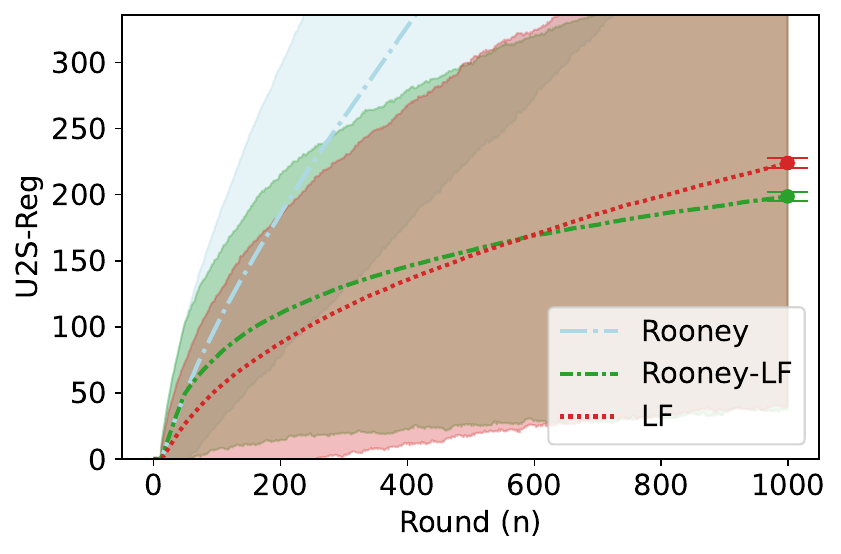}
        \caption{U2S-Reg under the LF, Rooney, and Rooney-LF decision rules.}
        \label{fig:rooney_largeg2_strong_regret}
    \end{minipage}
    \raggedright
    \footnotesize
    \medskip
    
    \textbf{Left:} Across $\Numrun$ runs. The error bars represent the two-sigma binomial confidence intervals.
    
    \textbf{Right:} The lines are averages over sample paths, the areas cover between $5\%$ and $95\%$ percentiles of runs, and the error bars at $N = 1,000$ are the two-sigma confidence intervals.
\end{figure}

This subsection compares the performance of the Rooney Rule with that of the laissez-faire decision rule. Figure~\ref{fig:rooney_pu} depicts the relationship between the frequency of perpetual underestimation and the informativeness of the signal obtained at the second stage (measured by $\sigma_\eta^2$, the variance of $\eta_i$) under both rules. When $\sigma_\eta^2$ is large, the Rooney Rule effectively resolves underestimation.

Figure~\ref{fig:rooney_largeg2_strong_regret} compares U2S-Reg associated with each rule. We set $\sigma_\eta = \Simsigmaeta$. While both rules produce linear regret, the Rooney Rule suffers from more regret due to underexploitation. This shortcoming can be overcome by using the Rooney Rule as a temporary policy. the ``Rooney-LF'' decision rule begins with the Rooney Rule and shifts to laissez-faire after $\Simrooneystop$ rounds. This approach achieves both less regret and fairer hiring.

\section{Conclusion}\label{sec: conclusion}

We have studied statistical discrimination using a contextual multi-armed bandit model.
Our dynamic model articulates how a failure of social learning produces statistical discrimination. In our model, the insufficiency of data about minority groups is endogenously generated. This data shortage prevents firms from accurately estimating the skill of minority candidates. 
Consequently, firms tend to prefer hiring majority candidates, leading the data sufficiency to persist. This form of statistical discrimination is not only unfair but also inefficient.
We have demonstrated that an unbalanced population ratio leads laissez-faire to tend toward perpetual underestimation, an unfair and inefficient consequence.

We analyzed two possible policy interventions. One is subsidy rules that incentivize firms to hire minority candidates. Our hybrid mechanism achieves $\tilO(\sqrt{N})$ regret with $\tilO(1)$ subsidy. Another intervention is the Rooney Rule, which requires firms to interview at least one minority candidate. Our result indicates that terminating the Rooney Rule at an appropriate point would resolve statistical discrimination while maintaining the social welfare level. These results contrast with some of the previous studies  \citep[e.g.,][]{Foster1992AnEA,coate1993will,moro_general_2004} demonstrating the possible counterproductivity of affirmative-action policies.

Our analyses of the two interventions provide a consistent policy implication: Affirmative actions effectively resolve statistical discrimination caused by data insufficiency, but such actions should be lifted upon acquiring sufficient information. Accordingly, a \emph{temporary} affirmative action constitutes the best approach to resolving statistical discrimination as a social learning failure.

\bibliographystyle{apalike}
\bibliography{manual.bib}

\appendix

\section{The Design of Subsidy Rules}\label{sec: appendix subsidy rule design}

\subsection{Pivot Subsidy Rules}

This subsection provides a rationale for focusing on the UCB and hybrid index subsidy rules. First, we define an \emph{index} and \emph{index policy} as follows.

\begin{definition}[Index]
A sequence of functions $Q = (Q_i)$ where $Q_i(n; \cdot): H(n) \to \Real$ is an \emph{index} if for all $n$ and $i\in I(n)$, $Q_i(n;\cdot)$ only depends on $\bX_{g(i)}(n)$, $Y_{g(i)}(n)$,  and $\bx_i$. A subsidy rule $s$ is an \emph{index policy} if $s$ is an index.
\end{definition}

In our study, we slightly modify the standard definition of an index policy often used in multi-armed bandit literature \citep{gittins1979}. The conventional definition demands that the index of an arm (in this context, a worker) is contingent only on the data generated by that particular arm. However, as we consider a group of arms as a collective entity, focusing on the data generated by a single arm isn't very meaningful. Hence, in our definition, we stipulate that the subsidy for worker $i$ should be unaffected by two factors: (i) the characteristics of other agents (workers) in the same round (i.e., $\bx_j$ for any $j\in I(n)\setminus{i}$), and (ii) the data pertaining to other groups (i.e., $\bX_{g'}(n)$ for any $g'\neq g(i)$). This modification accommodates our group-based approach to analyzing statistical discrimination in hiring decisions.

Having a subsidy rule as an index policy is practically beneficial. When determining the subsidy assigned to the employment of worker $i$, the government does not need to observe the characteristics of all other potential candidates in the pool $I(n)\setminus\{i\}$. This feature is particularly advantageous for real-world applications. In many instances, it is challenging for government entities to gain access to the data regarding the characteristics of candidates who were not selected for the job. Consequently, implementing a non-index policy, which would require such information, becomes extremely difficult.

The estimated skill $\hat{q}$, the UCB index $\tilde{q}$, and the hybrid index $\Uimp$ are indices. The UCB index subsidy rule and the hybrid index rule are index policies. Furthermore, they also belong to a class of \emph{pivot subsidy rules} that are defined as follows:
\begin{definition}
Given that a decision rule $\iota$ that maximizes an index $Q$, i.e., 
\begin{equation}
    \iota(n) = \argmax_{i \in I(n)} Q_i(n),
\end{equation}
a \emph{pivot subsidy rule} $s$ is specified by
\begin{equation}
    s_i(n) = Q_i(n) - \hat{q}_i(n).
\end{equation}
\end{definition}

Thus, the UCB index subsidy rule is obtained by substituting $Q_i(n) = \tilde{q}_i(n)$ and the hybrid index subsidy rule is obtained by substituting $Q_i(n) = \Uimp_i(n)$.

The following theorem states the optimality of the pivot subsidy rule. Among all index policies, the pivot subsidy rule requires the smallest subsidy amount under certain conditions.

\begin{thm}[Optimality of the Pivot Subsidy Rule]\label{thm: pivot subsidy rule}
Suppose that (i) a decision rule $\iota$ maximizes an index $Q$, (ii) $Q_i(n; h(n)) \ge \hat{q}_i(n; h(n))$ for all $i$, $n$, $h(n)$, and (iii) $\inf_{i, n, h(n)}Q_i(n; h(n)) = \hat{q}_i(n)$. Then, we have the following.
\begin{enumerate}[(a)]
    \item A pivot subsidy rule implements $\iota$.
    \item Let $s$ be a pivot subsidy rule and $s'$ be an arbitrary index policy that implements $\iota$. Then, for all $i$, $n$ and $h(n)$, we have
    \begin{equation}
        s_i(n;h(n)) \le s_i'(n; h(n)).
    \end{equation}
\end{enumerate}
\end{thm}

Note that, the UCB index $\tilde{q}$ and the hybrid index $\Uimp$ satisfy Conditions (ii) and (iii). Accordingly, among all index policies that implement the same decision rule, the UCB index subsidy rule and the hybrid index subsidy rule require the smallest subsidy.

\paragraph*{Proof.}
(a) Since $\hat{q}_i(n) + s_i(n) = Q_i(n)$, firm $n$'s payoff from hiring worker $i$ is equal to $Q_i(n)$. Furthermore, since $Q_i(n) \ge \hat{q}_i(n)$, $s_i(n)\ge 0$ always holds. Accordingly, the pivot subsidy rule implements the targeted decision rule.

(b) For notational simplicity, we omit $n$, $\bX_{g}$, $Y_{g}$ from this proof. Define a correspondence $\EU$ by
\begin{equation}
    \EU(Q_i; s') \coloneqq \left\{u_i \in \Real \mid \exists i, \exists\bx_i \text{ s.t. } \hat{q}_i(\bx_i) + s_i'(\bx_i) = u_i, Q_i = Q_i(\bx_i)\right\}.
\end{equation}
The set $\EU(Q_i; s')$ represents the set all of firm $n$'s possible payoffs from hiring a worker with index $Q_i$, given that the subsidy rule $s'$ is used.

Clearly, subsidy rule $s'$ implements the decision rule $\iota$ if and only if for all distinct $i,j$, $Q_j > Q_i$ implies
\begin{equation}\label{eq: strong increasingness}
    \min \EU(Q_j; s') > \max \EU(Q_i; s').
\end{equation}
Since $\min \EU(\cdot ; s')$ is an increasing function, it is continuous at all but countably many points. Thus, $\EU(Q_i; s)$ is a singleton for almost all values of $Q_i$.

Now, suppose that $\EU(Q_i^*; s')$ is not a singleton for some $Q_i^* \in \Real_+$. Define $\Delta$ by
\begin{equation}
    \Delta \coloneqq \max\EU(Q_i^*; s') - \min\EU(Q_i^*; s').
\end{equation}
Define another subsidy rule $s''$ by setting
\begin{equation}
    s_i''(\bx_i) = \begin{cases}
    s_i'(\bx_i) & \ \text{if }Q_i(\bx_i) < Q_i^*,\\
    \min\EU(Q_i^*; s') - \hat{q}_i(\bx_i) & \ \text{if }Q_i(\bx_i) = Q_i^*,\\
    s_i'(\bx_i) - \Delta & \ \text{otherwise},
    \end{cases}
\end{equation}
for all $i$. Then, we have
\begin{equation}
    \EU(\tilde{q}_i; s'') = \begin{cases}
    \EU(\tilde{q}_i; s') & \ \text{if }Q_i < Q_i^*,\\
    \left\{\min\EU(Q_i^*; s')\right\} & \ \text{if }Q_i = Q_i^*,\\
    \EU(Q_i; s') - \Delta & \ \text{otherwise},
    \end{cases}
\end{equation}
which implies that $\EU(\cdot; s'')$ also satisfies \eqref{eq: strong increasingness}, or equivalently, $s''$ also implements the decision rule $\iota$. Furthermore, $s_i''(\bx_i) \le s_i'(\bx_i)$ for all $\bx_i$, with a strict inequality for some $\bx_i$. Accordingly, $s''$ needs a smaller budget than $s'$.

By the argument above, whenever $\EU(\cdot; s')$ does not return a singleton for some $Q_i$, the subsidy amount can be improved by filling a gap. Now, we discuss the case that $\EU(\cdot; s')$ returns a singleton for all $Q_i$; i.e., $\EU$ reduces to a function. We use $u(Q_i; s')$ to represent the firm's utility when it hires a worker with the index $Q_i$. We have
\begin{equation}
    s_i'(\bx_i) = u(Q_i(\bx_i) ;s') - \hat{q}_i(\bx_i)
\end{equation}
for all $\bx_i$. Since we require that $s_i'(\bx_i) \ge 0$ for all $\bx_i$,
\begin{equation}
    u(Q_i(\bx_i); s') - \hat{q}_i(\bx_i) \ge 0.
\end{equation}
After some history, $\hat{q}_i$ may become arbitrarily close to $Q_i$. Accordingly, $u$ must satisfy
\begin{equation}\label{eq: subsidy rule necessary condition}
    u(Q_i; s) \ge Q_i
\end{equation}
for all $q$. The pivot subsidy rule satisfies \eqref{eq: subsidy rule necessary condition} with equalities for all $q$: The UCB index subsidy rule satisfies $s_i = Q_i - \hat{q}_i$, and therefore, $u(Q_i; s) = Q_i$ for all $Q_i$. Accordingly, the pivot subsidy rule demands the minimum possible budget. \qed

\subsection{Cost-Saving Subsidy Rules}

If a subsidy rule need not be an index policy, then a decision rule can be implemented with a smaller budget. The \emph{cost-saving subsidy rule} provides a minimum subsidy to change the firm's hiring decision.

\begin{definition}[Cost-Saving Subsidy Rule]\label{defn: cost-saving}
Given an arbitrary decision rule $\iota$, a \emph{cost-saving subsidy rule} $s$ is specified by
\begin{equation}\label{eq: cost-saving subsidy amount}
    s_i(n;h(n)) = \begin{cases}
    \max_{j \in I(n)}\hat{q}_j(n;h(n)) - \hat{q}_i(n;h(n)) & \text{ if }i = \iota(n);\\
    0 & \text{ otherwise}.
    \end{cases}
\end{equation}
\end{definition}

The UCB and hybrid cost-saving subsidy rules are obtained by applying Definition~\ref{defn: cost-saving} to the UCB and hybrid decision rules.

A cost-saving subsidy rule subsidizes only the targeted worker, $\iota(n)$. Hence, for other workers $j \neq \iota(n)$, the payoff from the employment is $\hat{q}_j(n)$. The UCB cost-saving subsidy rule sets the subsidy amount $s_{\iota(n)}$ such that the payoff from hiring worker $\iota(n)$, which is $\hat{q}_{\iota(n)}(n) + s_{\iota(n)}(n)$, is equal to (or slightly larger than) the payoff from hiring the worker with the greatest estimated skill, $\max_{j \in I(n)}\hat{q}_j(n)$.

Clearly, the UCB cost-saving subsidy rule is the subsidy rule that requires the smallest budget to implement the UCB decision rule. Since fines (negative subsidies) are not allowed, the government cannot further discourage the employment of the other candidates, $j \in I(n)\setminus\{\iota(n)\}$. Hence, the UCB cost-saving subsidy rule requires the smallest budget among all subsidy rules that implements the decision rule \eqref{eq: UCB decision rule}.

\begin{thm}[Optimality of the Cost-Saving Subsidy Rule]\hspace{0.1em}\label{cor:ucb cost-saving subsidy scheme}
\begin{enumerate}[(a)]
    \item A cost-saving subsidy rule implements the decision rule with which the subsidy rule is associated.
    \item Let $s$ be a cost-saving subsidy rule and $s'$ be an arbitrary subsidy rule that implements the same decision rule. Then, for all $i$, $n$ and $h(n)$, we have
    \begin{equation}
        s_i(n;h(n)) \le s_i'(n; h(n)).
    \end{equation}
\end{enumerate}
\end{thm}
The proof is straightforward from the argument above.

Since the government hardly observes rejected candidates' characteristics, a cost-saving subsidy rule is difficult to implement. Nevertheless, since it provides the smallest subsidy for implementing a decision rule, its performance is a useful theoretical benchmark. In Section~\ref{subsec: simulation cost saving}, we demonstrate that the cost-saving method effectively reduces the subsidy required for the UCB mechanism, whereas the hybrid mechanisms outperform substantially in the long run.

\section{Bayesian Approach}\label{sec: Bayesian Approach}

The frequentist approach, which we have adopted in our analysis, views probability as the long-run frequency of events. This approach is widely used in the multi-armed bandit literature primarily due to its robustness and the difficulty in implementing the Bayesian approach in practice. In the Bayesian approach, the analyst forms a prior belief about the unknown parameter's distribution and updates it as data becomes available. However, the selection of these prior beliefs can be somewhat subjective and can significantly influence the model's results. This can be especially challenging in situations where there is limited knowledge or lack of consensus about what the prior should be. On the other hand, the frequentist approach does not rely on prior beliefs. It bases its estimation solely on the observed data, making it robust to any realization of the parameter $(\btheta_g)_{g\in G}$. Given these advantages, we have chosen to develop and analyze a frequentist model in our study.

Nevertheless, adopting a Bayesian setting leads to a similar conclusion, so long as all firms share a common prior belief. Specifically, when the common prior belief is endowed as a normal distribution, i.e.,
\begin{equation}
\btheta_g \sim \Normal(0, \kappa^2 \bI_d),
\end{equation}
then the posterior belief would also be a normal distribution, with its mean becoming $\bx_i' \hat{\btheta}_{g(i)}(n)$.

Regarding the UCB mechanism, there exists a Bayesian version of confidence region\footnote{The bound here is derived from Eq.~(4.8) in \citet{kaufm2014}.} $\EC^{\mathrm{Bayes}}_g(n;\delta)$ such that 
\begin{equation}
\mathrm{Pr}^{\mathrm{Bayes}} 
\left( 
\bigcap_n
\{
\btheta_g \in \EC_g^{\mathrm{Bayes}}(n;\delta)
\}
\right)
\ge 1 - \delta,
\end{equation}
where $\mathrm{Pr}^{\mathrm{Bayes}}$ denotes probability over the Bayes posterior, by defining
\begin{equation}
    \EC_g^{\mathrm{Bayes}}(n;\delta)
    = \left\{ \bbtheta_g \in \Real^d: \norm{\bbtheta_g - \hbtheta_g(n)}_{\bar{\bV}_g(n)} \le \sigma_\eps
    \sqrt{
    d + 
    \log\left(\frac{\pi^2 N^2}{6 \delta}\right) + 
    2 \sqrt{d  \log\left( \frac{\pi^2 N^2}{6 \delta} \right)}
    }
    \right\}.
\end{equation}
Using $\EC_g^{\mathrm{Bayes}}(n;\delta)$, we can obtain a Bayesian version of the UCB mechanism.\footnote{Note that, to run the UCB mechanism in a model, the regulator needs to know the common prior belief of firms to calculate the confidence bound $\EC_g^{\mathrm{Bayes}}(n;\delta)$.}

\section{Lemmas}

This section describes the technical lemmas that are used for deriving the theorems. 

The Hoeffding inequality, which is one of the most well-known versions of concentration inequality, provides an upper bound of the sum of bounded independent random variables. 
\begin{lem}[Hoeffding Inequality]
Let $x_1,x_2,\dots,x_n$ be i.i.d.\ random variables in $[0,1]$. Let $\bar{x} = (1/n)\sum_{t=1}^n x_t$. Then,
\begin{align}
\Pr\left[\bar{x} - \Ex[\bar{x}] \ge k \right] &\le e^{-2nk^2} \nn
\Pr\left[\bar{x} - \Ex[\bar{x}] \le -k \right] &\le e^{-2nk^2}
\end{align}
and taking union bound yields
\begin{align}
\Pr\left[|\bar{x} - \Ex[\bar{x}]| \ge k \right] &\le 2e^{-2nk^2}.
\end{align}
\end{lem}

The following is a version of concentration inequality for a sum of squared normal variables.
\begin{lem}[Concentration Inequality for Chi-squared distribution]\label{lem:concentration_chisq}
Let $Z_1,Z_2,\dots,Z_n$ be independent standard normal variables. Then,
\begin{equation}
\Pr\left[\left|\frac{1}{n} \sum_{k=1}^n Z_k^2 - 1\right| \ge t \right] \le 2 e^{-nt^2 / 8}
\end{equation}
\end{lem}

\begin{lem}[Normal Tail Bound \citep{feller-vol-1}]\label{lem:normpdf}
Let $\phi(x) \coloneqq e^{-x^2/2}/\sqrt{2 \pi}$ be the probability density function (pdf) of a standard normal random variable. 
Let $\Phi^c(x) = \int_{x}^\infty \phi(x') dx'$.
Then,
\begin{equation}
\left(\frac{1}{x} - \frac{1}{x^3}\right) \frac{e^{-x^2/2}}{\sqrt{2 \pi}}
\le
\Phi^c(x)
\le 
\frac{1}{x} \frac{e^{-x^2/2}}{\sqrt{2 \pi}}
\end{equation}
\end{lem}

\begin{lem}[Largest Context, Theorem 1.14 in \cite{gaussianmax_lecnote}]\label{lem:largestcontext}
Let
\begin{equation}
 \bx_i \sim \Normal(\bmu_x, \sigma_x \bI_d)
\end{equation}
for each $i \in I(n)$.
Let $\mu_x = ||\bmu_x||$ and
\begin{equation}
 \LCont{\delta} \coloneqq \mu_x + \sigma_x \sqrt{2 d (2 \log(KN) + \log(1/\delta))}.
\end{equation}
Then, with a probability of at least $1-\delta$, we have
\begin{equation}
\forall {i \in I(n), n \in [N]}, \ ||\bx_i|| \le \LCont{\delta}.
\end{equation}

\end{lem}

The following bounds the variance of a conditioned normal variable.
\begin{lem}[Conditioned Tail Deviation]\label{lem:incdev}
Let $x \sim \Normal(a, 1)$ be a scalar normal random variable with its mean $a \in \Real$ and unit variance. Then, for any $b \in \Real$, the following inequality holds.
\begin{align}\label{ineq:incdev}
 \Var(x| x \ge b) \ge \frac{1}{10}.
\end{align}
\end{lem}

\paragraph*{Proof.}
Without loss of generality, we assume $b = 0$ (otherwise we can reparametrize $x' = x - b \sim \Normal(a-b, 1)$). 
If $a \le 0$, the pdf of conditioned variable $x| x \ge 0$ is $2 \psi(x)$ for $x \ge 0$. Manual evaluation of this distribution\footnote{This distribution is called a folded normal distribution.} reveals that $\Var(x) \ge 1/10$. 
Otherwise ($a > 0$), the pdf of $x|x \ge b$ is $p(x) \ge \psi(x-a)$ for $x \ge a$, which implies $\Var(x|x \ge b) \ge \Var(z)$, where $z$ be a ``half-normal'' random variable\footnote{Half of the mass lies in $z>0$, the other half of mass is at $z=0$.} with its cumulative distribution function
\begin{equation}
P(z) =  
\begin{cases} \Phi(z) &\mbox{if } z > 0 \\
1/2 & \mbox{if } z = 0 \\
0 & \mbox{otherwise} 
\end{cases}.
\end{equation}
Manual evaluation of $\Var(z)$ also shows that $\Var(z) \ge 1/10$. \qed

The following diversity condition that simplifies the original definition of \cite{kannan2018} is used to lower-bound the expected minimum eigenvalue of $\bar{\bV}_g$.
\begin{lem}[Diversity of Multivariate Normal Distribution]\label{lem:normdiversity}
The context $\bx$ is $\lambda_0$-diverse for $\lambda_0 > 0$ if for any $\hat{b} \in \Real$, $\hbtheta \in \Real^d$
\begin{equation}
\lambdamin\left(\Ex\left[ \bx \bx' | \bx' \hbtheta \ge \hat{b} \right]\right) \ge \lambda_0.
\end{equation}
Let $\bx \sim \Normal(\bmu_x, \sigma_x \bI_d)$.
Then, the context $\bx$ is
$\lambda_0$-diverse with $\lambda_0 =  \sigma_x^2 / 10$. 
\end{lem}

\paragraph*{Proof.}
By definition,
\begin{align}
\lambdamin\left(\Ex\left[ \bx \bx' |  \bx' \hbtheta \ge \hat{b} \right]\right)
= \min_{\bv: ||\bv||=1} \Ex\left[ (\bv' \bx)^2  |  \bx' \hbtheta \ge \hat{b} \right] \ge \min_{\bv: ||\bv||=1} \Var\left[ \bv' \bx  |  \bx' \hbtheta \ge \hat{b} \right].
\end{align}
Let $\bbe_1,\bbe_2,\dots,\bbe_d$ be the orthogonal bases.
Without loss of generality, we assume $\hbtheta = \theta_1 \bbe_1$ for some $\theta_1 \ge 0$ and $\mu_x = u_1 \bbe_1 + u_2 \bbe_2$ for some $u_1, u_2 \in \Natural$.
Let 
\begin{align}
\bx = x_1 \bbe_1 + x_2 \bbe_2 + \dots + x_d \bbe_d.
\end{align}
Due to the property of the normal distribution, each coordinate $x_l$ for $l \in [d]$ is independent of each other. 
We will show the variance of 
\begin{equation}\label{ineq_vareachfeature}
 \Var\left[ x_l | \bx' \hbtheta \ge \hat{b} \right] \ge \sigma_x^2 / 10,
\end{equation}
which suffices to prove Lemma \ref{lem:normdiversity}.
\begin{itemize}
\item
For the first dimension, we have $x_1 \sim \Normal(u_1, \sigma_x^2)$ and
\begin{equation}
 \Var\left[ x_1 | \bx' \hbtheta \ge \hat{b} \right]  
 =
 \Var\left[ x_1 | x_1 \ge \hat{b}/\theta_1 \right].
\end{equation}
Applying Lemma \ref{lem:incdev} with $x = \sgn(\hat{b}/\theta_1)/\sigma_x$, $a = \mu_x/\sigma_x$, and $b = |\hat{b}/\theta_1|$ yield 
\begin{equation}
    \Var\left[ x_1 | x_1 \ge \hat{b}/\theta_1 \right] \ge \sigma_x^2/10.
\end{equation}
\item
For the second dimension, we have $x_2 \sim \Normal(u_2, \sigma_x^2)$ and
\begin{equation}
 \Var\left[ x_2 | \bx' \hbtheta \ge \hat{b} \right]  
 =
 \Var\left[ x_2 \right] = \sigma_x^2 > \sigma_x^2/10.
\end{equation}
\item $(x_3, x_4,\dots, x_d) \sim \Normal(0, \sigma_x^2 \bI_{d-2})$. In other words, these characteristics are normally distributed, and thus $\Var(x_l) = \sigma_x^2 > \sigma_x^2/10$.
\end{itemize}
In summary, we have Eq. \eqref{ineq_vareachfeature}, which concludes the proof. \qed

\begin{lem}[\citet{abbasi2011}]
\label{lem:abbasi}
Assume that $||\btheta_g|| \le S$. Let $\del > 0$ be arbitrary.
With a probability at least $1-\delta$, the true parameter $\btheta_g$ is bounded as
\begin{equation}\label{ineq:evtboundbase}
\forall n, \ \norm{\hbtheta_g(n) - \btheta_g}_{\bar{\bV}_g(n)} 
\le 
\sigma_\eps \sqrt{2\log\left(\frac{\det(\bar{\bV}_g(n))^{1/2}\det(\lambda \bI)^{-1/2}}{\delta}\right)}
+
\lambda^{1/2} S.
\end{equation}
Moreover, let $L = \max_{i,n} \norm{\bx_i(n)}_2$ and
\begin{equation}
\beta_n(L, \delta) = 
\sigma_\eps \sqrt{d \log\left(\frac{1 + nL^2 / \lambda}{\delta}\right)}
+
\lambda^{1/2} S.
\end{equation}
Then, with a probability at least $1-\delta$, 
\begin{equation}\label{ineq:evtbound}
\forall n, \ \norm{\hbtheta_g(n) - \btheta_g}_{\bar{\bV}_g(n)} \le \beta_n(L, \delta).
\end{equation}
\end{lem}

The following lemma is used in deriving a regret bound.
\begin{lem}[\citet{abbasi2011}]
\label{lem:contextsqsum}
Let $\lambda \ge 1$ and $L = \max_{n,i} \norm{\bx_i(n)}_2$.
If $\lambda \ge \max(1,L^2)$, then the following inequality holds:
\begin{equation}
\sum_{n: \iota(n)=g} \norm{\bx_{\iota(n)}}_{(\bar{\bV}_g(n))^{-1}}^2 
\le
2 L^2 \log\left( \frac{\det(\bar{\bV}_g(N))}{\det(\lambda \bI_d)} \right)
\end{equation}
for any group $g$.
\end{lem}

The following inequality is used to bound the variation of the minimum eigenvalue of the sum of characteristics (contexts).
\begin{lem}[Matrix Azuma Inequality  \citep{troppmatconc}]\label{lem:matrixazuma}
Let $\bX_1,\bX_2,\dots,\bX_n$ be adaptive sequence of $d \times d$ symmetric matrices such that $\Ex_{k-1} \bX_k = \bZero$ and $\bX_k^2 \preceq \bA_k^2$ almost surely, where $\bA \succeq \bB$ between two matrices denotes $\bA - \bB$ is positive semidefinite. Let
$%
 \sigma_A^2 \coloneqq \norm{ \frac{1}{n} \sum_k \bA_k^2 }
$%
, where the matrix norm is defined by the largest eigenvalue.
Then, for all $t \ge 0$,
\begin{equation}
\Pr\left[ \lambdamin\left(\sum_k \bX_k\right) \le t \right]
\le
d \exp(-t^2/(8n\sigma_A^2) ).
\end{equation}
\end{lem}

\paragraph*{Proof.}
The proof directly follows from Theorem 7.1 and Remark 3.10 in \citet{troppmatconc}.

The following lemma states that the selection bias makes its variance slightly ($O(1/\log K)$ times) smaller than the original variance. 
\begin{lem}[Variance of Maximum, Theorem 1.8 in \citet{ding2015}]\label{lem_varmax}
Let $x_1,\dots,x_K \in \Real$ be i.i.d.\ samples from $\Normal(0, 1)$. Let $I_{\max} = \argmax_{i \in [K]} x_i$. Then, there exists a distribution-independent constant $C_{\mathrm{varmax}} > 0$ such that 
\begin{align}\label{ineq:varmax}
\Var[ I_{\max} ] \ge \frac{C_{\mathrm{varmax}}}{\log(K)}.
\end{align}
\end{lem}

\section{Proofs}\label{sec: proofs}
\setcounter{subsection}{-1}

\subsection{Common Inequalities}\label{subsec common inequalities}

In the proofs, we often ignore the events that happen with probability $O(1/N)$. Since the expected regret per round is at most $\max_i \bx_i' \btheta_{g(i)} - \min_i \bx_i' \btheta_{g(i)}$, which is $O(1)$ in expectation, the events that happen with probability $O(1/N)$ contributes to the regret by $O(N \times 1/N) = O(1)$, which are insignificant in our analysis.

Specifically, we regard all the contexts are bounded by $\LCont{1/N} = O(\sqrt{\log N}) = \tilO(1)$ because
\begin{equation}\label{ineq:bound_context}
\Pr\left[
\forall {n \in [N], i \in I(n)}, \  
||\bx_i(n)|| \le \LCont{1/N}
\right] \ge 1 - \frac{1}{N}.
\text{\ \ \ (by Lemma \ref{lem:largestcontext})}
\end{equation}

Moreover, we also regard all the confidence bounds hold with
\begin{equation}
    \beta_n\left(\LCont{1/N}, 1/N\right) \le \beta_N\left(\LCont{1/N}, 1/N\right) = O(\sqrt{\log N}) = \tilO(1)
\end{equation}
because
\begin{equation}\label{ineq:bound_conf}
\Pr\left[
\forall {n \in [N], g \in G}, \ 
\norm{\hbtheta_g(n) - \btheta_g}_{\bar{\bV}_g(n)} \le \beta_n\left(\LCont{1/N}, \frac{1}{N}\right)
\right]
\ge 1 - \frac{|G|}{N}.
\end{equation}
follows form Eq.~\eqref{ineq:evtbound} in Lemma \ref{lem:abbasi},

Throughout the proof, we ignore the case these events do not hold.
We also denote $\LContRaw \coloneqq \LCont{1/N}$ and $\beta_N = \beta_N\left(\LContRaw, 1/N\right)$.

We next discuss the upper confidence bounds.
\begin{remark}[Bound for $\tbtheta_i$]
Let $\tbtheta_i = \argmax_{\bbtheta_{g(i)} \in \EC_{g(i)}(n;\delta)} \bx_i'\bbtheta_{g(i)}$.
By definition of $\tbtheta_i$, the following inequality always holds:
\begin{equation}\label{ineq:bound_conf_th}
\forall n, \ \norm{\tbtheta_i - \hbtheta_{g(i)}(n)}_{\bar{\bV}_g(n)} \le \beta_N.
\end{equation}
and Eq.~\eqref{ineq:bound_conf} implies 
\begin{equation}\label{ineq:bound_conf_tildeup}
\forall n, \ \bx_i' (\tbtheta_i - \btheta_{g(i)}(n)) \ge 0.
\end{equation}
Moreover, by using triangular inequality, we have
\begin{equation}
    \norm{\tbtheta_i - \btheta_g(n)}_{\bar{\bV}_g(n)} \le \norm{\tbtheta_i - \hbtheta_g(n)}_{\bar{\bV}_g(n)} + \norm{\hbtheta_g(n) - \btheta}_{\bar{\bV}_g(n)}
\end{equation}
and thus Eq.~\eqref{ineq:bound_conf} implies
\begin{equation}\label{ineq:bound_conf_ttrue}
\forall n, \ \norm{\tbtheta_i - \btheta_g(n)}_{\bar{\bV}_g(n)} \le 2\beta_N.
\end{equation}
\end{remark}

We use the calligraphic font to denote events. 
For two events $\EA, \EB$, let $\EA^c$ be a complementary event and $\{\EA, \EB\} \coloneqq \{\EA \cap \EB\}$.
We also use prime to denote events that are close to the original event. For example, event $\EA'$ is different from event $\EA$ but these two events are deeply linked.
Finally, we discuss the minimum eigenvalue. We denote $\bA \succeq \bB$ for two $d \times d$ matrices if $\bA - \bB$ is positive semidefinite: That is, $\lambdamin(\bA-\bB) \ge 0$. Note that $\lambdamin(\bA + \bB) \ge \lambdamin(\bA) + \lambdamin(\bB)$ and $\lambdamin(\bA + \bB) \ge \lambdamin(\bA)$ if $\bB \succeq \bZero$. We have $\bx \bx' \succeq \bZero$ for any vector $\bx \in \Real^d$.

\subsection{Proof of Theorem~\ref{thm_smlcand}}\label{subsec_smlcand}

We first bound regret per round $\regret(n) \coloneqq \Regret(n) - \Regret(n-1)$ in Lemma \ref{lem:regretperround}. Then, we prove Theorem \ref{thm_smlcand}.

\begin{lem}[Regret per Round]\label{lem:regretperround}
Under the laissez-faire decision rule, the regret per round is bounded as:
\begin{equation}
\regret(n) \le 2 \max_{i\in I(n)} \norm{\bx_i}_{\bar{\bV}_g^{-1}} \norm{\btheta_{g(i)} - \hbtheta_{g(i)}}_{\bar{\bV}_g}.
\end{equation}
\end{lem}

\paragraph*{Proof.}
We denote the first-best decision rule by $i^*(n) \coloneqq \argmax_{i\in I(n)} \bx_i' \btheta_{g(i)}$. Then,
\begin{align} 
\regret(n) 
&= \bx_{i^*}' \btheta_{g(i^*)} - \bx_\iota' \btheta_{g(\iota)} \\
&\le \bx_{i^*}' \left(\hbtheta_{g(i^*)} + \btheta_{g(i^*)} - \hbtheta_{g(i^*)} \right) - \bx_{\iota}' \left(\hbtheta_{g(\iota)} + \btheta_{g(\iota)} - \hbtheta_{g(\iota)} \right) \\
&\le \bx_{i^*}' \left(\btheta_{g(i^*)} - \hbtheta_{g(i^*)} \right) - \bx_{\iota}' \left(\btheta_{g(\iota)} - \hbtheta_{g(\iota)} \right) \text{\ \ (by the greedy choice of firm)} \\
&\le \norm{\bx_{i^*}}_{\bar{\bV}_{g(i^*)}^{-1}} \norm{\btheta_{g(i^*)} - \hbtheta_{g(i^*)}}_{\bar{\bV}_{g(i^*)}} + \norm{\bx_{\iota}}_{\bar{\bV}_{g(\iota)}^{-1}} \norm{\btheta_{g(\iota)} - \hbtheta_{g(\iota)}}_{\bar{\bV}_{g(\iota)}} \\
&\ \ \ \ \ \ \text{\ \ (by the Cauchy-Schwarz inequality)}\\
&\le 2  \max_{i\in I(n)} \norm{\bx_i}_{\bar{\bV}_{g(i)}^{-1}} \norm{\btheta_{g(i)} - \hbtheta_{g(i)}}_{\bar{\bV}_{g(i)}}.
\end{align}
\qed

Now, we provide the proof of Theorem~\ref{thm_smlcand}.

\paragraph*{Proof of Theorem~\ref{thm_smlcand}}
For ease of discussion, we assume $\Nzero = 0$. That is, there is no initial sampling phase. Extending our results to the case of $\Nzero > 0$ is trivial.
We first show that regardless of estimated values $\hbtheta_1$, $\hbtheta_2$, the candidate of group $2$ is drawn with constant probability. 
Let $\mu_x = ||\bmu_x||$.
Let
\begin{align}
\EM_1(n) &= \left\{\bx_1'(n) \hbtheta_1 \le 0\right\}, \\
\EM_2(n) &= \left\{\bx_2'(n) \hbtheta_2 > 0\right\}. 
\end{align}
The sign of $\bx_1' \hbtheta_1(n)$ is solely determined by the component of $\bx_1(n)$ that is parallel to $\hbtheta_1(n)$.
This component is drawn from $\Normal(\mu_{x, \parallel}, \sigma_x)$ where $\mu_{x, \parallel}$ is the component of $\bmu_x$ that is parallel to $\hbtheta_1(n)$.
Therefore, for any $\hbtheta_1$, we have\footnote{$\Pr[\EM(n)] = \Phi^c(\mu_x/\sigma_x)$ when $\mu_{x, \parallel} = \mu_x$. That is, the direction of $\bmu_x$ is exactly the same as $\hbtheta_1$. } 
\begin{equation}\label{ineq:phiunit_one}
\Pr[\EM_1(n)] \ge \Phi^c(\mu_x/\sigma_x).
\end{equation}
Likewise, for $\hbtheta_2 \ne 0$, we have\footnote{In the subsequent discussion, we do not care point mass $\hbtheta_2 = 0$ of measure zero for $N_2(n) > 0$.}
\begin{equation}\label{ineq:phiunit_two}
\Pr[\EM_2(n)] \ge \Phi^c(\mu_x/\sigma_x).
\end{equation}
Let $\EX_2(n) = \{g(\iota(n)) = g\}$ for $g \in \{1,2\}$. By using Eq. \eqref{ineq:phiunit_one} and \eqref{ineq:phiunit_two}, 
\begin{align}
\Pr[\EX_2(n)] 
&= \Pr[x_1'(n) \hat{\theta}_1 < x_2'(n) \hat{\theta}_2] \\
&\ge \Pr[x_1'(n) \hat{\theta}_1 \le 0 < x_2'(n) \hat{\theta}_2] \\
&= \Pr[\EM_1(n), \EM_2(n)] \\
&\ge \left( \Phi^c(\mu_x/\sigma_x) \right)^2. \label{ineq_phisqdraws} \text{\ \ \ \ (by Eq.~\eqref{ineq:phiunit_one}, \eqref{ineq:phiunit_two})}
\end{align}

Let $N_2^{(\EM)}(n) = \sum_{n'=1}^n \Ind[\EM_1(n'), \EX_2(n')] \le N_2(n)$. Eq. \eqref{ineq_phisqdraws} implies 
\begin{equation}
    \Ex[N_2^{(\EM)}(n)] \ge \left(  \Phi^c(\mu_x/\sigma_x) \right)^2 n.
\end{equation}

By using the Hoeffding inequality, with a probability at least $1-2/N^2$, we have
\begin{equation}\label{ineq:ntwosmallmin}
N_2^{(\EM)} \ge n \left((\Phi^c(\mu_x/\sigma_x))^2 - k\right)
\end{equation}
for $k = \sqrt{\frac{\log(N)}{n}}$.
Therefore, union bound over $n=1,2,\dots,N$ implies Eq.~\eqref{ineq:ntwosmallmin} holds with a  probability at least $1-\sum_n 2/N^2 = 1 - 2/N$.

In the following we bound the $\lambdamin(\bar{\bV}_g)$.
Note that a hiring of a worker $i_2$ under events $\EM_1(n), \EX_2(n)$ satisfies a diversity condition (Lemma \ref{lem:normdiversity}) with $\hat{b} = 0$, and we have 
\begin{equation}
\lambdamin(\Ex[\bx_\iota \bx_\iota'| \EM_1(n), \EX_2(n)]) \ge \lambda_0 
\end{equation}
with $\lambda_0 = \sigma_x^2/10$.
Using the matrix Azuma inequality (Lemma \ref{lem:matrixazuma}) for subsequence $\{\bx_\iota \bx_\iota': \EM_1(n), \EX_2(n)\}$ with $\bX = \bx_\iota \bx_\iota' - \Ex[\bx_\iota \bx_\iota']$ and $\sigma_A = 2 L^2$, for $t = \sqrt{32 N_2 \sigma_A^2} \log(dN)$, with probability $1-1/N$
\begin{align}\label{ineq:contextmin}
\lambdamin\left(\sum_{n:\iota(n)=2} \bx_\iota \bx_\iota'\right) \ge N_2^{(\EM)} \lambda_0 - t.
\end{align}

In summary, with probability $1 - 4/N$, Eq.~\eqref{ineq:ntwosmallmin} and \eqref{ineq:contextmin} hold, and then, we have
\begin{align}
\lambdamin(\bar{V}_2) 
&\ge N_2^{(\EM)} \lambda_0 - \sqrt{32 N_2 \sigma_A^2} \log(dN) \\
&\ge (n (\Phi^c(\mu_x/\sigma_x))^2 - k) \lambda_0 - \sqrt{32 N_2 \sigma_A^2} \log(dN) \\
&= n (\Phi^c(\mu_x/\sigma_x))^2  \lambda_0 - \tilO(\sqrt{n}). \label{ineq:linearntwo}
\end{align}
By using the symmetry of the two groups, exactly the same results as Eq. \eqref{ineq:linearntwo} holds for group $1$.

In the following, we bound the regret as a function of $\min_g \lambdamin(\bar{V}_g)$. 
Eq. \eqref{ineq:linearntwo} holds with probability $1 - O(1/N)$, and we ignore events of probability $O(1/N)$ that do not affect the analysis. 
The regret is bounded as 
\begin{align}\label{ineq:thm_smlcand_final}
\Regret(N) &  
\le 2 \sum_n \max_i \norm{\bx_{i}}_{\bar{\bV}_{g(i)}^{-1}} \norm{\btheta_{g(i)} - \hbtheta_{g(i)}}_{\bar{\bV}_{g(i)}} \text{\ \ \ (by Lemma \ref{lem:regretperround})} \nn
&\le 2 \sum_n \max_i \norm{\bx_{i}}_{\bar{\bV}_{g(i)}^{-1}} \beta_N \text{\ \ \ (by Eq.~\ref{ineq:bound_conf})} \\
&\le 2 \sum_n \max_i \frac{||\bx_{i}||}{\lambdamin(\bar{\bV}_{g(i)})} \beta_N \text{\ \ \ (by definition of eigenvalues)} \\
&\le 2 \sum_n \max_i \frac{L}{\lambdamin(\bar{\bV}_{g(i)})} \beta_N \text{\ \ \ (by Eq.~\eqref{ineq:bound_context})} \\
&\le 2 L \sum_n \max_i \min\left(\frac{1}{\lambdamin(\bar{\bV}_{g(i)})}, \frac{1}{\lambda}\right) \beta_N \text{\ \ \ (by $\lambdamin(\bar{\bV}_{g(i)}) \ge \lambda$)} \\
&\le 2 L \sum_n \min\left(\sqrt{\frac{1}{n (\Phi^c(\mu_x/\sigma_x))^2 \lambda_0 - \tilO(\sqrt{n})}}, \frac{1}{\lambda}\right) \beta_N \text{\ \ \ (by Eq. \eqref{ineq:linearntwo})} \\
&\le 4 L \sqrt{\frac{N}{(\Phi^c(\mu_x/\sigma_x))^2 \lambda_0}}\beta_N + \tilO(1) \\ 
&\text{\ \ \ \ \ \ $\biggl($by $\sum_{n=C^2+1}^{N} \left\{\dfrac{1}{\sqrt{n - C \sqrt{n}}}\right\} = 2 \sqrt{N} + \tilO(1)$ for $C=\tilO(1))\biggr)$} 
\end{align}
which completes Proof of Theorem \ref{thm_smlcand}. \qed

\subsection{Proof of Theorem \ref{thm_reglower}}\label{subsec_reglower}

\paragraph*{Proof.}
Since we consider $d=1$ case in this theorem, we remove bold styles in scalar variables. 
In this proof, we assume $\mu_x \theta > 0$ and $\theta > 0$. The proof for the case of $\mu_x \theta < 0$ or $\theta < 0$ is similar. 
Let $\htheta_{g,t}$ be the value of $\htheta_g$ when group $g$ candidate was chosen $t$ times. 
With a slight abuse of notation, we use $i_2 = i_2(n)$ to denote the unique candidate of group $2$ in each round $n$.
We first define the several events that characterize the perpetual underestimation. That are, 
\begin{align}
\EP &= \left\{ \left|\hat{\theta}_{2,\Nzero_2}\right| < \frac{b}{2} \theta \right\}  \\
\EP'(n) &= \left\{x_{i_2(n)} \hat{\theta}_{2,\Nzero_2} < \frac{1}{2} \mu_x \theta \right\} \\ 
\EQ &= \left\{\forall{t \ge \Nzero_1}, \  \htheta_{1,t} \ge \frac{1}{2} \theta \right\} \\ 
\EQ'(n) &= \left\{\exists i \text{ s.t. } g(i)=1, x_i \htheta_{1,N_1(n)} \ge \frac{1}{2} \mu_x \theta \right\} 
\end{align}
where $b$ is a small\footnote{We will specify $b = O(1/(\log N))$.} constant that we specify later. 
$\EP$ and $\EP'$ are about the minority whereas $\EQ$ and $\EQ'$ are about the majority:
Intuitively, Event $\EP$ states that $\htheta_2$ is largely underestimated, and $\EP'$ states that the minority candidate is undervalued.
$\EQ$ states that the majority parameter $\htheta_1$ is consistently lower-bounded, and $\EQ'$ states the stability of the best candidate of the majority after $n$ rounds.
Under laissez-faire, 
\begin{equation}
\bigcap_{n=\Nzero+1}^N (\EP'(n) \cap \EQ'(n)) 
\end{equation}
implies the majority candidate is always chosen ($g(\iota) = 1$ for all $n$), which is exactly the perpetual underestimation of Definition \ref{def:perpunderrepr}. Therefore, proving 
\begin{equation}\label{ineq_reglower_total}
\Pr\left[ \bigcap_{n=\Nzero+1}^N (\EP'(n) \cap \EQ'(n)) \right] \ge \tilO(1)
\end{equation}
concludes the proof. We bound these events by the following lemmas and finally derive Eq. \eqref{ineq_reglower_total}.

\begin{lem}\label{lem:lflarge_smltheta2}
\begin{equation}
\Pr[\EP] \ge C_1 b    
\end{equation}
for some constant $C_1$.
\end{lem}

\paragraph*{Proof.}
We denote $x_{i_2,t}$ for representing $t$-th sample of group $2$ during the initial sampling phase, which is an i.i.d.\ sample from $\Normal(\mu_x, \sigma_x^2)$.
Likewise, we also denote $y_{i_2,t} = x_{i_2,t} \theta + \eps_t$.
\begin{align}
\Pr[\EP] 
&= \Pr\left[ \left| \frac{\sum_{t=1}^{\Nzero_2} x_{i_2,t} (x_{i_2,t} \theta +  \eps_t)}{\sum_{t=1}^{\Nzero_2} x_{i_2,t}^2 + \lambda} \right| \le \frac{b}{2} \theta \right] \\
&= \Pr\left[ \left| \sum_{t=1}^{\Nzero_2} x_{i_2,t} (x_{i_2,t} \theta +  \eps_t) \right| \le \frac{b}{2} \theta \left(\sum_{t=1}^{\Nzero_2} x_{i_2,t}^2 + \lambda\right) \right] \\
&= \Pr\left[ -g(b) \le \sum_{t=1}^{\Nzero_2} x_{i_2,t} (x_{i_2,t} \theta +  \eps_t) \le g(b) \right] 
\end{align}
where
\begin{equation}
    g(b) = \frac{b}{2} \theta \left(\sum_{t=1}^{\Nzero_2} x_{i_2,t}^2 + \lambda\right).
\end{equation}

Let $x_{i_2,t} = \mu_x + e_t$. 
Define an event $\ER$ as follows.
\begin{equation}
\ER = \left\{\sum_{t=1}^{\Nzero_2} e_t^2 \le 5 \sigma_x^2 \Nzero_2\right\}
\subseteq \left\{\sum_{t=1}^{\Nzero_2} x_{i_2,t}^2 \le 2 \Nzero_2 (\mu_x^2 + 5 \sigma_x^2)  \right\}
\end{equation}
where we used $x_{i_2,t}^2 = (\mu_x + e_t)^2 \le 2 (\mu_x^2 + e_t^2)$ in the last transformation.
By using Lemma \ref{lem:concentration_chisq}, we have
\begin{equation}
\Pr[\ER^c] \le 1 - 2 e^{-2 N_2^{(0)}} \le 1/4.
\end{equation}
Moreover, let 
\begin{equation}
\ES = \left\{\sum_{t=1}^{\Nzero_2} x_{i_2,t}^2 = \sum_{t=1}^{\Nzero_2} (\mu_x + e_t)^2 \ge \Nzero_2 \mu_x^2\right\}.
\end{equation}
It is easy to confirm that $\Pr[\sum_n (\mu_x + e_t)^2 \ge \Nzero_2 \mu_x^2] \ge 1/2$, and thus 
\begin{equation}\label{ineq:cdfourth}
\Pr[\ER \cap \ES] \ge 1 - 1/4 - 1/2 = 1/4.
\end{equation}
Note that $\ES$ implies
\begin{equation}\label{ineq:gblower}
g(b) \ge \frac{b}{2} \theta \Nzero_2 \mu_x + \lambda.
\end{equation}

Conditional on $x_{i_2,t}$, we have $x_{i_2,t} \eps_t \sim \Normal(0, x_{i_2,t}^2 \sigma_\eps^2)$. 
Moreover, by using the property on the sum of independent normal random variables, 
\begin{equation}\label{ineq:epsvarcond}
\sum_t x_{i_2,t} \eps_t \sim \Normal\left(0, \sum_t x_{i_2,t}^2 \sigma_\eps^2\right).
\end{equation}
Letting 
\begin{align}
 L_R &= \frac{-g(b) - \sum_t x_{i_2,t}^2\theta}{ \sigma_\eps \sqrt{\sum_t x_{i_2,t}^2} }, \\
 U_R &= \frac{g(b) - \sum_t x_{i_2,t}^2\theta}{ \sigma_\eps \sqrt{\sum_t x_{i_2,t}^2} }, \\
 M_R &= \frac{L_R+U_R}{2} = \frac{- \left(\sqrt{\sum_t x_{i_2,t}^2}\right) \theta}{ \sigma_\eps },
\end{align}
we have
\begin{align}
\lefteqn{
\Pr\left[ -g(b) \le \sum_{t=1} (x_{i_2,t}^2 \theta + x_{i_2,t} \eps_n) \le g(b) \right] 
}\\
&\ge \Pr\left[ -g(b) \le \sum_{t=1} (x_{i_2,t}^2 \theta + x_{i_2,t} \eps_t) \le g(b), \ER, \ES\right] \\
&\ge \Pr\left[ -g(b)-\sum_{t=1} x_{i_2,t}^2 \theta \le \sum_{t=1} x_{i_2,t} \eps_t \le g(b) - \sum_{t=1} x_{i_2,t}^2 \theta, \ER, \ES\right] \\
&\ge \Pr[\ER, \ES] \min_{\{e_n: \ER, \ES\}}\left[ \int_{L_R}^{U_R} \phi(y) dy\right] \text{\ \ \ (by Eq. \eqref{ineq:epsvarcond})} \\ 
&\ge \frac{1}{4} \min_{\{e_n: \ER, \ES\}}\left[ \int_{L_R}^{U_R} \phi(y) dy\right].  \text{\ \ \ (by Eq. \eqref{ineq:cdfourth})} \label{ineq:integralthin}
\end{align}
The following bounds Eq. \eqref{ineq:integralthin}. 
The integral's bandwidth is 
\begin{equation}
U_R - L_R 
= \frac{2g(b)}{\sigma_\eps \sqrt{\sum_t x_{i_2,t}^2}} 
\ge \frac{2g(b)}{\sigma_\eps \sqrt{2 \Nzero_2 (\mu_x^2 + 5 \sigma_x^2)}}.
\text{\ \ \ (by event $\ER$)}
\end{equation}
The value of $\phi(y)$ within $[M_R-1, M_R+1]$ is at least $\phi(M_R)/e^{1/2} \ge (1/2) \phi(M_R)$. Therefore,
\begin{equation}\label{ineq:intbandwidth}
\int_{L_R}^{U_R} \phi(y) dy \ge \min\left(2, \frac{2g(b)}{\sigma_\eps \sqrt{2 \Nzero_2 (\mu_x^2 + 5 \sigma_x^2)}}\right) \times \frac{\phi(M_R)}{2}.
\end{equation}
Moreover,
\begin{align}
\phi(M_R) 
&= \frac{1}{\sqrt{2\pi}} \exp\left(-\frac{(M_R)^2}{2}\right) = \frac{1}{\sqrt{2\pi}} \exp\left(-\frac{\theta^2 \sum_{t=1}^{\Nzero_2} x_{i_2,t}^2 }{2\sigma_\eps^2}\right) \\
&\le \frac{1}{\sqrt{2\pi}} \exp\left(-\frac{2 \theta^2 \Nzero_2 (\mu_x^2 + 5 \sigma_x^2)}{2\sigma_\eps^2}\right). \text{\ \ \ (by event $\ER$)} \label{ineq:phim}
\end{align}
By using these, we have
\begin{align}
\int_{L_R}^{U_R} \phi(y) dy
&\ge \min\left(2, \frac{2g(b)}{\sigma_\eps \sqrt{2  (\mu_x^2 + 5 \sigma_x^2)}}\right) \frac{\phi\left( M_R \right)}{2}  \text{\ \ \ (by Eq.~\eqref{ineq:intbandwidth})}\\
&= \min\left(1, \frac{g(b)}{\sigma_\eps \sqrt{2 \Nzero_2 (\mu_x^2 + 5 \sigma_x^2)}}\right) \phi(M_R)\\
&= O\left(b \sqrt{\Nzero_2} \exp\left(-\frac{2 \theta^2 \Nzero_2 (\mu_x^2 + 5 \sigma_x^2)}{2\sigma_\eps^2}\right) \right). \text{\ \ \ (by Eq. \eqref{ineq:gblower}, \eqref{ineq:phim})}
\end{align}
The exponent does not depend on $b$: Given all model parameters as constant, the probability of $\EP$ is $O(b)$, which concludes the proof. \qed

The following Lemma \ref{lem:lflarge_goodtheta1} on $\EQ$ is about the stability of the mean estimator, which is widely used to prove lower bounds in multi-armed bandit problems. That is, for any $\Delta > 0$, a wide class of mean estimators $\htheta$ of $\theta$ satisfies 
\begin{equation}\label{ineq:meanstable}
\Pr\left[\bigcup_{n=1}^\infty \left(\htheta(n) \ge \theta - \Delta\right) \right] \ge C
\end{equation}
for some constant $C = C(\theta, \Delta) > 0$. Lemma \ref{lem:lflarge_goodtheta1} is a version Eq. \eqref{ineq:meanstable} for our ridge estimator.
\begin{lem}\label{lem:lflarge_goodtheta1}
There exists a constant $n$ that is independent on $N$ such that, with a warm-start of size $\Nzero_1 \ge n$,
\begin{equation}
\Pr[\EQ] \ge C_2
\end{equation}
holds with $C_2 = 1/4$.
\end{lem}

\paragraph*{Proof.}
In this proof, we use $t \ge 0$ to denote the estimator where the $t$-th sample is drawn. 
For example, $\bar{V}_{g,t} \coloneqq \bar{V}_g(n)$ of $n: N_1(n-1) = t$. 
Note that we consider $d=1$ case and $\bar{V}_{1,t} = \sum_{t'=1}^{t} x_{1,t}^2 + \lambda$. 
By martingale bound (Eq. \eqref{ineq:evtboundbase}), with probability $1-\delta$, 
\begin{equation}\label{ineq:theta1dbound}
\forall {t \ge 1},\ \ \ |\htheta_{1,t} - \theta| \sqrt{\bar{V}_{1,t}} 
\le 
\sigma_\eps \sqrt{\log\left(\frac{\bar{V}_{1,t}^{1/2}\lambda^{-1/2}}{\delta}\right)}
+
\lambda^{1/2} S.
\end{equation}
Let $\delta = 1/2$. 
It follows from $\sqrt{\log{x}} \le \sqrt{x}$ for any $x > 0$ that
\begin{equation}\label{ineq:sqrtlog}
\sqrt{\log\left(2\bar{V}_{1,t}^{1/2}\lambda^{-1/2}\right)} \le \sqrt{2\bar{V}_{1,t}^{1/2}\lambda^{-1/2}}.
\end{equation}
Therefore,
\begin{align}
|\htheta_{1,t} - \theta| &\le 
\frac{
\sigma_\eps \sqrt{\log\left(\dfrac{\bar{V}_{1,t}^{1/2}\lambda^{-1/2}}{\delta}\right)}
+
\lambda^{1/2} S
}{
\sqrt{\bar{V}_{1,t}} 
} \text{\ \ \ (by Eq. \eqref{ineq:theta1dbound})} \\
&\le 
\frac{
\sigma_\eps \sqrt{2\bar{V}_{1,t}^{1/2}\lambda^{-1/2}}
+
\lambda^{1/2} S
}{
\sqrt{\bar{V}_{1,t}} 
} \text{\ \ \ (by \eqref{ineq:sqrtlog})} \\
\end{align}
and thus 
\begin{equation}
\forall {t \ge \Nzero_1}, \ |\htheta_{1,t} - \theta| \le \frac{1}{2} |\theta|
\end{equation}
holds if 
\begin{align}
\sqrt{\bar{V}_{1,\Nzero_1}} 
\ge 
2 \theta \max\left(
\sigma_\eps \sqrt{2\bar{V}_{1,\Nzero_1}^{1/2}\lambda^{-1/2}}
,
\lambda^{1/2} S
\right)
\end{align}
whose sufficient condition for the initial sample size $\Nzero_1$ is 
\begin{equation}\label{ineq_suffnzero}
\bar{V}_{1,\Nzero_1} \ge 
\max\left[
\frac{64}{\theta^4} (\sigma_\eps^4 / \lambda^2), 
\frac{4}{\theta^2} \lambda S^2
\right].    
\end{equation}
Note that $\Pr[\bar{V}_{1,\Nzero_1} \ge \mu_x^2 \Nzero_1] \ge 1/2$.
Letting the observation noise $\sigma_\eps$ and regularizer $\lambda$ be constants, constant size of warm-start is enough to assure this bound with probability $C_2 = 1/2 \times 1/2 = 1/4$. \qed

The following lemma states that, when $\htheta_2$ is very small, the estimated quality $x_{i_2} \htheta_2$ of the minority group is likely to be small.
\begin{lem}\label{lem:lflarge_loweval2} 
There exists a constant $C_3, C_4$ that is independent of $N$ such that
\begin{equation}\label{ineq:epprime}
\Pr[\EP'(n)|\EP] \ge 1 - C_3 \exp\left(- C_4/b \right)
\end{equation}
holds.
\end{lem}

\paragraph*{Proof.}
We have
\begin{align}
\Pr[\EP'(n)|\EP] 
&\ge 1 - \Pr\left[x_{i_2}(n) \ge \frac{2}{b}\right] \\
&\ge 1 - \Phi^c\left(\frac{1}{\sigma_x}\left(\frac{2}{b} - \mu_x\right)\right) \\
&\ge 1 - \frac{1}{\sqrt{2 \pi \sigma_x^2}}
\exp\left(-
\frac{1}{\sigma_x}\left(\frac{2}{b} - \mu_x\right)
\right),
\text{\ \ \ (by Lemma \ref{lem:normpdf})}
\end{align}
where we have assumed $\left(\frac{2}{b} - \mu\right)/\sigma_x
\ge 1$ in the last transformation (which holds for sufficiently small $b$).
Eq. \eqref{ineq:epprime} holds for 
$C_3 = \frac{1}{\sqrt{2 \pi \sigma_x^2}} e^{\mu/\sigma_x}$ and 
$C_4 = 2/\sigma_x$.
\qed

\begin{lem}\label{lem:lflarge_goodcand1} 
\begin{equation}
\Pr[\EQ'(n)\mid \EQ] \ge 1 - (1/2)^{K_1}.
\end{equation}
\end{lem}
Event $\EQ'(n)$ states that all the candidates' estimated quality $x_i \htheta$ is not below mean.
Lemma \ref{lem:lflarge_goodcand1} states that the probability of $\EQ'(n)$ is exponentially small to the number of candidates. The proof of Lemma \ref{lem:lflarge_goodcand1} directly follows from the symmetry of normal distribution and independence of characteristics $\bx_i$.

\noindent\textit{Proof of Theorem \ref{thm_reglower}, continued.}
By using Lemmas \ref{lem:lflarge_smltheta2}--\ref{lem:lflarge_goodcand1}, we have
\begin{align}
\Pr[\EP] &\ge C_1 b \label{ineq:lflarge_smltheta2} \\
\Pr\left[\EQ \right] &\ge C_2 \label{ineq:lflarge_goodtheta1} \\
\Pr[\EP'(n) |\EP] &\ge 1 - C_3 \exp\left(- C_4 b \right) \label{ineq:lflarge_loweval2} \\
\Pr[\EQ'(n)|\EQ] &\ge 1 - (1/2)^{K_1}. \label{ineq:lflarge_goodcand1}
\end{align}
From these equations, the probability of perpetual underestimation is bounded as:
\begin{align}
&\Pr\left[\bigcup_n \{\iota(n)=1\} \right]\\
&\ge 
\Pr\left[\bigcup_n \{\EP'(n), \EQ'(n)\}, \EP, \EQ \right] \\
&\ge
\Pr\left[\EP\right] \Pr\left[\EQ\right] 
\Pr\left[\bigcup_n \{\EP'(n), \EQ'(n)\} \mid \EP, \EQ \right]\text{\ \ \ (by the independence of $\EP$ and $\EQ$)} \\
&\ge 
C_1 b \times
C_2 \times 
\left(1 - N C_3 \exp\left(- C_4 b \right)\right) \times 
\left( 1 - N \left(\frac{1}{2}\right)^{K_1}\right)\text{\ \ \ (by the union bound)} \label{ineq:reglower_finalform}
\end{align}
which, by letting $b = O(1/\log(N))$ and $K_1 > \log_2(N)$, is $\tilO\left(1\right)$. 
\qed

\subsection{Proof of Theorem \ref{thm:ucb}}\label{subsec_ucb}

\paragraph*{Proof.}
Let $\regret(n) = \Regret(n) - \Regret(n-1)$. 
Notice that under the UCB decision rule, 
\begin{equation}\label{ineq:ucbselect}
 \iota(n)   
= \max_{i \in I(n)} ( \bx_i' \tbtheta_i(n) ).
\end{equation}

By Lemma \ref{lem:abbasi}, with a probability at least $1- \delta$, the true parameter of group $g$ lies in $\EC_g$, and thus
\begin{equation}\label{ineq:ucbbound}
\bx_i' \tbtheta_i(n) \ge \bx_i' \btheta_g 
\end{equation}
for each $i \in I(n)$. 

Let $i^* = i^*(n) \coloneqq \argmax_{i \in I(n)} \bx_i' \btheta_{g(i)}$ be the first-best worker, and $g^* = g(i^*)$ be the group $i^*$ belongs to. The regret in round $n$ is bounded as
\begin{align}
\regret(n)
&= \bx_{i^*}' \btheta_{g^*}  - \bx_{\iota}' \btheta_{g(\iota)} \\
&\le \bx_{i^*}' \tbtheta_{i^*}  - \bx_{\iota}' \btheta_{g(\iota)} 
\text{\ \ \ (by Eq. \eqref{ineq:ucbbound})}\\
&\le 
\bx_{\iota}' \tbtheta_{\iota}  
- 
\bx_{\iota}' \btheta_{g(\iota)} \label{ineq:iotadiff}
\text{\ \ \ (by Eq. \eqref{ineq:ucbselect})}\\
&\le ||\bx_{\iota}'||_{\bar{\bV}_{g(\iota)}^{-1}} \norm{\btheta_{g(\iota)} - \tbtheta_{\iota}}_{\bar{\bV}_{g(\iota)}}
\text{\ \ \ (by the Cauchy-Schwarz inequality)} \\
&\le ||\bx_{\iota}'||_{\bar{\bV}_{g(\iota)}^{-1}} \beta_N. \text{\ \ \ (by Eq.~\eqref{ineq:bound_conf})} \label{ineq:regucbper}
\end{align}

The total regret is bounded as:
\begin{align}
\Regret(N) = \sum_n \regret(n) 
&\le \sqrt{N \sum_n \regret(n)^2} \text{\ \ \ (by the Cauchy-Schwarz inequality)}\\
&\le 2 \beta_N \sqrt{N \sum_n ||\bx_{\iota}'||^2_{\bar{\bV}_{g(\iota)}^{-1}}(n) } \\
&\le 2 \beta_N \sqrt{2 N L^2 \sum_{g \in G} \log(\det(\bar{\bV}_g(N)))} \text{\ \ (by Lemma \ref{lem:contextsqsum})} \label{ineq:subsec_ucb_factor} \\
&\le \tilO(\sqrt{N|G|})
\end{align} 
where we have used the fact that $\log(\det(\bar{\bV}_g)) = O(\log(N)) = \tilO(1)$.
\qed

\subsection{Proof of Theorem \ref{thm:ucb index subsidy scheme}}\label{subsec: ucb index subsidy rule}

\paragraph*{Proof.}
We bound the amount of total subsidy $\Subsid(N)$.
\begin{align}
\Subsid(N) 
&\coloneqq \sum_n \bx_{\iota(n)}' (\tbtheta_{\iota(n)} - \hbtheta_{g(\iota(n))}) \le \sum_n ||\bx_{\iota(n)}'||_{\bar{\bV}_{g(\iota(n))}^{-1}} \beta_N, 
\text{\ \ \ (by Eq.~\eqref{ineq:bound_conf_th})}
\end{align}
which is the same as Eq. \eqref{ineq:regucbper} and thus the same bound as regret applies.
\qed

\subsection{Proof of Theorem~\ref{thm:ucbimp}}\label{subsec:ucbimp}

\paragraph*{Proof.}
We adopt ``slot'' notation for each group. 
Group $g$ is allocated $K_g$ slots and at each round $n$, one candidate arrives for each slot. 
We use index $i \in [K]$ to denote each slot: Although $\bx_i$ at two different rounds $n$,$n'$ (= $\bx_i(n), \bx_i(n')$) represent different candidates, they are from the identical group $g = g(i)$.
In summary, we use index $i$ to represent the $i$-th slot and with a slight abuse of argument. We also call candidate $i$ to represent the candidate of slot $i$.
Note that this does not change any part of the model, and the slot notation here is  for the sake of analysis.

Under the hybrid decision rule, a firm at each round hires the candidate of the largest index. 
Letting $\ahybconst \coloneqq \ahybconst \sigma_x/\sigma_x$,
\begin{equation}
\iota(n) = \argmax_{i\in I(n)} \Uimp_i(n)
\end{equation}
where
\begin{equation}\label{eq: new UCB decision rule rep}
\Uimp_i(n) \coloneqq
\begin{cases} 
\tilde{q}_i(n) &\mbox{if } \sui_i(n) > \ahybconst \sigma_x ||\hbtheta_{g(i)}(n)||, \\
\hat{q}_i(n) & \mbox{otherwise}.
\end{cases}
\end{equation}
We also denote $\tiliota(n) = \argmax_{i \in I(n)} \bx_i' \tbtheta_{i}$. 
That is, $\tiliota$ indicates the candidate who would have been hired if we have used the standard UCB decision rule (Eq.~\eqref{eq: UCB decision rule})

The following bounds the regret into estimation errors of $\tiliota$ and $\iota$.
\begin{align}
\regret(n)
&= \bx_{i^*}' \btheta_{g^*}  - \bx_{\iota}' \btheta_{g(\iota)} \\
&\le 
\bx_{i^*}' \tbtheta_{i^*}  - \bx_{\iota}' \btheta_{g(\iota)} 
\text{\ \ \ (by Eq. \eqref{ineq:bound_conf_tildeup})}\\
&\le \bx_{\tiliota}' \tbtheta_{\tiliota}  - \bx_{\iota}' \btheta_{g(\iota)} 
\text{\ \ \ (by definition of $\tiliota$)}\\
&= \bx_{\tiliota}' \tbtheta_{\tiliota}  - \bx_{\iota}' \tbtheta_{\iota}
+ \bx_{\iota}' (\tbtheta_{\iota}-\btheta_{g(\iota)})\\
&\le \bx_{\tiliota}' (\tbtheta_{\tiliota}  -  \hbtheta_{g(\tiliota)})
+ \bx_{\iota}' (\tbtheta_{\iota}-\btheta_{g(\iota)}). \text{\ \ \ (by definition of $\iota$) } \label{ineq:iucb_twoiota} 
\end{align}

Here,
\begin{align}
\bx_{\tiliota}' (\tbtheta_{\tiliota}-\hbtheta_{g(\tiliota)})
&\le ||\bx_{\tiliota}'||_{\bar{\bV}_{g(\tiliota)}^{-1}} \norm{ \tbtheta_{\tiliota} - \hbtheta_{g(\tiliota)}}_{\bar{\bV}_{g(\tiliota)}}
\text{\ \ \ (by the Cauchy-Schwarz inequality)} \\
&\le ||\bx_{\tiliota}'||_{\bar{\bV}_{g(\tiliota)}^{-1}} \beta_N. \text{\ \ \ (by Eq. \eqref{ineq:bound_conf_th})} \\
&\le \frac{||\bx_{\tiliota}'||}{\sqrt{\lambdamin(\bar{\bV}_{g(\tiliota)})}} \beta_N \text{\ \ \ (by definition of eigenvalues)} \\
&\le \frac{L}{\sqrt{\lambdamin(\bar{\bV}_{g(\tiliota)})}} \beta_N. \text{\ \ \ (by Eq.~\eqref{ineq:bound_context})}
\label{ineq:iucb_reg_lmbmin}
\end{align}

Moreover, the estimation error of candidate $\iota$ is bounded as
\begin{align}
\bx_\iota' (\tbtheta_\iota-\btheta_{g(\iota)})
&\le ||\bx_{\iota}'||_{\bar{\bV}_{g(\iota)}^{-1}} \norm{ \tbtheta_{\iota} - \btheta_{g(\iota)}}_{\bar{\bV}_{g(\iota)}}
\text{\ \ \ (by the Cauchy-Schwarz inequality)} \\
&\le 2 ||\bx_\iota'||_{\bar{\bV}_{g(\iota)}^{-1}} \beta_N \text{\ \ \ (by Eq. \eqref{ineq:bound_conf_ttrue})} \\
&\le \frac{2||\bx_\iota'||}{\sqrt{\lambdamin(\bar{\bV}_{g(\iota)})}} \beta_N \text{\ \ \ (by definition of eigenvalues)} \\
&\le \frac{2L}{\sqrt{\lambdamin(\bar{\bV}_{g(\iota)})}} \beta_N. \text{\ \ \ (by Eq.~\eqref{ineq:bound_context})}
\label{ineq:iucb_reg_lmbmintwo}
\end{align}

Based on the above bounds, the regret is bounded as follows.
\begin{align}
\Regret(N) 
&= \sum_{n=1}^N \regret(n) \\
&\le \sum_{n=1}^N 
\left( 
\frac{2}{\sqrt{\lambdamin(\bar{\bV}_{g(\iota)})}} +
\frac{1}{\sqrt{\lambdamin(\bar{\bV}_{g(\tiliota)})}} 
\right)
L \beta_N \text{\ \ \ (by Eq.\eqref{ineq:iucb_twoiota}, \eqref{ineq:iucb_reg_lmbmin}, \eqref{ineq:iucb_reg_lmbmintwo} )}\\
&\le 2 L \beta_N \sum_{i \in [K]} \sum_{n=1}^N  \Ind[\iota=i]
\frac{1}{\sqrt{\lambdamin(\bar{\bV}_{g(i)})}}
 +
L \beta_N \sum_{i \in [K]} \sum_{n=1}^N  \Ind[\tiliota = i]
\frac{1}{\sqrt{\lambdamin(\bar{\bV}_{g(i)})}}. \label{ineq_twotildecompose}
\end{align}
Eq.~\eqref{ineq_twotildecompose} consisted of two components. The first component is the estimation error of the hired candidate $\iota$.
The second component is the estimation error of $\tiliota$, the candidate who would have hired if we had posed the UCB decision rule. The Hybrid decision rule $\iota$ can be different from the UCB decision rule $\tiliota$, which is the main challenge of deriving regret bound in the hybrid decision rule.

We first define the following events
\begin{align}
\EV_i(n) &\coloneqq \left\{(\bx_i(n))' (\tbtheta_i(n) - \hbtheta_{g(i)}(n)) \le \ahybconst \sigma_x \norm{\hbtheta_{g(i)}(n)}\right\}, \\
\EW_i(n) &\coloneqq \{ \tiliota(n) = i \}, \\
\EX_i(n) &\coloneqq \{ \iota(n) = i \}, \\
\EX_i'(n) &\coloneqq \left\{ \bx_i(n)' \hbtheta_{g(i)}(n) \ge \argmax_{j \ne i} \Uimp_j \right\} \subseteq \EX_i.
\end{align}
Event $\EV_i$ states that the candidate $i$ is not subsidized.
Event $\EW_i$ states that $i$ would have been hired if it was subsidized in the UCB decision rule. 
Event $\EX_i$ states that $i$ is hired and $\EX_i'$ states that $i$ is hired regardless of the subsidy.

The following lemma is the crux of bounding the components in Eq.~\eqref{ineq_twotildecompose}.
\begin{lem}[Proportionality]\label{lem:parity}
The following two inequalities hold.
\begin{align}
\Pr[\EX_i'] &\ge \exp(-\ahybconst^2/2) \Pr[\EW_i],\label{ineq:parity} \\ 
\Pr[\EX_i'] &\ge \exp(-\ahybconst^2/2) \Pr[\EX_i]. \label{ineq:parity2} 
\end{align}
\end{lem}
\paragraph*{Proof.}
We first prove, for any $c \in \Real$, $d>0$,
\begin{equation}\label{ineq:parity_tail}
\Pr\left[\bx_i' \hbtheta_{g(i)} \ge c \right]
\ge 
\exp(-d^2/2) \Pr\left[
\bx_i' \hbtheta_{g(i)} \ge c - d \sigma_x \norm{\hbtheta_{g(i)}}
\right].
\end{equation}
Let $x_{\parallel} \coloneqq (\bx_i' \hbtheta_{g(i)})/||\hbtheta_{g(i)}||$ be the projection of $\bx_i$ into the direction of $\hbtheta_{g(i)}$. Then, 
$\bx_i' \hbtheta_{g(i)} = x_{\parallel} ||\hbtheta_{g(i)}||$.
From the symmetry of a normal distribution, $x_{\parallel} ||\hbtheta_{g(i)}||$ is drawn from a normal distribution with its standard derivation $\sigma_x ||\hbtheta_{g(i)}||$, from which Eq. \eqref{ineq:parity_tail} follows.

Eq. \eqref{ineq:parity} follows by
letting $c = \max_{j \ne i} \Uimp_j$, $d = \ahybconst$ because 
\begin{align}
\EW_i &\subseteq \left\{\bx_i' \hbtheta_{g(i)} \ge c - d \sigma_x \norm{\hbtheta_{g(i)}}\right\}  \\
\EX_i' &\supseteq \left\{\bx_i' \hbtheta_{g(i)} \ge c\right\}
\end{align}

Eq. \eqref{ineq:parity2} also follows 
\begin{align}
\EX_i &\subseteq \left\{\bx_i' \tbtheta_{i} \ge c\right\}  \\
\EX_i' &\supseteq \left\{\bx_i' \tbtheta_{i} \ge c + d\sigma_x \norm{\hbtheta_{g(i)}}\right\}
\end{align}
and exactly the same discussion as Eq.~\eqref{ineq:parity_tail} applies for\footnote{Note that $\bx_i' \hbtheta_{g(i)}$ in Eq.~\eqref{ineq:parity_tail} is replaced by $\bx_i' \tbtheta_{i}$ in Eq.~\eqref{ineq:parity_tail_ttheta}, which does not change the subsequent derivations at all.
}
\begin{equation}\label{ineq:parity_tail_ttheta}
\Pr\left[\bx_i' \tbtheta_{i} \ge c + d \sigma_x \norm{\hbtheta_{g(i)}}\right]
\ge 
\exp(-d^2/2) \Pr\left[
\bx_i' \tbtheta_{i} \ge c 
\right].
\end{equation}
\qed
Lemma \ref{lem:parity} is intuitively understood as follows. 
Assume that candidate $i$ would have been hired under the UCB rule. The candidate may not be hired under the hybrid rule because it can cut subsidies for that candidate. However, there is a constant probability such that a slightly better (``$a $-good'') candidate appears on slot $i$, and such a candidate is hired under the hybrid rule.

The following two lemmas, which utilizes Lemma \ref{lem:parity}, bounds the two terms of Eq.~\eqref{ineq_twotildecompose}.
\begin{lem}\label{lem_iucb_iota}
\begin{equation}
\Ex\left[
\sum_{n=1}^N  \Ind[\iota=i]
\frac{1}{\sqrt{\lambdamin(\bar{\bV}_{g(i)})}}
\right]
\le 
\frac{2 e^{\ahybconst^2/4}}{\lambda_0} \sqrt{N} + O(1).
\end{equation}
\end{lem}
\begin{lem}\label{lem_iucb_tiliota}
\begin{equation}
\Ex\left[
\sum_{n=1}^N  \Ind[\tiliota = i]
\frac{1}{\sqrt{\lambdamin(\bar{\bV}_{g(i)})}}
\right]
\le 
\frac{2 e^{\ahybconst^2/4}}{\lambda_0} \sqrt{N} + O(1).
\end{equation}
\end{lem}

With Lemmas~\ref{lem_iucb_iota} and \ref{lem_iucb_tiliota}, the regret is bounded as
\begin{align}
\Regret(N) 
&\le 2 L \beta_n(L, 1/N) \sum_{i \in [K]} \sum_{n=1}^N  \Ind[\iota=i]
\frac{1}{\sqrt{\lambdamin(\bar{\bV}_{g(i)})}}
\nn&\ \ \ +
L \beta_n(L, 1/N) \sum_{i \in [K]} \sum_{n=1}^N  \Ind[\tiliota = i]
\frac{1}{\sqrt{\lambdamin(\bar{\bV}_{g(i)})}} \text{\ \ \ (by Eq.~\eqref{ineq_twotildecompose})}\\
&\le 6 L \beta_n(L, 1/N) K \frac{e^{\ahybconst^2/4} \sqrt{N}}{\lambda_0 } + \tilO(1)
\text{\ \ \ (by Lemma \ref{lem_iucb_iota} and \ref{lem_iucb_tiliota})} \label{ineq:ucbimp_const}
\end{align}
which completes the proof of Theorem~\ref{thm:ucbimp}.
\qed %

The following is the proof of Lemma~\ref{lem_iucb_iota}.

\paragraph*{Proof.}
Let $N_i(n)$ be the number of the rounds before $n$ such that the worker of slot $i$ is selected.
Let $\tau_t$ be the first round such that $N_i(n)$ reaches $t$ and $N_{i,t} = \sum_{n \le \tau_t} \Ind[\EX_i'(n)]$. 
Lemma \ref{lem:parity} implies $\Ex[N_{i,t}] \ge e^{-\ahybconst^2/2} t$ and applying the Hoeffding inequality on binary random variables $(\Ind[\EX_i'(\tau_1)],\Ind[\EX_i'(\tau_2)],\dots,...,\Ind[\EX_i'(\tau_t)])$ yields
\begin{equation}\label{ineq:bdhoeffding}
\Pr\left[
N_{i,t} <
\left( 
e^{-\ahybconst^2/2} t - \sqrt{(\log N)t}
\right)
\right] \le \frac{2}{N^2}.
\end{equation}
By using this, we have
\begin{align}
\lefteqn{
\Pr\left[
\bigcap_{t=1}^N \left\{ N_{i,t} <
\left( 
e^{-\ahybconst^2/2} t - \sqrt{(\log N)t}
\right)
\right\}
\right]
}\nn
&\le 
\sum_t \Pr\left[
N_{i,t} <
\left( 
e^{-\ahybconst^2/2} t - \sqrt{(\log N)t}
\right)
\right]
\text{\ \ \ (by union bound)}\nn
&\le
\sum_t \frac{2}{N^2} = \frac{2}{N}.
\text{\ \ \ (by Eq.~\eqref{ineq:bdhoeffding})}\nn
\end{align}
In the following, we focus on the case 
\begin{equation}\label{ineq:draw_ew_enough}
N_{i,t} \ge e^{-\ahybconst^2/2} t - \sqrt{(\log N)t},
\end{equation}
which occurs with a probability at least $1-2/N$.

Let $\bar{\bV}_i(n) \coloneqq \sum_{n'\le n} \Ind[\iota = i] \bx_i \bx_i' \preceq \bar{\bV}_{g(i)}(n)$. The context $\bx_i$ conditioned on event $\EX_i'$ satisfies assumptions in Lemma \ref{lem:normdiversity} with $\hbtheta = \tbtheta_i$ and $\hat{b} = \max_{j \ne i} \Uimp_j$.
We have,
\begin{align}
\sum_{n=1}^N  \Ind[\iota = i]
\frac{1}{\sqrt{\lambdamin(\bar{\bV}_{g(i)})}}
&\le 
\sum_{n=1}^N \Ind[\iota = i]
\frac{1}{\sqrt{\lambdamin(\bar{\bV}_i)}}
\text{\ \ \ (by $\bar{\bV}_{g(i)} \succeq \bar{\bV}_i$)} \nn
&\le 
\sum_{n=1}^N \sum_{t=1}^N  \Ind[\iota = i, N_i(n) = t]
\frac{1}{\sqrt{\lambdamin(\bar{\bV}_i)}} \nn
&\text{\ \ \ (by $N_i(N) \le N$)} \nn
&\le 
\sum_{t=1}^N 
\frac{1}{\sqrt{\lambdamin(\bar{\bV}_i(\tau_t))}}. \nn
&\text{\ \ \ (by $\Ind[\iota = i, N_i(n) = t]$ occurs at most once)} 
\end{align}
In other words, lower-bounding $\lambdamin(\bar{\bV}_i(\tau_t))$ suffices the regret bound, which we demonstrate in the following.
We have 
\begin{align}
\Ex\left[
\lambdamin(\bar{\bV}_i(\tau_t))
\right]
&\ge \lambdamin(\sum_n \Ex[\Ind[\EX_i'(n)] \bx_i \bx_i'] ) \ge \lambda_0
N_{i,t}.
\text{\ \ \ (by Lemma \ref{lem:normdiversity})}
\end{align}
By using the matrix Azuma inequality (Lemma \ref{lem:matrixazuma}), with probability of at least $1-1/N$
\begin{equation}\label{ineq:iota_lmdmin}
\lambdamin(\bar{\bV}_i(\tau_t)) \ge 
\left(\lambda_0 N_{i,t} - \sqrt{32 N_{i,t} \sigma_A^2} \log(dN)\right)
\end{equation}
where $\sigma_A = 2L^2$.
By using Eq. \eqref{ineq:draw_ew_enough}, \eqref{ineq:iota_lmdmin}, we have
\begin{equation}
\lambdamin(\bar{\bV}_i(\tau_t)) \ge \lambda_0 e^{-\ahybconst^2/2} t -  O(\sqrt{t})
\end{equation}
and thus
\begin{align}
\sum_{t=1}^N 
\frac{1}{\sqrt{\lambdamin(\bar{\bV}_i(\tau_t))}}
&\le 
\sum_{t=1}^N 
\frac{1}{\sqrt{
\lambda_0 e^{-\ahybconst^2/2} t - O(\sqrt{t})
}} \le 
\frac{2  e^{\ahybconst^2/4}}{\lambda_0} \sqrt{N} + O(1).
\end{align}
\qed

The following is the proof of Lemma~\ref{lem_iucb_tiliota}.

\paragraph*{Proof.}
Let $N_i^{\EW_i}(n) = \sum_{n'\le n} \Ind[\EW_i]$ and 
let $\tau_t$ be the first round such that $N_i^{\EW_i}(n)$ reaches $t$ and $N_{i,t} = \sum_{n \le \tau_t} \Ind[\EX_i'(n)]$.
The following discussions are very similar to the one of Lemma \ref{lem_iucb_iota}, which we write for completeness.
Then, we have
\begin{align}
\lefteqn{
\Pr\left[
\bigcap_{t=1}^N \left\{ N_{i,t} <
\left( 
e^{-\ahybconst^2/2} t - \sqrt{(\log N)t}
\right)
\right\}
\right]
}\nn
&\le 
\sum_t \Pr\left[
N_{i,t} <
\left( 
e^{-\ahybconst^2/2} t - \sqrt{(\log N)t}
\right)
\right]
\text{\ \ \ (by union bound)}\nn
&\le
\sum_t \frac{2}{N^2}
= \frac{2}{N}.
\text{\ \ \ (by Lemma \ref{lem:parity} and the Hoeffding inequality)} 
\end{align}
In the following, we focus on the case
\begin{equation}\label{ineq:draw_ew_enough_tiliota}
N_{i,t} \ge e^{-\ahybconst^2/2} t - \sqrt{(\log N)t}
\end{equation}
that occurs with a probability at least $1-2/N$.

We have,
\begin{align}
\sum_{n=1}^N  \Ind[\tiliota = i]
\frac{1}{\sqrt{\lambdamin(\bar{\bV}_{g(i)})}}
&\le 
\sum_{n=1}^N \Ind[\tiliota = i]
\frac{1}{\sqrt{\lambdamin(\bar{\bV}_i)}}
\text{\ \ \ (by $\bar{\bV}_{g(i)} \succeq \bar{\bV}_i$)} \nn
&\le 
\sum_{n=1}^N \sum_{t=1}^N  \Ind[\tiliota = i, N_i^{\EW_i}(n) = t]
\frac{1}{\sqrt{\lambdamin(\bar{\bV}_i)}} \nn
&\le 
\sum_{t=1}^N 
\frac{1}{\sqrt{\lambdamin(\bar{\bV}_i(\tau_t))}}. \nn
&\text{\ \ \ (by $\{\tiliota = i\}$ increments $N_i^{\EW_i}$)} 
\end{align}
The following lower-bounds $\lambdamin(\bar{\bV}_i(\tau_t))$.

We have 
\begin{align}
\Ex\left[
\lambdamin(\bar{\bV}_i(\tau_t))
\right]
\ge \lambdamin\left(\sum_n \Ex[\Ind\left[\EX_i'(n)] \bx_i \bx_i'\right] \right) \ge \lambda_0
N_{i,t}.
\text{\ \ \ (by Lemma \ref{lem:normdiversity})}
\end{align}
By using the matrix Azuma inequality (Lemma \ref{lem:matrixazuma}), at least $1-1/N$
\begin{equation}\label{ineq:tiliota_lmdmin}
\lambdamin(\bar{\bV}_i(\tau_t)) \ge 
\left(\lambda_0 N_{i,t} - \sqrt{32 N_{i,t} \sigma_A^2} \log(dN)\right)
\end{equation}
where $\sigma_A = 2(L(1/N))^2$.
By using Eq. \eqref{ineq:draw_ew_enough_tiliota}, \eqref{ineq:tiliota_lmdmin}, we have
\begin{equation}
\lambdamin(\bar{\bV}_i(\tau_t)) \ge \lambda_0 e^{-\ahybconst^2/2} t -  O(\sqrt{t})
\end{equation}
and thus
\begin{align}
\sum_{t=1}^N 
\frac{1}{\sqrt{\lambdamin(\bar{\bV}_i(\tau_t))}}
&\le 
\sum_{t=1}^N 
\frac{1}{\sqrt{
\lambda_0 e^{-\ahybconst^2/2} t - O(\sqrt{t})
}} \le 
\frac{2  e^{\ahybconst^2/4}}{\lambda_0} \sqrt{N} + O(1).
\end{align}

We here bound the amount of the subsidy.  Eq.~\eqref{ineq:iucb_reg_lmbmin}, \eqref{ineq:iucb_reg_lmbmintwo}
imply
\begin{align}
\bx_i' \left(\tbtheta_{i} - \hbtheta_{g(i)}\right) 
&\le \frac{1}{\sqrt{\lambdamin(\bar{\bV}_{g(i)})}} L \beta_N \nn
\left|\bx_i' \hbtheta_{g(i)} - \btheta_{g(i)}\right|
&\le 2 \frac{1}{\sqrt{\lambdamin(\bar{\bV}_{g(i)})}} L \beta_N
\end{align}
and thus the subsidy $\shi_i(n) = 0$ for 
\begin{equation}
\lambdamin(\bar{\bV}_{g(i)}) 
\ge 
\left(\frac{2 L \beta_N}{\norm{\btheta}}\right)^2
\max\left(1, \frac{1}{{\ahyb}^2} \right) 
\eqqcolon C_s = \tilO(1). \label{ineq:subsid_cutoff}
\end{equation}
Hence, it follows that
\begin{align}
\Subsid(N) 
&= \sum_n \shi_\iota(n) \le \sum_n \sum_i \Ind[\EX_i] \shi_\iota(n) \nn
&\le L \beta_N \sum_i \sum_n  \Ind[\lambdamin(\bar{\bV}_{g(i)}) 
 \le C_s] \frac{1}{\sqrt{\lambdamin(\bar{\bV}_{g(i)})}} 
\text{\ \ \ (by Eq.~\eqref{ineq:subsid_cutoff})} \nn
&\le L \beta_N \sum_i \sum_t \Ind[\lambda_0 e^{-\ahybconst^2/2} t - O(\sqrt{t})
 \le C_s] \frac{1}{\sqrt{\lambda_0 e^{-\ahybconst^2/2} t - O(\sqrt{t})}} \nn
&\text{\ \ \ \ \ (by the same discussion as Lemma \ref{lem_iucb_iota})} \nn
&\le L \beta_N K \sum_t 
\Ind[\lambda_0 e^{-\ahybconst^2/2} t \le C_s] 
\frac{1}{\sqrt{\lambda_0 e^{-\ahybconst^2/2} t}} + \tilO(1)\nn
&\le L \beta_N K \frac{2 e^{\ahybconst^4/2}}{\lambda_0} \sqrt{\frac{C_s e^{\ahybconst^2/2}}{\lambda_0}} + \tilO(1) 
= \tilO(1). \label{ineq:ucbimp_subsid_const}
\end{align}
\qed %

Note that $C_s$ diverges as $a \rightarrow +0$. The bound of the subsidy is meaningful for $a > 0$. If $a=0$, the hybrid mechanism is reduced to the UCB mechanism, and thus Theorem \ref{thm:ucb index subsidy scheme} for UCB applies.

\subsection{Proof of Theorem \ref{thm_2SLF}}\label{subsec: proof of thm_2SLF}

We modify the proof of Theorem~\ref{thm_reglower}. Accordingly, unless we explicitly mention it, we use the same notation as the proof of Theorem~\ref{thm_reglower}.

We define
\begin{equation}
    \EQ''(n) = \left\{\exists i^A,i^B \text{ s.t. } g(i^A) = g(i^B) = 1, i^A \neq i^B,\text{ and } x_i \htheta_{1,N_1(n)} \ge \frac{1}{2} \mu_x \theta \text{ for }i = i^A,i^B \right\}.
\end{equation}
When the event $\EQ''(n)$ occurs, there are two majority workers whose estimated skill $\hat{q}_i(n)$ is larger than its mean.

\begin{lem}\label{lem:second order statistic}
\begin{equation}\label{ineq:second order statistic}
\Pr[\EQ''(n)|\EQ] \ge 1 - \left(K_1 + 1\right)\left(\dfrac{1}{2}\right)^{K_1}.
\end{equation}
\end{lem}
Event $\EQ''(n)$ states that the second order statistics of $\{\hat{q}_i\}_{i: g(i)=1}$ is below mean.
Lemma~\ref{lem:second order statistic} states that this event is exponentially unlikely to $K_1$. By the symmetry of normal distribution and independence of characteristics $\bx_i$, each candidate is likely to be below mean with probability $1/2$, and the proof of Lemma~\ref{lem:second order statistic} directly follows by counting the combinations such that at most one of the worker(s) are above mean.

When we have $\EP'(n)$ and $\EQ''(n)$ for all $n$, then for every round $n$, the top-$2$ workers in terms of quality $\hat{q}_i(n)$ are from the majority. In this case, the minority worker is not hired regardless of the additional signal $\eta_i$. Accordingly, this is a sufficient condition for a perpetual underestimation.

The following is the proof of Theorem~\ref{thm_2SLF}.

\paragraph*{Proof.}
By using Lemmas~\ref{lem:lflarge_smltheta2}, \ref{lem:lflarge_goodtheta1}, \ref{lem:lflarge_loweval2}, and \ref{lem:second order statistic}, we have \eqref{ineq:lflarge_smltheta2}, \eqref{ineq:lflarge_goodtheta1}, \eqref{ineq:lflarge_loweval2}, and \eqref{ineq:second order statistic}. From these equations, the probability of perpetual underestimation is bounded as:
\begin{align}
&\Pr\left[\bigcup_n \{\iota(n)=1\} \right]\\
&\ge 
\Pr\left[\bigcup_n \{\EP'(n), \EQ''(n)\}, \EP, \EQ \right] \\
&\ge
\Pr\left[\EP\right] \Pr\left[\EQ\right] 
\Pr\left[\bigcup_n \{\EP'(n), \EQ''(n)\} \mid \EP, \EQ \right]\text{\ \ \ (by the independence of $\EP$ and $\EQ$)} \\
&\ge 
C_1 b \times
C_2 \times 
\left(1 - N C_3 \exp\left(- C_4 b \right)\right) \times 
\left( 1 - N \left(K_1 + 1\right)\left(\dfrac{1}{2}\right)^{K_1}\right)\text{\ \ \ (by the union bound)},
\end{align}
which, by letting $b = O(1/\log(N))$ and $K_1 + \log_2 (K_1 + 1) \ge \log_2 N$, is $\tilO\left(1\right)$. 
\qed

\subsection{Proof of Theorem \ref{thm:rooney_main}}\label{subsec:Rooney_main proof}

\paragraph*{Proof.} 
We have
\begin{align} 
|\bx_i' (\hbtheta_g - \btheta_g)| 
&\le ||\bx_i||_{\bar{\bV}_g^{-1}} \norm{\hbtheta_g - \btheta_g}_{\bar{\bV}_g} \\
&\le \frac{L}{\lambdamin(\bar{\bV}_g)} \beta_n \text{\ \ \ (by Eq.~\eqref{ineq:bound_context} and \eqref{ineq:bound_conf})}\\
&\le \frac{L}{\lambda} \beta_N \text{\ \ \ (by $\bar{\bV}_g \succeq \lambda \bI_d$)} \nn
&\eqqcolon C_5 = \tilO(1). \label{ineq:lmbound}
\end{align}

Let $i_1$ and $i_2$ be the finalists chosen from group $1$ and $2$, respectively. Let
\begin{equation}
\EJ(n) = \{ \eta_{i_1}(n) - \eta_{i_2}(n) > 2 C_5 \}. 
\end{equation}
 Under $\EJ$, the finalist of group $1$ is chosen because Eq. \eqref{ineq:lmbound} implies that $|\bx_{i_1}' \hbtheta_{g_1} - \bx_{i_2}' \hbtheta_{g_2}| \le 2 C_5$ and thus $\bx_{i_1}' \hbtheta_{g_1} + \eta_{i_1} - \bx_{i_2}' \hbtheta_{g_2} + \eta_{i_2} > 0$.
Note that $\eta_{i_1} - \eta_{i_2}$ is drawn from $\Normal(0, 2 \sigma_\eta^2)$.
Let $C_6 = \Phi^c(\sqrt{2} C_5 / \sigma_\eta)$. Then,
\begin{equation}\label{ineq:evt_twoststrong}
\Pr[\EJ(n)] = C_6.
\end{equation}

Let $N_1^{\EJ} = \sum_{n'=1}^{n-1} \Ind[g(\iota) = 1, \EJ] \le N_1(n)$ be the number of hiring of group $1$ under event $\EJ$. 
By using the Hoeffding inequality, with probability $1-1/N^2$ we have
\begin{align}\label{ineq:ngminrooney}
N_1^{\EJ} \ge n C_6 - \sqrt{n \log(N)}.
\end{align}
By taking union bound, Eq.~\eqref{ineq:ngminrooney} holds for all $n$ with probability $1 - \sum_n 1/N^2 \ge 1 - 1/N$.
From now, we evaluate $\lambdamin\left(\bar{\bV}_1(n)\right)$. It is easy to see that 
\begin{align} 
\bar{\bV}_1 
&\coloneqq \sum_{n'=1: \iota(n') = g}^n \bx_{i_1} \bx_{i_1}' + \lambda I \ge \sum_{n'=1: \iota(n') = g}^n \bx_{i_1} \bx_{i_1}' \ge \sum_{n'=1: \EJ}^n \bx_{i_1} \bx_{i_1}'.
\end{align}
In the following, we lower-bound the quantity
\begin{equation} \label{ineq:rooneylmbmin}
\lambdamin(\Ex[\bx_i \bx_i' | \EJ]) 
\ge \min_{\bv: ||\bv||=1} \lambdamin(\Var[\bv' \bx_i | \EJ]).
\end{equation}
Note that $i_1 = \argmax_{i: g(i)=1} \bx_i' \hbtheta_1$ is biased toward the direction of $\hbtheta_1$, and we cannot use the diversity condition (Lemma \ref{lem:normdiversity}).
Let $\bv_{\parallel}$ and $\bv_{\perp}$ be the component of $\bv$ that is parallel to and perpendicular to $\hbtheta_1$ (i.e., $||\bv_{\parallel}||^2 + ||\bv_{\perp}||^2 = 1$).
It is easy to confirm that $\Var[\bv_{\perp}' \bx_i] = ||\bv_{\perp}||^2 \sigma_x^2$ because selection of $\argmax_i \bx_i' \hbtheta_g$ does not yield any bias in perpendicular direction. 
Regarding $\bv_{\parallel}$, Lemma \ref{lem_varmax} characterize the variance, which is slightly ($O(1/\log K)$) smaller than the original variance due to biased selection. That is,
\begin{equation}
\min_{\bv: ||\bv||=1} \lambdamin(\Var[\bv' \bx_i | \EJ])  \ge \sigma_x \left( \frac{C_{\mathrm{varmax}}}{\log(K)} ||\bv_{\parallel}||^2 + ||\bv_{\perp}||^2 \right) \ge \sigma_x \frac{C_{\mathrm{varmax}}}{\log(K)}.
\end{equation}
By using the matrix Azuma inequality (Lemma \ref{lem:matrixazuma}) with $\sigma_A = 2 L^2$, for $t = \sqrt{32 N_1^{\EJ} \sigma_A^2} \log(dN)$, with probability $1-1/N$
\begin{equation}\label{ineq:lmdminrooney}
\lambdamin(\bar{\bV}_1) \ge \sigma_x \frac{C_{\mathrm{varmax}}}{\log(K)} N_g^{\EJ} - t.
\end{equation}

Combining Eq. \eqref{ineq:ngminrooney} and \eqref{ineq:lmdminrooney}, with a probability at least $1-2/N$, we have
\begin{equation}\label{ineq:rooney_lmdmin_sqrt}
\lambdamin(\bar{\bV}_1(n)) \ge \sigma_x \frac{C_p}{\log(K)} n - \tilO(\sqrt{n})
\end{equation}
where $C_p = C_6 C_{\mathrm{varmax}} = \tilO(1)$.
By symmetry, exactly the same bound as Eq.~\eqref{ineq:rooney_lmdmin_sqrt} holds for group $2$.
Finally, by using similar transformations as Eq.~\eqref{ineq:thm_smlcand_final}, the regret is bounded as 
\begin{align}
\Ex[\Regret(N)] 
&\le 2 \sum_{n=1}^N \max_{i \in [K]} \left|\bx_i'(n) (\hbtheta_g - \btheta_g)\right| \\
&\le 2 \sum_{n=1}^N \frac{L}{\sqrt{\lambdamin(\bar{\bV}_{g})}} \beta_N
\text{\ \ \ (by Eq.~\eqref{ineq:bound_context}, \eqref{ineq:bound_conf})}\\ 
&\le 2L\beta_N \sum_{n=1}^N \sqrt{ \frac{\log(K)}{\sigma_x C_p n - \tilO(\sqrt{n})} } 
\text{\ \ \ (by Eq.~\eqref{ineq:rooney_lmdmin_sqrt})} \\ 
&\le 4 L \beta_N \sqrt{\frac{N \log(K)}{\sigma_x C_p}} + \tilO(1) 
= \tilO(\sqrt{N}). \label{ineq:thm_rooney_form}
\end{align}
\qed

\section{Sublinear Regret and Other Fairness Criteria}\label{sec: fairness criteria}

In this section, we analyze the relationship between sublinear regret social learning (under which per-round expected regret converges to zero) and fairness notions prevalent in fair machine learning literature. For clarity, we invoke Assumption~\ref{assp:twogroups}, specifying that group $1$ is the majority (unprotected) group and group $2$ is the minority (protected) group.

The literature has proposed and analyzed many different fairness notions \citep[see, e.g., a survey by][]{MakhloufZP21}. The most frequently discussed measures of fairness are \emph{demographic parity} and \emph{equalized odds}, highlighting their prominence in scholarly discourse.

\subsection{Sublinear Regret}

Recall that a decision rule $\iota$ is said to have \emph{sublinear regret} if $\Ex[\Regret(N)]=O(N^a)$ for some $a < 1$. In the machine learning literature, a policy is said to have sublinear regret if the per-round expected regret approaches zero as the number of rounds ($N$) increases. The sublinear-regret property indeed requires that $\Ex[\Regret(N)]/N \to 0$ as $N \to \infty$. This implies that the decision rule eventually makes an unbiased decision in the sense that each firm hires based on accurately estimated skill predictors, $q_i$, uninfluenced by the workers' group affiliation. Consequently, sublinear regret implies asymptotic unbiasedness in firms' decision-making processes.

The main part of this paper has evaluated the regret order of various decision rules. In summary, the laissez-faire decision rule exhibits sublinear regret in a balanced population scenario (Theorem~\ref{thm_smlcand}), whereas it encounters significant regret even in the long run in scenarios with unbalanced populations (Theorem~\ref{thm_reglower}). By contrast, the UCB and hybrid decision rules demonstrate sublinear regret even with unbalanced populations (Theorems~\ref{thm:ucb} and \ref{thm:ucbimp}). Although nearly all these findings are confirmed for symmetric groups (i.e., Assumption~\ref{assp:idcontext} is assumed), the sublinear-regret property of the UCB decision rule remains independent of this assumption.

\subsection{Equalized Odds}

Equalized odds, defined below, is the fairness notion the most directly related to sublinear regret.

\begin{definition}[Equalized Odds \citep{HardtPNS16}]\label{defn: equalized odds}
A decision rule asymptotically aligns with \emph{equalized odds} if, for any $\eps > 0$, there exists a number $N_0$ such that for all $N \ge N_0$,
\begin{align}
&\frac{1}{N}
\sum_{n \le N}
\left(
\left| 
\Prob[
\iota(n) = i | i^*(n) = i, g(i) = 1
] 
-
\Prob[
\iota(n) = i | i^*(n) = i, g(i) = 2
] 
\right|
\right)
< \eps, \text{ and}
\\
&\frac{1}{N}
\sum_{n \le N}
\left(
\left| 
\Prob[
\iota(n) = i | i^*(n) \ne i, g(i) = 1
] 
-
\Prob[
\iota(n) = i | i^*(n) \ne i, g(i) = 2
] 
\right|
\right)
< \eps,
\end{align}
where $i^*(n) := \argmax_{i\in I(n)} q_i$. 
The probability is marginalized on the candidate $i$ as well as the other candidates on the same round $n$. 
\end{definition}

Equalized odds requires that the hiring practices of firms perform equitably across various groups. This implies that given a worker possesses the highest skill predictor $q_i$ for a given round, the probability of her being hired remains independent of her group affiliation. Considering we are dealing with a social learning problem, imposing this condition for all rounds $t=1,2,\dots$ would be overly restrictive; hence, we only enforce it as an asymptotic condition.

The following theorem characterizes the relationship between sublinear regret and equalized odds.

\begin{thm}[Sublinear Regret Implies Equalized Odds]\label{thm: sublinear regret implies equalized odds}
Suppose Assumptions~\ref{assp:twogroups} and \ref{assp:normalcontext}. If a decision rule has sublinear regret, then it asymptotically aligns with equalized odds.
\end{thm}

The proof is in Appendix~\ref{subsec: proof of sublinear regret theorems}. The intuition is as follows. The violation of equalized odds leads to persistent biased decisions by firms. Consequently, society experiences enduring, non-diminishing regret, which signifies a failure in sublinear-regret learning. Therefore, the realization of sublinear regret inherently necessitates the fulfillment of equalized odds.

\subsection{Demographic Parity}

Next, we discuss another fairness notion called demographic parity. Demographic parity is defined as follows.

\begin{definition}[Demographic Parity]\label{defn: demographic parity}
A decision rule asymptotically aligns with \emph{demographic parity} if, for any $\eps > 0$, there exists a number $N_0$ such that for all $N \ge N_0$,
\begin{equation}
\frac{1}{N}
\sum_{n \le N}
\left(
\left| 
\Prob[
\iota(n) = i | g(i) = 1
] 
-
\Prob[
\iota(n) = i | g(i) = 2
] 
\right|
\right) < \eps.
\end{equation}
\end{definition}

Demographic parity necessitates the hiring probability to be indifferent to affiliation with the minority group. When this criterion is met, the proportion of hired workers from group $g$ aligns with the overall population ratio of group-$g$ workers. Analogous to Definition~\ref{defn: equalized odds}, the imposition of this condition for all rounds would be excessively restrictive; hence, we establish it as an asymptotic condition.

It is widely recognized that numerous fairness notions in machine learning are often at odds with each other and coexist only within very limited contexts \citep[for more details, see][]{KleinbergMR17}. Demographic parity and equalized odds are prime examples of this conflict. Demographic parity necessitates equal treatment across groups, ignoring the individual skills of each worker. In contrast, equalized odds demands that firms hire proficient workers based on their skills, uninfluenced by their group affiliations. It is evident that these two objectives are incompatible when the distribution of workers' skills differs among groups. 
Since the sublinear-regret principle aligns with equalized odds, it generally contradicts demographic parity within general environments.

This paper does not aim to adjudicate between conflicting fairness notions, which cannot be simultaneously satisfied. Therefore, the core discussion assumes that groups have no significant disparities beyond their sizes (Assumption~\ref{assp:idcontext}). Under this assumption, the successful hiring of the most skilled workers will naturally result in the hired group's composition reflecting the population ratio. Consequently, sublinear regret, equalized odds, and demographic parity align harmoniously. The theorem presented subsequently formalizes this assertion.

\begin{thm}[Sublinear Regret Implies Demographic Parity with Symmetric Groups]\label{thm: sublinear regret implies demographic parity}
Suppose Assumptions~\ref{assp:twogroups}, \ref{assp:idcontext}, and \ref{assp:normalcontext}. If a decision rule has sublinear regret, then it asymptotically aligns with demographic parity.
\end{thm}

The proof is in Appendix~\ref{subsec: proof of sublinear regret theorems}. Theorem~\ref{thm: sublinear regret implies demographic parity} elucidates that in instances of group symmetry, equalized odds and demographic parity are compatible, eliminating the contention over the choice of fairness notions. Moreover, these two properties' satisfaction is guaranteed by the principle of sublinear regret. Hence, regret has been employed as an outcome measure under the assumption of group symmetry.

While our discussion has centered on equalized odds and demographic parity, numerous other fairness notions have been developed within the machine-learning literature. It is anticipated that many of these notions align within symmetric groups, but typically exhibit conflicts with sublinear regret and equalized odds in more general environments.

\subsection{Proofs}\label{subsec: proof of sublinear regret theorems}

We first introduce several technical lemmas and then prove Theorems~\ref{thm: sublinear regret implies equalized odds} and \ref{thm: sublinear regret implies demographic parity}.

\begin{lem}[Small Gap]\label{lem_propdist}
Suppose Assumptions~\ref{assp:twogroups} and Assumption \ref{assp:normalcontext}.
Let the minimal gap among $K$ candidates in round $n$ for a fixed $n$.
\begin{equation}
\Delta(n) \coloneqq \min_{i \ne j} 
|\bx_i' \btheta_{g(i)} 
- \bx_j' \btheta_{g(j)}
|.
\end{equation}
Then, there exists a constant $C >0$ such that, for any sufficiently small $\eps>0$
\begin{equation}\label{ineq_propdist}
\Prob[\Delta(n) \le \eps] \le C \eps.
\end{equation}
\end{lem}
\paragraph*{Proof of Lemma \ref{lem_propdist}.} %
By Assumption \ref{assp:normalcontext}, $\bx_i' \btheta_{g(i)} 
- \bx_j' \btheta_{g(j)}$ is a normal distribution with a constant variance, and thus for a sufficiently small $\eps>0$, there exists $C_2$ such that 
\[
\Prob[
|\bx_i' \btheta_{g(i)} 
- \bx_j' \btheta_{g(j)}| \le \eps
] \le C_2 \eps.
\]
Event 
$\bigcup_{i \ne j} \{
|\bx_i' \btheta_{g(i)} 
- \bx_j' \btheta_{g(j)}| \le \eps\}$
is a necessary condition for
$\{\Delta(n) \le \eps\}$, and thus Eq.~\eqref{ineq_propdist} holds with $C = C_2 K^2$ by using a union bound over $i,j$, which completes the proof. \qed

\begin{lem}\label{lem_propdist_cumulative}
Suppose Assumptions~\ref{assp:twogroups} and Assumption \ref{assp:normalcontext}.
With probability at least $1 - e^{-N \eps/8}$, we have
\[
\sum_{n \le N} \Ind[
\Delta(n) \le \eps
] \le \frac{C}{2} N \eps,
\]
where $C>0$ is the same constant as Lemma \ref{lem_propdist}.
\end{lem}
\paragraph*{Proof of Lemma \ref{lem_propdist_cumulative}.} %
We use the multiplicative Chernoff bound. Namely, for a sequence of $N$ binary random variables $X_1,X_2,\dots,X_N$ with its mean no less than $\mu$, its summation $S$ satisfies
\begin{equation}
\Prob\left[
S \le (1-\delta) \mu
\right] \le e^{-\delta^2 N \mu / 2}.
\end{equation}
Applying the bound above with $\delta = 1/2$ to the sequence of binary events $\{n: \Ind[\Delta \le \eps]\}$ (each occurs with probability at least $C\eps$) yields the desired bound. \qed

The following lemma states that a sublinear regret decision rule chooses the best candidate for almost all the rounds.

\begin{lem}[Characterization of Sublinear-Regret Decision Rule]\label{lem:correctness}
Suppose Assumption \ref{assp:normalcontext}. Then, under a sublinear-regret decision rule, the following holds:
\[
\sum_{n \le N}
\Prob[\iota(n) \ne i^*(n)] 
= o(N).
\]
\end{lem}
\paragraph*{Proof of Lemma \ref{lem:correctness}.} %
By definition, a sublinear-regret decision rule satisfies
\[
\Ex[\Regret(N)]=O(N^a).
\]
We prove Lemma \ref{lem:correctness} by contradiction. Assume that there exists $C_2, N_0>0$ such that 
\begin{equation}\label{ineq:linearincorrect}
\sum_{n \le N} \Prob[\iota(n) \ne i^*(n)] > C_2 N
\end{equation}
for all $N \ge N_0$. Applying the multiplicative Chernoff bound to the sequence of binary events $\{\iota(n) \ne i^*(n)\}$ yields the fact that, with probability at least $1 - e^{-C_2 N / 8}$, we have
\begin{equation}\label{ineq:misidentifycount}    
\sum_{n \le N} \Ind[\iota(n) \ne i^*(n)] > \frac{C_2}{2} N
\end{equation}
Moreover, Lemma \ref{lem_propdist_cumulative} implies the following: Letting $C_3 = C_2/(2C)$, with probability at least $1 - e^{-C_3 N / 8}$, 
\begin{equation}\label{ineq:lineargap}
\sum_{n \le N} \Ind[\Delta(n) \le C_3] \le \frac{C_2}{4} N.
\end{equation}
Eq.~\eqref{ineq:misidentifycount} and \eqref{ineq:lineargap} imply that, with probability at least $1 - e^{-C_3 N / 8}$, there are at least $C_2 N / 2 - C_2 N/4 = C_2N/4$ rounds where $\regret(n)$ is at least $C_3$, 
which implies that the regret is 
\begin{equation}
\Regret(N) = \sum_{n \le N} \regret(n) \ge C_3 \times \frac{C_2}{4} N = \Omega(N)
\end{equation}
which contradicts the fact that the decision rule has sublinear regret. \qed

\paragraph{Proof of Theorem~\ref{thm: sublinear regret implies equalized odds}.}

Lemma~\ref{lem:correctness} implies $\Prob[\iota(n)=i|i^*(n)\ne i,g], \Prob[\iota(n)\ne i|i^*(n) = i,g]  = o(N)$, $\Prob[\iota(n)=i|i^*(n)=i,g], \Prob[\iota(n)\ne i|i^*(n)\ne i,g] = 1-o(N)$. Accordingly, the decision rule $\iota$ asymptotically aligns with equalized odds. \qed

\paragraph{Proof of Theorem~\ref{thm: sublinear regret implies demographic parity}.}
Let $\eps > 0$ be a constant.
Assume that
\begin{equation}\label{eq: demographic parity proof}
\frac{1}{n}
\sum_{n \le N}
\left| 
\Prob[
\iota(n) = i | g(i) = 1
] 
-
\Prob[
\iota(n) = i | g(i) = 2
] 
\right| \ge \eps
\end{equation}
for a sufficiently large $N$. Under Assumption~\ref{assp:idcontext}, the probability that a worker has the highest skill predictor $q_i$ is independent of the worker's group. Accordingly, \eqref{eq: demographic parity proof} implies that the decision rule hires at least $\eps N/(2K)$ suboptimal workers in the first $N$ rounds, in expectation. A similar discussion to Lemma \ref{lem:correctness} yields the fact that the regret due to choosing these suboptimal candidates is $\Omega(\eps N/K)$, implying that the decision rule fails to have sublinear regret. \qed

\section{Additional Simulation}\label{sec: additional simulation}

\subsection{The Pivot Subsidy Rule vs the Cost-Saving Subsidy Rule}\label{subsec: simulation cost saving}

\begin{figure}[t!]
    \centering
    \begin{minipage}[t]{0.48\textwidth}
        \centering
        \includegraphics[width=\textwidth]{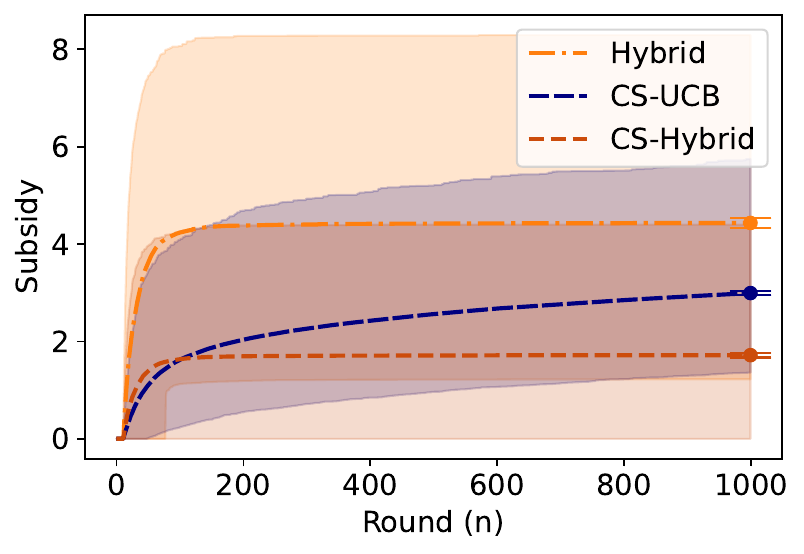}
        \footnotesize
        (a): $N = 1,000$
    \end{minipage}
    \hspace{0.02\textwidth}
    \begin{minipage}[t]{0.48\textwidth}
         \centering
         \includegraphics[width=\textwidth]{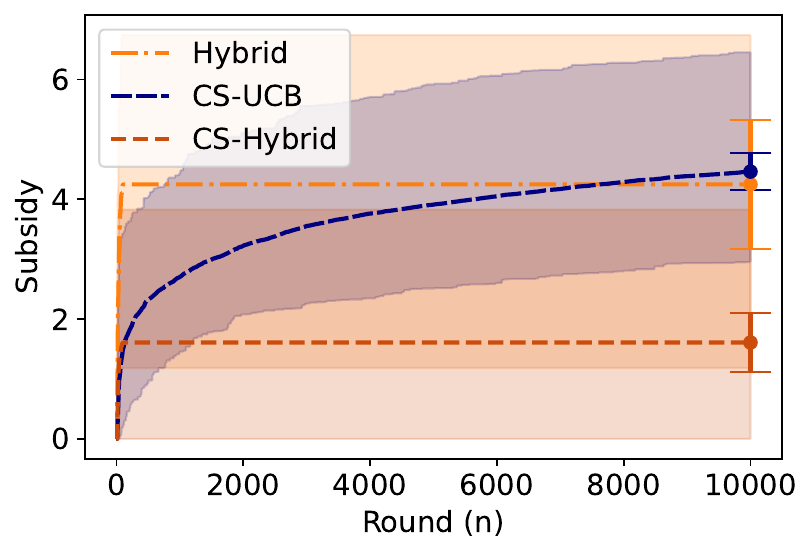}
        \footnotesize
        (b): $N = 10,000$
    \end{minipage}
    \smallskip
    \caption{Budget required by the hybrid index subsidy rule (Hybrid), the UCB cost-saving subsidy rule (CS-UCB), and the hybrid cost-saving subsidy rule (CS-Hybrid).}
    \label{fig: iucb_subsidy_cs_whole}
    \footnotesize
    \raggedright
    \medskip
    
    \textbf{Note:} The lines are averages over sample paths. The areas cover between $5\%$ and $95\%$ percentiles of runs, and the error bars at $N = 1,000$ and $N=10,000$ are the two-sigma confidence intervals. We only run 50 simulations for the case of $N = 10,000$ (Panel b) because its aim is to visualize the long-run subsidy growth under the UCB cost-saving subsidy rule.
\end{figure}

Figure~\ref{fig: iucb_subsidy_cs_whole} compares the subsidy amount associated with the UCB cost-saving subsidy rule and the hybrid subsidy rules. The UCB index subsidy rule is excluded because it requires a much larger subsidy (as shown in Figure~\ref{fig:iucb_subsidy}, the UCB index subsidy rule requires more than $150$ unit of subsidy for $N = 1,000$).

For $N = 1,000$ (Panel a), the hybrid cost-saving subsidy rule achieves the smallest subsidy, followed by the UCB cost-saving subsidy rule and the hybrid index subsidy rule.
We observe that the cost-saving method is very effective.

While the subsidy required by the hybrid rule is proven to be $\tilO(1)$ (this is immediate from Theorems~\ref{thm:ucbimp} and \ref{cor:ucb cost-saving subsidy scheme}), there is no such guarantee for the UCB cost-saving subsidy rule. We conjecture that the subsidy required by the UCB cost-saving rule is $\tilOmega(\sqrt{N})$. 

Panel (b) supports this conjecture by showing regret under a longer time horizon ($N = 10,000$). While the subsidy required by the hybrid subsidy rules remain constant after a few (about 100) rounds, the subsidy required by the UCB cost-saving rule grows gradually. Consequently, (i) the hybrid cost-saving subsidy rule substantially outperforms the UCB cost-saving subsidy rule at $N = 10,000$, and (ii) the hybrid index subsidy rule overtakes the UCB cost-saving subsidy rule.

\subsection{The Hybrid Mechanism vs Uniform Sampling}\label{subsec: hybrid vs uniform sampling}

\citet{kannan2018} show that with sufficiently large initial samples (i.e., $\Nzero$ is large), the greedy algorithm (corresponding to laissez-faire in this paper) has sublinear regret.\footnote{We also note that the number of initial samples required by the relevant theorem ($n_{\mathrm{min}}$ of Lemma 4.3) is very large and cannot be satisfied in our simulation setting: Letting $R = \sigma_x \sqrt{2\log(N)}$, we have $n_{\mathrm{min}} \ge 320 R^2 \log(R^2 d K/\delta) / \lambda_0 \ge 10^3$.} Our analysis also indicates that the probability of perpetual underestimation is small when $\Nzero$ is large (see Lemma \ref{lem:lflarge_smltheta2} for full details).

This ``warm-start'' version of laissez-faire might be presumed efficient. However, the warm-start approach carries several disadvantages. First, although we have thus far ignored the cost of acquiring initial samples for analytical tractability, we need to consider this cost if we want to \emph{take} a sufficiently long warm-start period. Since uniform sampling ignores firms' incentives for hiring workers, practical implementation of it requires a large budget. Second, uniform sampling does not maximize any index. This precludes its implementation by any index policy.
Third, uniform sampling is inefficient in terms of information acquisition because it is not adaptive to currently estimated parameters.

We argue that our hybrid mechanism (Section \ref{sec: Hybrid Decision Rule}) is more efficient than laissez-faire with a warm start because it initially samples the data adaptively before switching to laissez-faire at an efficient time. Hence, we can naturally expect the hybrid mechanism outperforms laissez-faire with initial uniform sampling.

\begin{figure}[t!]
    \centering
    \begin{minipage}[t]{0.48\textwidth}
         \centering
         \includegraphics[width=\textwidth]{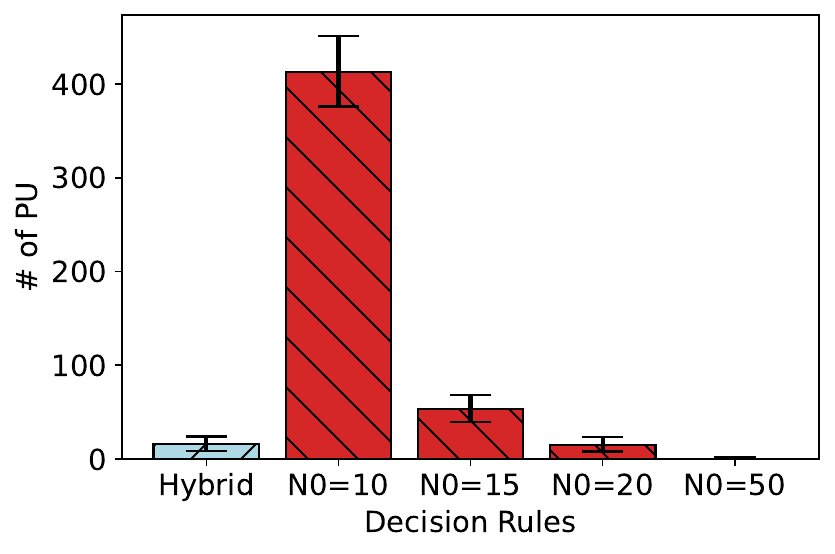}
        \footnotesize
        (a): Frequency of perpetual underestimation
    \end{minipage}
    \hspace{0.02\textwidth}
    \begin{minipage}[t]{0.48\textwidth}
         \centering
         \includegraphics[width=\textwidth]{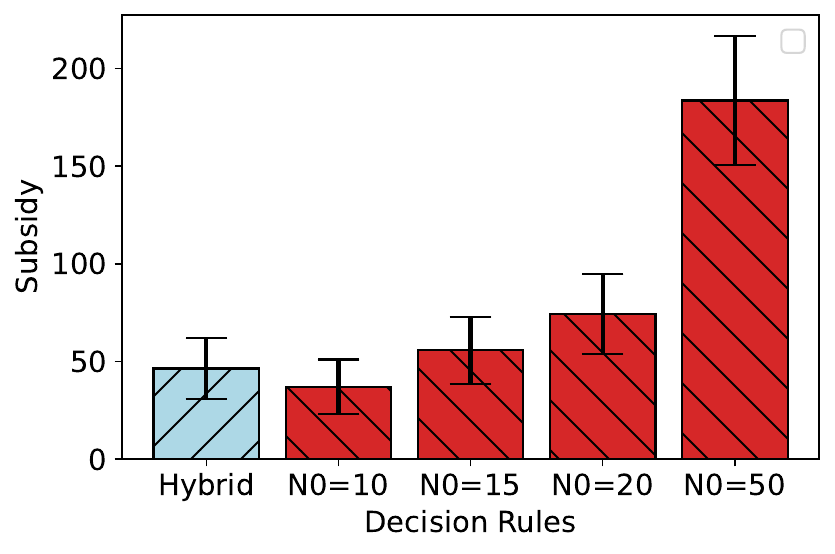}
        \footnotesize
        (b): Budget
    \end{minipage}
    \smallskip
    \caption{The comparison between the hybrid mechanism and laissez-faire with various lengths of initial sampling.}
    \label{fig: initial sampling}
    \raggedright
    \footnotesize
    \medskip
    
    \textbf{Note:} Across $\Numrun$ runs. The error bars represent the two-sigma binomial confidence intervals.
\end{figure}

Figure~\ref{fig: initial sampling} exhibits the simulation results comparing the hybrid mechanism with laissez-faire for various initial samples. In this simulation, the number of initial samples for each group is proportional to the population ratio; i.e., $\Nzero_g = (K_g/K) \cdot \Nzero$.

Panel (a) measures the frequency of perpetual underestimations. As our theory indicated, the larger the initial sample, the less frequently perpetual underestimation occurs. Additionally, we observed no perpetual underestimation for the hybrid mechanism, because it solidly incentivizes hiring candidates from an underexplored group.

Panel (b) depicts the subsidy amount required by the cost-saving subsidy rules. Here, we can observe that the hybrid cost-saving subsidy rule outperforms laissez-faire with uniform sampling. Laissez-faire requires at least $\Nzero \ge 50$ samples to mitigate perpetual underestimation, which requires a larger budget than the hybrid mechanism.

\subsection{Asymmetric Groups}\label{subsec: simulation on asymmetric models}

This section will illustrate how social learning behaves under an asymmetric environment, where the differences between the majority and minority groups extend beyond their group size. We employ $d=5$ and $\bmu_{x, g} = \bar{\mu}_{x, g} (1, 1, 1, 1, 1)$ for $\bar{\mu}_{x,g} \in \Real_{++}$.
Specifically, we fix the average characteristic of the majority group to $\bar{\mu}_{x, 1} = \Simmux$, and examine the varying average characteristic $\bar{\mu}_{x, 2}$ of the minority group.
The other model parameters are set based on the values provided in Section~\ref{sec: simulation}.

Even though the same value of $\btheta$ is applied, the agents are not aware of this beforehand. This means that the distinction between the groups, in terms of which group is superior on average, is only learned through the data that is gathered over time. In this particular setup, the mean of the skill predictor, $q_i$, differs across the groups, however, their variances are identical. This is because, under the assumption that $\sigma_x = 1$, the variance of $q_i$ is equal to $5$. This model allows us to examine how differences in the average skills between groups, as well as the learning process and the effects of any inherent biases, impact decision-making in hiring processes.

\begin{figure}[t!]
    \centering
    \begin{minipage}[t]{0.48\textwidth}
        \centering
         \includegraphics[width=\textwidth]{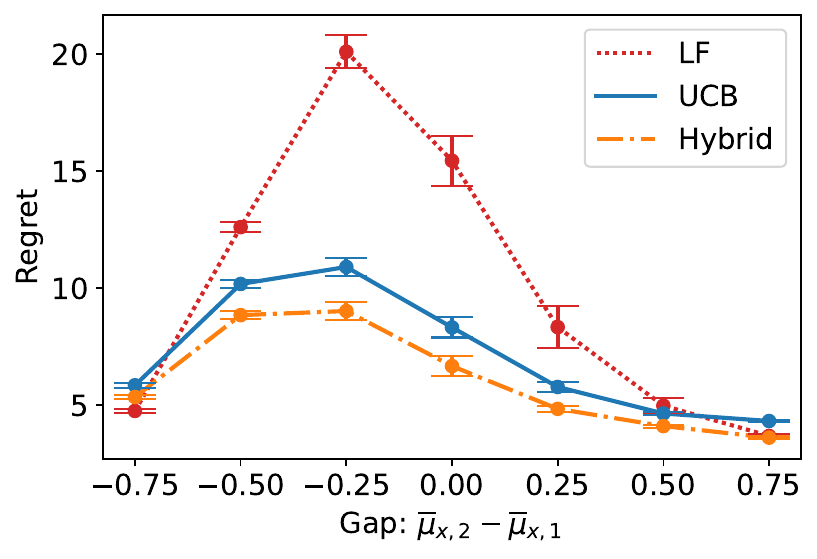}
         \caption{Average skill disparity and regret at the last round $N = 1,000$.}
         \label{fig:mux_regret}
    \end{minipage}
    \hspace{0.02\textwidth}
    \begin{minipage}[t]{0.48\textwidth}
         \centering
         \includegraphics[width=\textwidth]{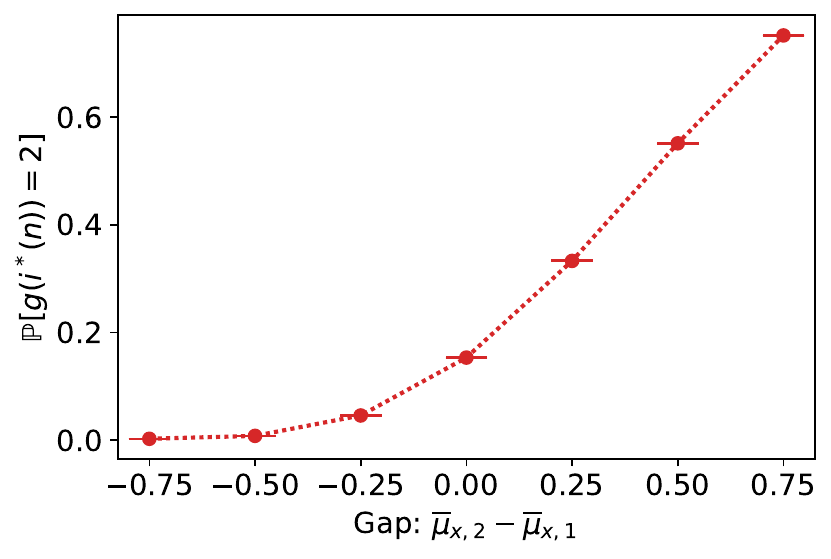}
         \caption{Average skill disparity and the proportion of the best candidate from group $2$.}
         \label{fig:mux_groupd2best}
    \end{minipage}
    \raggedright
    \footnotesize
    \medskip
    
    \textbf{Note:} Across $\Numrun$ runs. The error bars represent the two-sigma binomial confidence intervals.
\end{figure}

The relationship between the gap $\bar{\mu}_{x,2} - \bar{\mu}_{x, 1}$ and regret is demonstrated in Figure~\ref{fig:mux_regret}. Considering that $\bx$ is a five-dimensional variable and $\btheta = \Simtheta$, the skill predictor $q_i = \bx_i' \btheta$ adheres to a normal distribution 
$\Normal(5\bar{\mu}_{x, g}, 5)$. For example, when the gap is $-0.50$, an average worker from group $1$ (whose skill is $1.5 \times 5 = 7.5$) slightly outperforms a one-sigma (approximately top-$15$\%) worker from group $2$ (whose skill is $(1.5-0.5)\times 5 + \sqrt{5} \approx 7.23$).
In Figure~\ref{fig:mux_groupd2best}, we show the proportion of the most skilled candidate (the one with the highest $q_i$) from group $2$.

In cases where the disparity in average skills between groups is minimal, a laissez-faire approach tends to perform poorly, incurring much larger regret than the UCB and hybrid decision rules. The reason for this is similar to the situation with symmetric groups. When the level of asymmetry is low, a laissez-faire approach frequently leads to perpetual underestimation, which results in societal losses.

However, when the average skill disparity is extreme, laissez-faire achieves smaller regret. If the minority group has a significantly higher average skill level (i.e., a large positive gap $\bar{\mu}_{x, 2} - \bar{\mu}_{x, 1} \ge 0.75$), the distribution of top performers becomes balanced across groups. As a result, both groups have a substantial probability of being hired, reducing the likelihood of perpetual underestimation. Laissez-faire, which saves costs of exploration, can result in smaller regret under such circumstances. 

On the other hand, if the minority group has a much lower average skill level (i.e., a large negative gap $\bar{\mu}_{x, 2} - \bar{\mu}_{x, 1} \le -0.75$), almost all best candidates are from the majority group. Accordingly, perpetual underestimation is not costly in terms of regret. Consequently, laissez-faire may achieve smaller regret. However, it is important to note that this result could be perceived as highly unfair, especially from the perspective of highly-skilled minority workers. Additionally, the relative performance of laissez-faire tends to be worse with longer time horizons, as the social cost of consistently missing skilled minority workers starts to accumulate.

Theorem~\ref{thm:ucb} guarantees that the UCB decision rule has $\tilO(\sqrt{N})$ regret regardless of the value of $\bar{\mu}_{x,2} - \bar{\mu}_{x,1}$, showing that UCB is a reliable approach for maintaining low regret.
In contrast, Theorem~\ref{thm:ucbimp} requires $\bar{\mu}_{x,2} - \bar{\mu}_{x,1} = 0$ as a premise. Despite the absence of theoretical guarantee for asymmetric cases, our simulations suggest that the hybrid decision rule consistently outperforms the UCB rule across a wide range of skill disparities $\bar{\mu}_{x,2} - \bar{\mu}_{x, 1}$. We conjecture that the hybrid rule successfully leverages the benefits of both laissez-faire and UCB: Laissez-faire performs well in highly asymmetric environments, whereas UCB performs well in nearly symmetric environments. However, proving the effectiveness of the hybrid mechanism under a broader range of environments theoretically remains an open and intriguing research question.

In the following, we present the detailed behavior of perpetual underestimation, regret, and subsidy in an asymmetric environment. We posit that $\bar{\mu}_{x,1} = 1.5$ and $\bar{\mu}_{x, 2} = 1.0$. That is, $q_i \sim \Normal(7.5, 5)$ for a group-$1$ worker, whereas $q_i \sim \Normal(5, 5)$ for a group-$2$ worker. Consequently, the probability of a random group-$1$ worker outperforming a random group-$2$ worker is 78.57\%.

\paragraph{The Effects of Population Ratio}

\begin{figure}[t!]
    \centering
    \begin{minipage}[t]{0.48\textwidth}
        \centering
         \includegraphics[width=\textwidth]{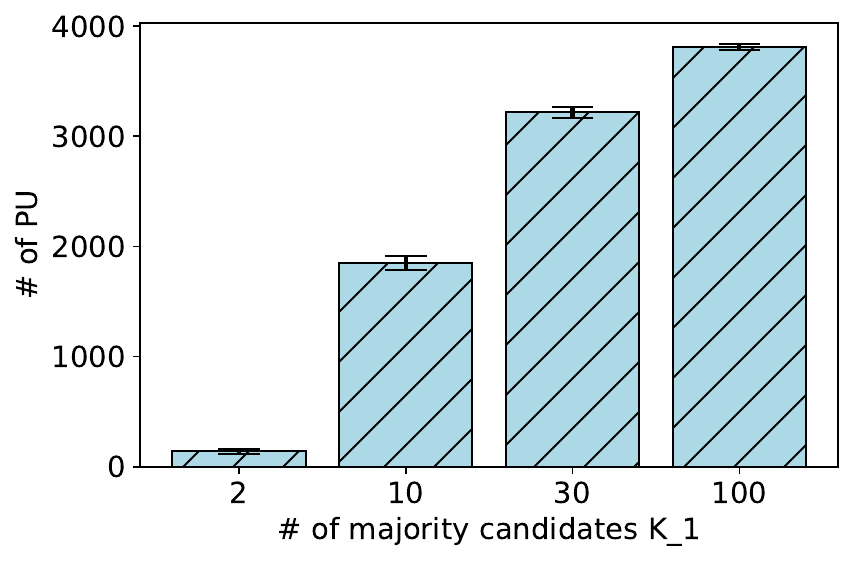}
         \caption{Frequency of perpetual underestimation under laissez-faire (Appendix~\ref{subsec: simulation on asymmetric models}).}
         \label{fig:groupsize_pu DAC}
    \end{minipage}
    \hspace{0.02\textwidth}
    \begin{minipage}[t]{0.48\textwidth}
         \centering
         \includegraphics[width=\textwidth]{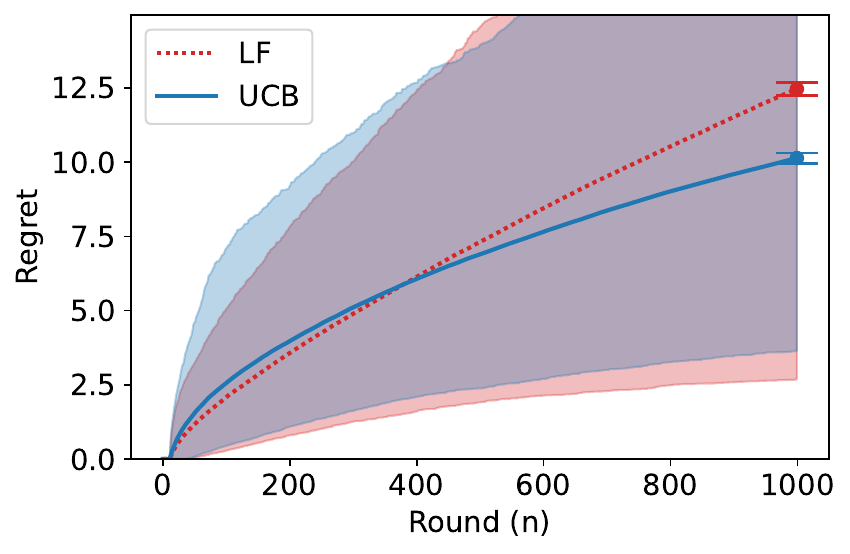}
         \caption{Regret under the LF and UCB decision rules (Appendix~\ref{subsec: simulation on asymmetric models}).}
         \label{fig:policycomp_regret DAC}
    \end{minipage}
    \raggedright
    \footnotesize
    \medskip
    
    \textbf{Left:} Across $\Numrun$ runs. The error bars represent the two-sigma binomial confidence intervals.
    
    \textbf{Right:} The lines are averages over sample paths, the areas cover between $5\%$ and $95\%$ percentiles of runs, and the error bars at the last round $N = 1,000$ are the two-sigma confidence intervals.
\end{figure}

Figure~\ref{fig:groupsize_pu DAC} demonstrates the frequency of perpetual underestimation under laissez-faire, in parallel with Figure~\ref{fig:groupsize_pu}. Perpetual underestimation occurs more frequently in this asymmetric environment because we have assumed that the skills of minority workers are genuinely likely to be lower. However, since the game continues until round $N = 1,000$, it is still exceedingly rare for there to be no round in which a minority worker is the most skilled.

\paragraph{Laissez-Faire vs the UCB Mechanism}

Figure~\ref{fig:policycomp_regret DAC} contrasts regret under laissez-faire and UCB decision rules, corresponding to Figure~\ref{fig:policycomp_regret}. As the regret order of the UCB decision rule (demonstrated in Theorem~\ref{thm:ucb}) is not reliant on symmetry between the groups, UCB effectively converges per-round regret to zero, even in this asymmetric environment. This highlights the robustness of the UCB decision rule against underlying differences in groups, ensuring an efficient and fair hiring process irrespective of inherent group disparities.

In the long run, even if the average skills of minority workers are lower, the cost of missing out on skillful minority workers accumulates, leading to substantial welfare loss under the laissez-faire policy. Despite the initial closeness in the performance of the UCB and laissez-faire policies due to the reduced welfare cost of perpetual underestimation, the UCB policy proves to be significantly more efficient over time. This finding underscores the importance of maintaining an inclusive hiring process that values individual merit over group averages, as this approach not only promotes fairness but also increases overall welfare by harnessing the skills of a broader talent pool.

\paragraph{The UCB Mechanism vs the Hybrid Mechanism}

\begin{figure}[t!]
    \centering
    \begin{minipage}[t]{0.48\textwidth}
         \centering
         \includegraphics[width=\textwidth]{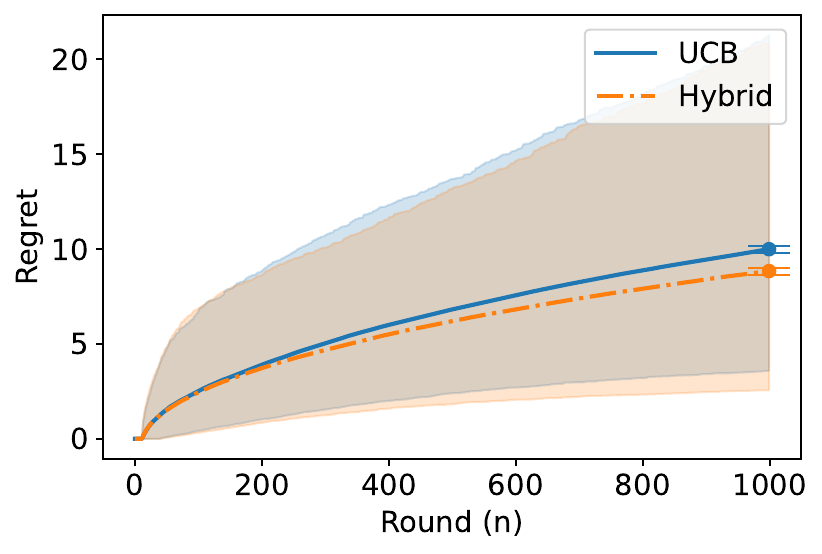}
         \caption{Regret under the UCB and hybrid decision rules (Appendix~\ref{subsec: simulation on asymmetric models}).}
         \label{fig:iucb_regret DAC}
    \end{minipage}
    \hspace{0.02\textwidth}
    \begin{minipage}[t]{0.48\textwidth}
         \centering
         \includegraphics[width=\textwidth]{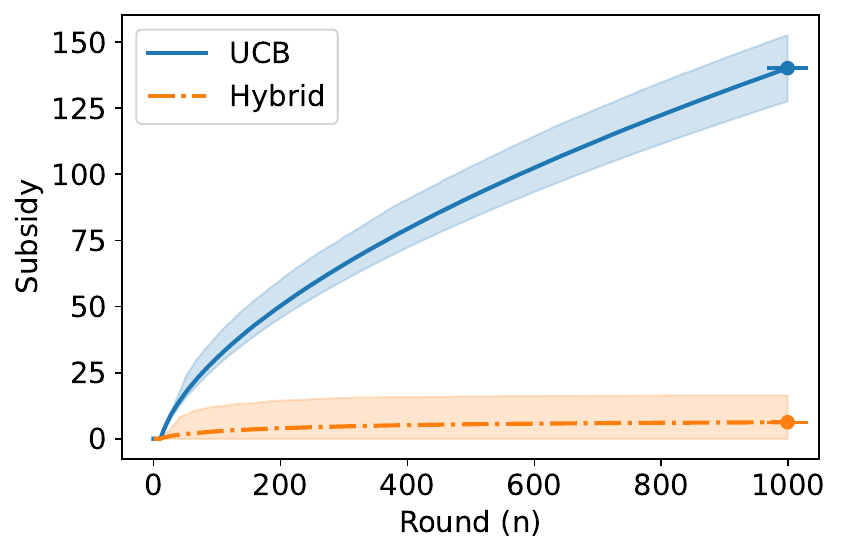}
         \caption{Budget required by the UCB and hybrid index subsidy rules (Appendix~\ref{subsec: simulation on asymmetric models}).}
         \label{fig:iucb_subsidy DAC}
    \end{minipage}
    \raggedright
    \footnotesize
    \medskip
    
    \textbf{Note:} The lines are averages over sample paths, the areas cover between $5\%$ and $95\%$ percentiles of runs, and the error bars at $N = 1,000$ are the two-sigma confidence intervals.
\end{figure}

Figures~\ref{fig:iucb_regret DAC} and \ref{fig:iucb_subsidy DAC} compare the regret and subsidy of the UCB and hybrid mechanisms (corresponding to Figures~\ref{fig:iucb_regret} and \ref{fig:iucb_subsidy}). We observe no quantitative difference from the symmetric case: The hybrid mechanism achieves small regret with a small budget, as with the main simulation setting. This suggests that the assumptions necessary for the hybrid mechanism to attain equalized odds (as outlined in Assumption~\ref{assp:idcontext}) might be less stringent than formally proved. Essentially, even with asymmetric groups, the hybrid mechanism is still effective in ensuring fairness (via equalized odds) and efficiency while minimizing costs.

\end{document}